\documentclass[twocolumn]{svjour3}          
\smartqed

\usepackage{graphicx}
\usepackage{mathptmx}

\usepackage[T1]{fontenc}
\usepackage[utf8]{inputenc}
\usepackage{amsfonts,amsmath,amssymb}
\usepackage{latexsym,multirow,xspace}

\usepackage[boxed,linesnumbered]{algorithm2e}

\usepackage[table]{xcolor}
\usepackage{tabularx}

\newcommand{\pf}{\ensuremath{PF}\xspace}
\newcommand{\apf}{\ensuremath{APF}\xspace}

\newcommand{\fsync}{{\sc FSync}\xspace}
\newcommand{\ssync}{{\sc SSync}\xspace}
\newcommand{\async}{{\sc Async}\xspace}

\newcommand{\angolo}{\sphericalangle}
\newcommand{\argmin}{\arg\!\min}

\newcommand{\mathsc}[1]{\text{\textsc{#1}}}
\newcommand{\R}{ \mathbb{R} } 
\newcommand{\Ex}{\mathbb{E}}  
\newcommand{\Int}{\mathit{int}}  
\newcommand{\minview}{\textit{min\_view}}
\newcommand{\h}{\mathit{H}}

\newcommand{\halfline}{\mathit{hline}}
\newcommand{\Line}{\mathit{line}}

\newcommand{\FORM}{\mathit{FORM}}
\newcommand{\EMB}{\mathit{EMB}}
\newcommand{\FOR}{\mathit{FOR}}
\newcommand{\FIN}{\mathit{FIN}}
\newcommand{\GAT}{\mathit{GAT}}
\newcommand{\INV}{\mathit{INV}}

\newcommand{\SB}{\mathtt{SymBreak}}
\newcommand{\RS}{\mathtt{RefSys}}
\newcommand{\PForm}{\mathtt{ParForm}}
\newcommand{\Fin}{\mathtt{Fin}}
\newcommand{\Member}{\mathtt{Membership}}

\newcommand{\bloccobis}[1]{\smallskip\noindent\textbf{#1}:}

\newcommand{\DistMin}{\mathsc{CollisionFreeMove}}  
\newcommand{\SMove}{\mathsc{StationaryMove}\xspace}

\newcommand{\bzero}{\mathtt{b}_0}
\newcommand{\buno}{\mathtt{b}_1}
\newcommand{\bdue}{\mathtt{b}_2}
\newcommand{\btre}{\mathtt{b}_3}
\newcommand{\bquattro}{\mathtt{b}_4}
\newcommand{\bcinque}{\mathtt{b}_5}
\newcommand{\e}{\mathtt{e}}
\newcommand{\zuno}{\mathtt{z}_1}
\newcommand{\zdue}{\mathtt{z}_2}
\newcommand{\uuno}{\mathtt{u}_1}
\newcommand{\udue}{\mathtt{u}_2}

\newcommand{\fdue}{\mathtt{f}_2}
\newcommand{\ftre}{\mathtt{f}_3}
\newcommand{\fquattro}{\mathtt{f}_4}
\newcommand{\fcinque}{\mathtt{f}_5}

\renewcommand{\a}{\mathtt{a}}
\renewcommand{\c}{\mathtt{c}}

\newcommand{\dzero}{\mathtt{d}_0}
\newcommand{\duno}{\mathtt{d}_1}
\newcommand{\ddue}{\mathtt{d}_2}

\newcommand{\iuno}{\mathtt{i}_1}
\newcommand{\idue}{\mathtt{i}_2}
\newcommand{\itre}{\mathtt{i}_3}
\newcommand{\iquattro}{\mathtt{i}_4}
\newcommand{\icinque}{\mathtt{i}_5}
\newcommand{\isei}{\mathtt{i}_6}

\newcommand{\gzero}{\mathtt{g}_0}
\newcommand{\guno}{\mathtt{g}_1}
\newcommand{\gdue}{\mathtt{g}_2}

\renewcommand{\l}{\mathtt{l}}
\newcommand{\mzero}{\mathtt{m}_0}
\newcommand{\muno}{\mathtt{m}_1}

\newcommand{\q}{\mathtt{q}}

\newcommand{\sdue}{\mathtt{s}_2}
\newcommand{\stre}{\mathtt{s}_3}
\newcommand{\spiu}{\mathtt{s}_+}

\newcommand{\tzero}{\mathtt{t}_0}
\newcommand{\tuno}{\mathtt{t}_1}
\newcommand{\w}{\mathtt{w}}

\newcommand{\robotuno}{r^{(1)}}  
\newcommand{\robotdue}{r^{(2)}}

\newcommand{\Fcinque}{\mathcal{F}5} 
\newcommand{\FcinqueS}{ \fcinque \wedge \neg\w} 
\newcommand{\FcinqueE}{ \Funo_s \vee \Fdue_s \vee \Fquattro_s \vee \w}

\newcommand{\T}{\mathcal{T}} 
\newcommand{\TS}{ \bzero \wedge \neg\c \wedge \buno \wedge \neg \zuno} 
\newcommand{\TE}{ \U_s \vee \W_s}

\newcommand{\U}{\mathcal{U}} 
\newcommand{\US}{ \zuno  } 
\newcommand{\UE}{ } 

\newcommand{\Vuno}{\mathcal{V}1} 
\newcommand{\VunoS}{ \bzero \wedge (\neg\c \vee \neg\uuno) \wedge \bdue}
\newcommand{\VunoE}{\W_s}

\newcommand{\Vdue}{\mathcal{V}2} 
\newcommand{\VdueS}{ \bzero \wedge (\neg\c \vee \neg\udue) \wedge \btre}
\newcommand{\VdueE}{\W_s}

\newcommand{\Vtre}{\mathcal{V}3} 
\newcommand{\VtreS}{ \bzero \wedge \neg\c  \wedge \bquattro}
\newcommand{\VtreE}{\W_s}

\newcommand{\Vquattro}{\mathcal{V}4} 
\newcommand{\VquattroS}{ \bzero \wedge \neg\c  \wedge \bcinque}
\newcommand{\VquattroE}{ \Vdue_s }

\newcommand{\W}{\mathcal{W}} 
\newcommand{\WS}{ \zdue \wedge \neg \Vuno_s \wedge \neg \Vdue_s}
\newcommand{\WE}{ }

\newcommand{\Funo}{\mathcal{F}1}  
\newcommand{\FunoS}{\neg \fdue \wedge \neg \ftre \wedge \neg \fquattro \wedge \neg \fcinque \wedge \neg\w} 
\newcommand{\FunoE}{ \Fdue_s \vee \Ftre_s \vee \Fquattro_s \vee \Fcinque_s \vee \w}

\newcommand{\A}{\mathcal{A}} 
\newcommand{\Auno}{\mathcal{A}1} 
\newcommand{\AunoS}{\spiu \wedge \l}

\newcommand{\AunoE}{\A_s \vee \C_s \vee \D_s \vee \Euno_s}
\newcommand{\Adue}{\mathcal{A}2} 
\newcommand{\AdueS}{ (\spiu \vee \stre) \wedge \neg\l }

\newcommand{\AdueE}{\A_s \vee \C_s \vee \D_s \vee \Euno_s}

\newcommand{\B}{\mathcal{B}} 
\newcommand{\BS}{\sdue \wedge \mzero \wedge \neg\muno }

\newcommand{\BE}{\C_s \vee \D_s}

\newcommand{\C}{\mathcal{C}} 
\newcommand{\Cuno}{\mathcal{C}1} 
\newcommand{\CunoS}{\stre \wedge \tzero \wedge \l}

\newcommand{\CunoE}{\Cdue_s \vee \D_s \vee \Euno_s}
\newcommand{\Cdue}{\mathcal{C}2} 
\newcommand{\CdueS}{ \stre \wedge \neg\tzero \wedge \neg\tuno  \wedge \l \wedge \mzero \wedge \neg\muno}

\newcommand{\CdueE}{\D_s}

\newcommand{\D}{\mathcal{D}} 
\newcommand{\DS}{\stre \wedge \tuno \wedge \l \wedge \mzero \wedge \neg\muno}

\newcommand{\DE}{\Euno_s}

\newcommand{\E}{\mathcal{E}} 
\newcommand{\Euno}{\mathcal{E}1} 
\newcommand{\EunoS}{\stre \wedge \neg \tzero \wedge \l \wedge ( \mzero \Rightarrow \muno )}

\newcommand{\Edue}{\mathcal{E}2} 
\newcommand{\EdueS}{ \sdue \wedge \neg\l \wedge ( \mzero \Rightarrow \muno ) }

\newcommand{\Fdue}{\mathcal{F}2} 
\newcommand{\FdueS}{\fdue \wedge \neg \ftre \wedge \neg \fquattro \wedge \neg \fcinque \wedge \neg\w} 

\newcommand{\FdueE}{ \Ftre_s \vee \Fquattro_s \vee \w}

\newcommand{\G}{\mathcal{G}} 
\newcommand{\Guno}{\mathcal{G}1} 
\newcommand{\GunoS}{\neg\gzero \wedge  (\c \Rightarrow \muno)  }

\newcommand{\GunoE}{\Gdue_s}
\newcommand{\Gdue}{\mathcal{G}2} 
\newcommand{\GdueS}{\gzero \wedge \neg\guno \wedge (\c \Rightarrow \muno)}

\renewcommand{\H}{\mathcal{H}} 
\newcommand{\HS}{\c \wedge \neg\mzero }

\newcommand{\HE}{\Gdue_s}

\newcommand{\Ftre}{\mathcal{F}3} 
\newcommand{\FtreS}{ \ftre \wedge \neg \fquattro \wedge \neg \fcinque \wedge \neg\w} 
\newcommand{\FtreE}{ \Fquattro_s}

\newcommand{\M}{\mathcal{M}} 
\newcommand{\MS}{\dzero \wedge \neg\duno }

\newcommand{\ME}{\O_s}

\newcommand{\N}{\mathcal{N}} 
\newcommand{\NS}{(\dzero \Rightarrow \duno) \wedge  \ddue }

\newcommand{\NE}{\M_s \vee \O_s}

\renewcommand{\O}{\mathcal{O}} 
\newcommand{\OS}{(\dzero \Rightarrow \duno) \wedge \neg \ddue}

\renewcommand{\OE}{\N_s \vee \O_s}

\newcommand{\Fquattro}{\mathcal{F}4} 
\newcommand{\FquattroS}{\fquattro \wedge \neg \fcinque \wedge \neg\w} 
\newcommand{\FquattroE}{ \w }

\renewcommand{\P}{\mathcal{P}} 
\newcommand{\Puno}{\mathcal{P}1} 
\newcommand{\PunoS}{\neg \q \wedge \iuno}

\newcommand{\PunoE}{\Pdue_s}
\newcommand{\Pdue}{\mathcal{P}2} 
\newcommand{\PdueS}{\neg \q \wedge \idue}

\newcommand{\Q}{\mathcal{Q}} 
\newcommand{\Quno}{\mathcal{Q}1} 
\newcommand{\QunoS}{\q \wedge \itre }

\newcommand{\QunoE}{\Qdue_s \vee \Qtre_s}
\newcommand{\Qdue}{\mathcal{Q}2} 
\newcommand{\QdueS}{\q \wedge \iquattro}

\newcommand{\QdueE}{\Qtre_s}
\newcommand{\Qtre}{\mathcal{Q}3} 
\newcommand{\QtreS}{\q \wedge \icinque}

\newcommand{\QtreE}{\Qquattro_s}
\newcommand{\Qquattro}{\mathcal{Q}4} 
\newcommand{\QquattroS}{\q \wedge \isei}

\begin{document}
\title{Asynchronous Arbitrary Pattern Formation: the effects of a rigorous approach%
\thanks{The work has been supported in part by the European project  
        ``Geospatial based Environment for Optimisation Systems Addressing Fire
        Emergencies'' (GEO-SAFE), contract no. H2020-691161, and by the Italian 
        National Group for Scientific Computation (GNCS-INdAM).}
}

\author{Serafino Cicerone \and
        Gabriele Di Stefano \and
        Alfredo Navarra
}

\authorrunning{S. Cicerone, G. Di Stefano, A. Navarra} 

\institute{S. Cicerone \at
           Department of Information Engineering, Computer Science and Mathematics -   
           University of L'Aquila, Via Vetoio, I--67100, L'Aquila, Italy.
           \email{serafino.cicerone@univaq.it}
           \and
           G. Di Stefano \at
           Department of Information Engineering, Computer Science and Mathematics -   
           University of L'Aquila, Via Vetoio, I--67100, L'Aquila, Italy.
           \email{gabriele.distefano@univaq.it}
           \and
           A. Navarra \at
           Department of Mathematics and Computer Science, University of Perugia, 
           Via Vanvitelli 1, I--06123, Perugia, Italy.\\
           \email{alfredo.navarra@unipg.it}           
}

\date{Received: date / Accepted: date}

\maketitle

\begin{abstract}

Given any multiset $F$ of points in the Euclidean plane and a set $R$ of robots such that $|R|=|F|$, the Arbitrary Pattern Formation (\apf) problem asks for a distributed algorithm that moves robots so as to reach a configuration similar to $F$. Similarity means that robots must be disposed as $F$ regardless of translations, rotations, reflections, uniform scalings. 
Initially, each robot occupies a distinct position. When active, a robot operates in standard Look-Compute-Move cycles. Robots are asynchronous, oblivious, anonymous, silent and execute the same distributed algorithm. 
So far, the problem has been mainly addressed by assuming chirality, that is robots share a common left-right orientation. We are interested in removing such a restriction.

While working on the subject, we faced several issues that required close attention. We deeply investigated how such difficulties were overcome in the literature, revealing that crucial arguments for the correctness proof of the existing algorithms have been neglected. The systematic lack of rigorous arguments with respect to necessary conditions required for providing correctness proofs deeply affects the validity as well as the relevance of strategies proposed in the literature. 

Here we design a new deterministic distributed algorithm that fully characterizes \apf\ showing its equivalence with the well-known \emph{Leader Election} problem in the asynchronous model without chirality.
Our approach is characterized by the use of logical predicates in order to formally describe our algorithm as well as its correctness. In addition to the relevance of our achievements, our techniques might help in revising previous results. In fact, it comes out that well-established results like [Fujinaga et al., \emph{SIAM J. Comp.} 44(3) 2015], more recent approaches like [Bramas et al., \emph{PODC} and \emph{SSS} 2016] and ‘unofficial' results like [Dieudonn{\'{e}} et al., \emph{arXiv:0902.2851}] revealed to be not correct.  
\end{abstract}

\section{Introduction}
In distributed computing, one of the most studied problem is certainly the \emph{Pattern Formation} (\pf) which is strictly related to \emph{Consensus} and \emph{Leader Election}. Given a team of robots (agents or entities) and a geometric pattern in terms of points in the plain with respect to an ideal coordinate system, the goal is to design a distributed algorithm that works for each robot to guide it so that eventually all robots together form the pattern if possible. As the global coordinate system might be unknown to the robots, a pattern is declared formed as soon as robots are disposed similarly to the input pattern, that is regardless of translations, rotations, reflections, uniform scalings. 

The \pf\ problem has been largely investigated in the last years under different assumptions. 
Here we refer to the standard \emph{Look}-\emph{Compute}-\emph{Move} model.
When active, a robot operates in Look-Compute-Move (LCM) cycles.
In one cycle a robot takes a snapshot of the current global configuration 
(Look) in terms of robots' positions according 
to its own coordinate system. Successively, in the Compute phase it decides 
whether to move toward a specific target or not, and in the positive case 
it moves (Move).

Different characterizations of the environment consider whether robots are 
fully-synchronous, semi-synchronous (cf.~\cite{SY99,YS10,YUKY17}) or asynchronous (cf.~\cite{BCM16,DFSY15,FYOKY15,FPSW08,GM10}):

\begin{itemize}
\item \emph{Fully-synchronous} (\fsync): The \textit{activation} phase (i.e.
the execution of a Look-Compute-Move cycle) of all robots can be logically divided into
global rounds. In each round all the robots are activated, obtain the same
snapshot of the environment, compute and perform their move. 
\item \emph{Semi-synchronous} (\ssync): It coincides with the \fsync model,
with the only difference that not all robots are necessarily activated in each
round.
\item \emph{Asynchronous} (\async): The robots are activated independently,
and the duration of each phase is finite but
unpredictable. As a result, robots do not have a common notion of time.
Moreover, they can be seen while moving, and computations can be made based on
obsolete information about positions.
\end{itemize}

One of the latest and most important results for \pf, see~\cite{FYOKY15}, solves the problem for robots endowed with few capabilities. Initially, no robots occupy the same location, and they are assumed to be:
\begin{itemize}
\item{Dimensionless}: modeled as geometric points in the plane;
\item{Anonymous}: no unique identifiers;
\item{Autonomous}: no centralized control;
\item{Oblivious}: no memory of past events;
\item{Homogeneous}: they all execute the same \emph{deterministic} algorithm;
\item{Silent}: no means of direct communication;
\item{Asynchronous}: there is no global clock that synchronizes their actions;
\item{Non-rigid}: robots are not guaranteed to reach a destination within one move;
\item{Chiral}: they share a common left-right orientation.
\end{itemize}

Since \async robots are considered, it is worth to remark they are activated independently, and the duration of each phase is finite but unpredictable, i.e., the time between Look, Compute, and 
Move phases is finite but unbounded. This is decided by an adversary for each 
robot and for each phase. As a result, robots do not have a common notion of time. Moreover, they can be seen while moving even though a robot does not perceive whether other 
robots are moving or not. It follows that the computations and hence the moves performed by robots can be made on the base of obsolete information about robots' positions as these can refer to outdated perceptions. In fact, due to asynchrony, by the time a robot takes a snapshot of the configuration, 
this might have drastically changed when it starts moving. The adversary determining 
the Look-Compute-Move cycles timing is assumed to be fair, that is, each robot 
becomes active and performs its cycle within finite time and infinitely often. 

During the Look phase, robots can perceive \emph{multiplicities}, that is whether a same point
is occupied by more than one robot, and how many robots compose a multiplicity.

The distance traveled within a move is neither infinite nor infinitesimally small. 
More precisely, the adversary has also the power to stop a moving robot before it 
reaches its destination, but there exists an unknown constant $\nu > 0$ such 
that if the destination point is closer than $\nu$, the robot will reach it, 
otherwise the robot will be closer to it of at least $\nu$. Note that, without 
this assumption, an adversary would make impossible for any robot to ever reach 
its destination.

The main open question left in~\cite{FYOKY15} within \async\ concerns the resolution of the more general \pf\ problem in the described setting but without chirality. 
So far, the only sub-problems solved within the weakest setting are the gathering problem~\cite{CFPS12}, where all robots must move toward a common point, and the circle-formation problem~\cite{FPSV14,MV16}, where $n$ robots must form a regular $n$-gon. 

Another interesting question represents the so-called \emph{Arbitrary Pattern Formation} (\apf), that is, from which initial configurations it is possible to form \emph{any} pattern? In~\cite{DPV10}, \apf\ has been solved for a number $n\geq 4$ of asynchronous robots with chirality but excluding patterns with multiplicities. The answer to this restricted setting for \apf\ provided a nice characterization of the problem that was shown to be equivalent to Leader Election within the same set of assumptions. In particular, the configurations from which the proposed algorithm could output any pattern (without multiplicities) are the so-called \emph{leader configurations}. These are configurations of robots from which it is possible to elect a leader. Clearly, if one provides an algorithm to solve \apf for some input configurations, then also \pf can be solved for all such configurations.

The contribution of this work is threefold: 
(1) to provide counter-examples to the correctness of algorithms for the \pf\ and \apf\ problems proposed in well-established and in recent papers;
(2) to solve \apf\ (and consequently \pf) with any number of asynchronous robots without chirality when the initial configurations are leader configurations, hence generalizing the equivalence of \apf and Leader Election also within our setting for any number of robots without chirality, including patterns with multiplicities;
(3) to use a rigorous approach for handling problems in the asynchronous environment able to provide accurate arguments to state the correctness of the designed algorithms. 

\subsection{Motivation and related work}
Starting with~\cite{CDN16}, we kept on studying \pf\ without chirality but we faced several issues mainly related to asynchrony that required close attention.  Since the main difficulties were not depending on the lack of chirality, we deeply investigated how such problems were overcome in the literature. In particular, we closely explored~\cite{FYOKY15} -- that can be considered as a milestone in the advancement of \pf\ study within \async\ -- and we found that surprisingly our difficulties with \pf\ have not been addressed at all.
Actually, we are able to devise fundamental counter-examples to the correctness of~\cite{FYOKY15}. We contacted the authors of~\cite{FYOKY15} and they confirmed us that the provided counter-example and most importantly the rationale behind it represents a main issue that requires deep investigation. 

In order to catch a first idea of our findings, it is necessary to understand what should be carefully analyzed when dealing with asynchrony and robots that do not share a common coordinate system. The first problem arising when approaching \pf\ is where robots should form the pattern $F$, that is how to embed $F$ on the area occupied by robots so as each point of $F$ can be seen as a ‘landmark'.
Both in~\cite{FYOKY15} and in our strategy, the algorithms are logically divided into phases. Usually the first phase is devoted to move a few of robots in such a way the embedding becomes easy. Then, there is an intermediate phase where all robots but those placed in the first phase are moved in order to form $F$. Finally there is a third phase where the ‘special' robots are moved to finalize $F$. While in the second phase it is relatively easy to move robots to partially form $F$ because the embedding is well defined, this is not the case for the other two phases. For instance, if not carefully managed, it may happen that during a move the configuration changes its membership to a different phase, especially from the first phase. 
In order to provide the correctness proof of an algorithm under the sketched scheme, it is not sufficient to define some invariants that exclusively define the membership of a configuration to a phase (as sometimes has been done in the literature), but it is mandatory to prove that the defined moves cannot change the membership of the current configuration \emph{while robots are moving}. 
In the \async model a change of membership is possible, if not carefully considered, as robots can be seen while moving. If such a situation happens, other robots ‘believing to be in a different phase' may start moving and then the situation becomes sometimes intractable or even they prevent the algorithm to accomplish the \pf. Unfortunately this is the case for~\cite{FYOKY15}, where the authors neglected such situations, and we can provide counter-examples where their algorithm fails. It is worth to remark, that the systematic lack of arguments with respect to such events completely invalidates the correctness and the relevance of the strategy proposed in~\cite{FYOKY15}. Moreover, it is not possible to recover the algorithm with some easy patches as it requires structural intervention. 

On this base, and in order to better understand the problem, we started restricting our attention to the case where initial configurations are asymmetric. An asymmetric configurations of robots is meant as a disposal on the Euclidean plane not admitting axis of symmetries nor rotations. Even if solving this ‘restricted' version of the \pf\ problem may be perceived to be an easy task, 
the analysis of the literature revealed the following state of the art:
\begin{itemize}
\item
A first study can be found in~\cite{GM10}. The authors provide a distributed algorithm that can form many patterns, subject to significant restrictions in the input configurations. One of such constraints almost coincides with asymmetry, but it is not the only one. Unfortunately, the way it is presented does not allow to exactly specify from which kind of configurations robots can start and which patterns are formable. 
\item
A formal characterization that includes asymmetric configurations has been conducted in~\cite{DPV10-2,DPV10}. Actually the problem considered is \apf and the authors show the equivalence of \apf and the Leader Election problems under some circumstances. 
Still the results are based on chirality for a number of robots greater than or equal to 4. Moreover, the patterns considered are not exactly ‘arbitrary' because they do not allow multiplicities, that is the points of any given pattern to be formed by $n$ robots compose a set of $n$ distinct elements. In our study, as well as in~\cite{FYOKY15}, patterns may allow multiplicities, and this deeply affects the design of resolution algorithms. 
\item
Other approaches found in the literature to solve \pf\ are probabilistic, see~\cite{BT15,BT16b,YY14}. In particular, in~\cite{BT15,BT16b} the authors claim to solve \apf\ (and hence \pf) with a strategy that is divided into two main phases: the first phase is probabilistic and is used to make asymmetric the input configuration; the second phase is deterministic and solves \pf\ from asymmetric configurations even without assuming any form of multiplicity detection. Basically, the strategy in the second phase solves \apf from asymmetric configurations.  
Unfortunately, the proposed strategy suffers of some lack of rigorous arguments supporting the correctness of the proposed algorithm. We are able to provide counter-examples that affect the rationale behind the proposed strategy, and then also in this case the proposed algorithm is not correct. An extended version of those papers can be found in~\cite{BT16}.
\item
Further ‘unofficial' results can be found in~\cite{CGJM14,DPVarx09}. Concerning~\cite{CGJM14}, not all patterns are considered but only asymmetric ones. Concerning~\cite{DPVarx09}, the authors claim to solve \apf  by slightly modifying the results of~\cite{DPV10}. There are three main issues about this paper. 
First of all, the main proof is given by a sketchy description, whereas we show how formal arguments are extremely necessary in this context. In fact, we are able to provide counter-examples that invalidate the proposed strategy, hence discovering a further algorithm which is not correct.

Secondly, the patterns considered in~\cite{DPVarx09} do not allow multiplicities. In our study, as well as in~\cite{FYOKY15}, patterns may allow multiplicities, and this deeply affects the design of resolution algorithms. In fact, also the gathering becomes a sub-problem of  \pf\ where it is required to solve the so-called \emph{point formation}. It is well-known how difficult is the resolution of the gathering problem~\cite{CFPS12}. Finally, the minimum number of robots required by the proposed algorithm is $5$. Note that in robot based computing systems, instances with small numbers of robots usually require very different and non-trivial arguments with respect to the general algorithm. This is the case for instance for the square formation~\cite{MV16} and for gathering both in the Euclidean plane~\cite{CFPS12} or in discrete rings~\cite{DDN12tr,DDN11,DMN15,BPT16} and grids~\cite{DDKN12}.
\end{itemize}
Such an analysis along with a useful interaction with anonymous referees, motivated us to study not only asymmetric configurations but also leader configurations, in order to see whether \apf is equivalent to Leader Election also in our more general setting.

Nevertheless, our investigation inspired the third aim of this paper, that is to provide a rigorous approach for designing algorithms in \async. The need of new ways of expressing algorithm in \async is widely recognized. For instance, in~\cite{BLP16,DBO17,MPST14} a formal model to describe mobile robot protocols under synchrony and asynchrony assumptions is provided. So far, these only concern robots operating in a discrete space i.e., with a finite set of possible robot positions. 

\subsection{Our results}
We provide fundamental arguments affecting the correctness of~\cite{FYOKY15},~\cite{BT15,BT16b} and~\cite{DPVarx09}.  
The relevance of our finding is given not only by the fact that we re-open the \pf\ and the \apf\ problems in some \async\ contexts, but also that possibly other tasks may suffer of the same arguments. It follows that problems considered already solved, like for \pf\ in the case of robots with chirality or like \apf\ in the case of robots without chirality, must be carefully revised. We show that the resolution of such problems might be really difficult, especially for providing correctness proofs with an adeguate level of formalism. 

We fully characterize the \apf\ problem (and hence \pf) from initial leader configurations, that is for any number of robots we can form any pattern, including symmetric ones and those with multiplicities. Since we do not assume chirality, symmetries to be considered for the robots and for the patterns are not only rotations as in~\cite{FYOKY15,DPV10-2,DPV10}, but also reflections and symmetries due to multiplicities. 
The proposed algorithm shows the equivalence between \apf and Leader Election within our setting, generalizing previous results in terms of number of robots that now can be any, in terms of robots capability as we removed the chirality assumption, and in terms of formable patterns that now can include multiplicities and symmetries. Our main result can be stated as follows:

\begin{theorem}\label{th:main1}
Let $R$ be an initial configuration of \async robots without chirality. \apf is solvable from $R$ if and only if Leader Election is solvable in $R$.
\end{theorem}

To reach this result, we design our algorithm according to a rigorous approach. Such an approach is based on basic predicates that composed in a Boolean logic way provides all the invariants needed to be checked during the execution of the algorithm. Differently to previous approaches used in the literature, we make a careful use of invariants to describe properties holding during the movements of robots. As already observed, the technique of specifying formal invariants to define the different phases of an algorithm is something that other authors have adopted. However in this paper the level of details reached to describe every single move and the corresponding trajectory is something new.
In turn, this implies that for each single move the algorithm may require three different invariants (to describe properties at the start, during, and at the end of the move). Hence, our algorithm is organized as a set of moves, each associated to up to three invariants. Moves are grouped and associated to a phase, where a phase represents a general task of the algorithm. Summarizing, the approach leads to a greater level of detail that provides us rigorous arguments to state the correctness of the algorithm.  This approach itself represents a result of this paper, as it highlights crucial properties in \async\ contexts that so far have been underestimate in the literature. 

As further remarks, it is worth to note that differently from~\cite{FYOKY15}, we do not require that the local coordinate system specific of a single robot remains the same among different Look-Compute-Move cycles. Moreover, the trajectories traced during a move specified by our algorithm are always well-defined either as straight lines or as rotations along specified circles. 

Finally, our algorithm does not require to specify the pattern to be formed as a set of coordinates in a Cartesian system. There are two possible options. As in~\cite{GM10}, the pattern can be specified as a list of ratios of its sides and angles between the sides. Another option is to provide the list of distances among all points. In both cases, each robot can locally evaluate a possible set of points consistent with the input and its local coordinate system. Clearly, it turns out that doing this way robots do not share the same information about the pattern but each one acquires its own representation.

\subsection{Outline}
In the next section, some further details on the considered robots' model are provided. In Section~\ref{sec:notation}, useful notation and definitions are introduced. Section~\ref{sec:strategy} provides a first description of our strategy to solve \apf in \async without chirality and starting from initial leader configurations. Section~\ref{sec:counter-ex} contains a description of some counter-examples about previous strategies for \pf\ and \apf\ in \async, hence re-opening some of these problems. Section~\ref{sec:algorithm}, provides our new distributed algorithm. It is given in terms of various sub-phases where different moves are performed. The correctness of the algorithm is given in Section~\ref{sec:correctness}. Finally, Section~\ref{sec:conclusion} concludes the paper. 

\section{Robot Model}
The robot model is mainly borrowed from~\cite{FYOKY15} and~\cite{CDN16}. We consider a system composed by a set $\mathcal R$ of $n$ mobile \textit{robots}. Let $\R$ be the set of real numbers, at any time the multiset $R=\{r_1,r_2,\ldots,r_n\}$, with $r_i\in \mathbb{R}^2$,  contains the \textit{positions} of all the robots. 

We arbitrarily fix an $x$-$y$ coordinate system $Z_0$ and call it the \emph{global coordinate system}. A robot, however, does not have access to it: it is used only for the purpose of description, including for specifying the input. All actions taken by a robot are done in terms of its local $x$-$y$ coordinate system, whose origin always indicates its current position. Let $r_i(t)\in \R^2$ be the location of robot $r_i$ (in $Z_0$) at time $t$. Then a multiset $R(t) = \{r_1(t), r_2(t), \ldots, r_n(t)\}$ is called the \emph{configuration} of $R$ at time $t$ (and we simply write $R$ instead of $R(t)$ when we are not interested in any specific time). 

A robot is said to be \emph{stationary} in a configuration $R(t)$ if at time $t$ it is:
\begin{itemize}
\item inactive, or
\item active, and:
\begin{itemize}
	\item it has not taken the snapshot yet;
	\item it has taken snapshot $R(t)$;
	\item it has taken snapshot $R(t')$, $t'<t$, which leads to a null movement.
\end{itemize}
\end{itemize}

\noindent
A configuration $R(t)$ is said to be  \emph{stationary}\footnote{The definition of stationary robot provided in~\cite{FYOKY15} is slightly different but also inaccurate. In fact, it does not catch the third scenario about active robots described by our definition. If removing such a case, no configuration might be declared stationary during an execution. Similarly, the definition of \emph{static} robot from~\cite{BT16b} does not describe for instance the case where a robot is not moving but has already performed the Look phase.} if all robots are stationary in $R(t)$. 

If two or more robots occupy the same position, then there is an element $r\in R$ occurring more than once. In such a case $r$
is said to belong to (or compose) a \emph{multiplicity}.

\begin{definition}
	A configuration $R$ is said \emph{initial} if it is stationary and all elements in $R$ are distinct, that is, no multiplicity occurs.
\end{definition}

Each robot $r_i$ has a local coordinate system $Z_i$, where the origin always points to its current location.  
Let $Z_i(p)$ be the coordinates of a point $p\in \R^2$ in $Z_i$. If $r_i$ takes a time interval $[t_0,t_1]$ for performing the Look phase, then it obtains a multiset $Z_i(R(t)) = \{Z_i(r_1(t)),Z_i(r_2(t)),...,Z_i(r_n(t))\}$ for some $t\in [t_0, t_1]$, where $Z_i(r_i(t)) = (0, 0)$. That is, $r_i$ has the (strong) multiplicity detection ability and can count the number of robots sharing a location. More generally, if $P$ is a multiset of points, for any $x$-$y$ coordinate system $Z$, by $Z(P)$ we denote the multiset of the coordinates $Z(p)$ in $Z$ for all $p \in P$.

Let $\{t_i : i = 0,1,\ldots\}$ be the set of time instants at which a robot takes the snapshot $R(t_i)$ during the Look phase. Without loss of generality, we assume $t_i = i$ for all $i = 0,1,\ldots$. Then, an infinite sequence $\Ex : R(0),R(1),\ldots$ is called an execution with an initial configuration $I = R(0)$ that by definition is stationary and without multiplicities. Actually, depending on the algorithm, multiplicities may be created in $R(i)$, with $i>0$. 

Let $P_1$ and $P_2$ be two multisets of points: if $P_2$ can be obtained from $P_1$ by  translation, rotation, reflection, and uniform scaling, then $P_2$ is \emph{similar} to $P_1$. Given a pattern $F$ expressed as a multiset $Z_0(F)$, an algorithm $A$ \emph{forms} $F$ from an initial configuration $I$ if for any execution $\Ex : R(0)(= I),R(1),R(2),\ldots$, there exists a time instant $i>0$ such that $R(i)$ is similar to $F$ and no robots move after $i$, i.e., $R(t) = R(i)$ hold for all real numbers $t\ge i$. 
 
Unlike the initial configuration, in general, not all robots are stationary in $R(i)$ when $i>0$, but at least one robot that takes the snapshot $R(i)$ is stationary by definition. 
Whether or not a given configuration $R$ is stationary (or a robot is stationary at $R$) depends not only on $R$ but also on the execution history, in general. Let $\Ex:R(0),R(1),\ldots, R(f)$ and $\Ex' : R'(0),R'(1),\ldots$ be two executions, and assume $R(f) = R'(0)$. Then $\Ex\Ex' : R(0),R(1),\ldots,R(f)(= R'(0)),R'(1),\ldots$ is always a correct execution for \ssync\ (and hence for \fsync) robots since $R(f)$ is stationary by the definition of  \ssync\ robots. However, this is not the case for  \async\ robots, since in $\Ex'$, $R'(0)$ is assumed to be stationary, but in $\Ex$, $R(f)$ may not be; the transition from $R'(0)$ to $R'(1)$ may be caused by the Look of a robot $r$ which is moving at $R(f)$ and hence cannot observe $R(f)$.
 If an algorithm can guarantee that $R(f)$ is stationary, like for \ssync robots, then we can safely concatenate $\Ex$ and $\Ex'$ to construct a legitimate execution even for \async robots. An execution fragment that starts and ends at a stationary configuration is called a \emph{phase}.

Consider an execution of an algorithm $A$, and assume that at a given time $t$ algorithm $A$ produces a stationary configuration $R(t)$. As observed in the Introduction, $A$ should use some invariants to decide which move to apply to $R(t)$. According to the invariants, a given move is then applied to a subset $S$ of robots in $R(t)$ in order to produce a new stationary configuration. In this general case, a phase is realized by the move, and we could say that the phase is a consequence of the invariants verified at $R(t)$. 

In order to provide the correctness proof for $A$, it is not sufficient to define the invariants associated to a phase, but it is mandatory to prove that the defined moves do not affect the stationarity of robots not belonging to $S$. In fact, if such a situation happens, robots not in $S$ may ‘believe to be in a different phase’ and hence may start moving (possibly producing intractable cases). 

For addressing correctly this problem, we introduce the following definition of \emph{transition-safe algorithm}.

\begin{definition}\label{def:safeness} Let $R$ be a configuration and $S\subseteq R$ be the robots allowed to move according to a move $m$ dictated by algorithm $A$. Let $[t_1, t_2]$ be any time interval in which some robots in $S$ are moving but at least one has not yet reached its target, and let $R'$ be any configuration observed during $[t_1,t_2]$. \\ We say that $m$ is \emph{safe} if in $R’$ algorithm $A$ allows only robots in $S$ to move; $A$ is \emph{transition-safe} if each move in $A$ is safe.
\end{definition}

Note that, in the situation described by the above definition, if $R'$ is stationary and has been obtained from $R$ by means of move $m$, necessarily the adversary has stopped some moving robots in $S$ (while the other moving robots have reached their targets). From there, if $A$ is transition-safe then either move $m$ is again performed by robots in $S$ or still robots in $S$ move according to a move $m'\neq m$ (possibly toward different targets).

\section{Notation and basic properties}\label{sec:notation}
Given two distinct points $u$ and $v$ in the Euclidean plane, let $d(u,v)$ denote their 
distance, let $\Line(u,v)$ denote the straight line passing through these points, 
and let $(u,v)$ ($[u,v]$, resp.) denote the open (closed, resp.) segment containing 
all points in $\Line(u,v)$ that lie between $u$ and $v$. The half-line starting at 
point $u$ (but excluding the point $u$) and passing through $v$ is denoted by $
\halfline(u,v)$. We denote by $\angolo(u,c,v)$ the angle centered in $c$ and with 
sides $\halfline(c,u)$ and $\halfline(c,v)$. The angle $\angolo(u,c,v)$ is measured 
from $u$ to $v$ in clockwise or counter-clockwise direction, the measure is 
always positive and ranges from $0$ to less than $360$ degrees, and the direction in 
which it is taken will be clear by the context.

Given an arbitrary multiset $P$ of points in $\R^2$, $C(P)$ and $c(P)$ denote the smallest enclosing circle of $P$ and its center, respectively. 
Let $C$ be any circle concentric to $C(P)$. 
We say that a point $p\in P$ is \emph{on} $C$ if and only if $p$ is on the circumference of $C$; $\partial C$ denotes all the points of $P$ that are on $C$. We say that a point $p\in P$ is \emph{inside} $C$ if and only if $p$ is in the area enclosed by $C$ but not in $\partial C$; $\Int(C)$ denotes all the points inside $C$. The radius of $C$ is denoted by $\delta(C)$.  
The smallest enclosing circle $C(P)$ is unique and can be computed in linear time~\cite{M83}. A useful characterization of $C(P)$ is expressed by the following property.

\begin{property}\cite{Welz91}\label{prop1} $C(P)$ passes either through two of the points of $P$ that are on the same diameter (antipodal points), or through at least three points. $C(P)$ does not change by eliminating or adding points to $\Int(P)$. $C(P)$ does not change by adding points to $\partial C(P)$. However, it may be possible that $C(P)$ changes by either eliminating or changing positions of points in $\partial C(P)$.
\end{property}

Given a multiset $P$, we say that a point $p\in P$ is \emph{critical} if and only if $C(P) \neq C(P\setminus \{p\})$\footnote{Note that in this work we use operations on multisets.}. It easily follows that if $p\in P$ is a critical point, then $|\partial C(P)\cap \{p\}|=1$.

\begin{property}\cite{CieliebakP02}\label{prop2} If $|\partial C(P)|\ge 4$ then there exists at least one point in $\partial C(P)$ which is not critical.
\end{property}

Given a multiset $P$, consider all the concentric circles that are centered in $c(P)$ and with at least one point of $P$ on them: $C_{\uparrow}^{i}(P)$ denotes the $i$-th of such circles, and they are ordered so that by definition $C_{\uparrow}^{0}(P) = c(P)$ is the first one, $C(P)$ is the last one, and the radius of $C_{\uparrow}^{i}(P)$ is greater than the radius of $C_{\uparrow}^{j}(P)$ if and only if $i>j$.
Additionally, $C_{\downarrow}^{i}(P)$ denotes one of the same concentric circles, but now they are ordered in the opposite direction: $C_{\downarrow}^{0}(P) = C(P)$ is the first one, $c(P)$ by definition is the last one, and the radius of $C_{\downarrow}^{i}(P)$ is greater than the radius of $C_{\downarrow}^{j}(P)$ if and only if $i<j$.

The radius of three of such circles will play a special role in the remainder: $\delta_0(P)= \delta( C_{\downarrow}^{0}(P) )$, $\delta_1(P) = \delta( C_{\downarrow}^{1}(P) )$, and $\delta_2(P)= \delta( C_{\downarrow}^{2}(P) )$ (with $\delta_1$ and $\delta_2$ equal to zero when the corresponding circles do not exist). 

\begin{definition}\label{def:scaletta}
Let $R$ be a configuration.  We  define $\delta_{0,1} = (\delta_0(R) + \delta_1(R))/2$ and $\delta_{0,2} = (\delta_0(R) + \delta_2(R))/2$, and we denote by $C^{0,1}(R)$ and $C^{0,2}(R)$ the circles centered in $c(R)$ and with radii $\delta_{0,1}$ and $\delta_{0,2}$, respectively.
\end{definition}

\begin{definition}\label{def:guard-disk}
Let $R$ be a configuration and $F$ a pattern. Assuming $C(R)=C(F)$, let $d=\delta(C_{\uparrow}^{1}(F))$.
The \emph{guard circle} $C^g(R)$ and the \emph{teleporter circle} $C^t(R)$  
are defined as the circles centered in $c(R)$ of radii equal to $d/2^i$ and $d/2^{(i-1)}$, respectively, with $i>1$ being the  minimum integer such that the following conditions hold:
\begin{itemize}
\item $\Int(C^g(R)) \setminus c(C^g(R)) = \emptyset$;
\item $|\Int(C^t(R)) \setminus \Int(C^g(R)) | \le 1$;
\item $|\partial C^t(R)| \le 1.$
\end{itemize}
\end{definition}

Circles $C^g(R)$ and $C^t(R)$ are not defined for any configuration. For a visualization see Figure~\ref{fig:guards}.
Circle $C^g(R)$ will be used by our algorithm as the place where a specific robot $g$ is moved in order to establish a common reference coordinate system by which robots can embed $F$ on the area occupied by robots. Circle $C^t(R)$, which is larger than $C^g(R)$, represents the place closest to $c(R)$ where a robot deviates in case its trajectory should traverse $C^g(R)$. Doing so, no robots can be confused with $g$.

\begin{figure}[t]
\begin{center}
  \resizebox{0.25\textwidth}{!}{\input{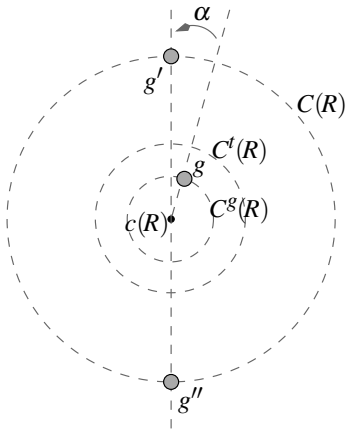}}
\caption{ A visualization of basic concepts and notation in the definition of the common reference system.}
\label{fig:guards}
\end{center}
\end{figure}

\begin{definition}\label{def:alpha}
Let $F$ be a pattern, the \emph{reference angle} $\alpha$ is defined as  the following value:
\[ \alpha = \frac 1 3\cdot \min \{ \angolo(x,c(F),y): x,y\in \partial C(F), x\neq y \}. \]
\end{definition}
Such a reference angle will be used to correctly place robot $g$ on $C^g(R)$ during the formation of pattern $F$.

\subsection{View of a point and symmetries}\label{ssec:view}
We now introduce the concept of \emph{view} of a point in the plane; it can be used by robots to determine whether a configuration $R$ and/or a pattern $F$ 
is symmetric or not.

Given a set of robots $\mathcal R$, the distance $d(r,s)$ between two robots $r$, $s\in \mathcal R$ is the length of the segment connecting their positions in the plane. A map $ \varphi : \mathcal R \rightarrow  \mathcal R$ is called an \emph{isometry} or distance preserving if for any $r,s \in \mathcal R$ one has $d(\varphi(r),\varphi(s))=d(r,s)$. 
If $\mathcal R$ admits only the identity isometry, then $\mathcal R$ is said \emph{asymmetric}, otherwise it is said \emph{symmetric}. These definitions naturally  extend to the corresponding configuration $R$ of robots in $\mathcal R$, and in general to a generic multiset of points $P$. Notice that, whenever a configuration (or a pattern) contains a multiplicity, then it is symmetric, no matter where the multiplicity is. In fact, two distinct robots (or target points of a pattern) in the multiplicity can be always mapped to each other.

Let $P$ be a generic multiset of points. For $p\in P$, with $p\neq c(P)$, we denote by $V^+(p)$ the counter-clockwise view of $P$ computed from $p$. Essentially, $V^+(p)$ is a string whose elements are the polar coordinates of all points in $P$. The elements in $V^+(p)$ are arranged as follows: first $p$, then in order and starting from $c(P)$ all the points in the ray $\halfline(c(P),p)$, and finally all the points in the other rays, with rays processed in counter-clockwise fashion. Similarly, $V^-(p)$ denotes the clockwise view of $P$ computed from $p$. By assuming a lexicographic order for polar coordinates, the view of $p$ is defined as $V(p) = \min \{V^+(p),V^-(p)\}$.

If $c(P)\in P$ then $c(P)$ is said the point in $P$ of \emph{minimum view}, otherwise any $p = \argmin \{V(p'): p'\in P\}$ is said of minimum view in $P$. 
Given $P$ and $P'\subseteq P$, we use the notation $\minview(P')$ to denote any point $p$ with minimum view in $P'$. 

The possible symmetries that $P$ can admit are reflections, rotations and those due to multiplicities. 
$P$ admits a reflection if and only if there exist two points $p,q\in P$, $p,q\neq c(P)$, not necessarily distinct, such that $V^+(p) = V^-(q)$;\footnote{When chirality can be exploited, reflections can be ignored as $V^+(p)$ can be always discriminated from $V^-(q)$.}
$P$ admits a rotation if and only if there exist two distinct points $p,q\in P$, $p,q\neq c(P)$, such that $V^+(p) = V^+(q)$.
It follows that if $P$ is asymmetric then there exists a unique multipoint with minimum view. The above properties can be exploited by robots to detect whether the observed configuration during the Look phase is symmetric or not. When $P$ is symmetric, in general, it cannot be guaranteed the existence of a single point of minimum view. However, if $P$ admits exactly one axis of reflection $\ell$ with $P\cap \ell\neq\emptyset$ and at least a single point on $\ell$, then among all the single points on $\ell$ there exists one of minimum view (otherwise $P$ would admit  another axis of reflection). Similarly, if $P$ admits a rotation with $c(P)\in P$ and $c(P)$ is a single point, then $c(P)$ can be uniquely determined. In all such cases basically a \emph{leader} can be elected in $P$. 

\begin{definition}
	A configuration $R$ is said a \emph{leader configuration} if there exists a robot $r \in R$ such that for each isometry $\varphi$, $\varphi(x)=x$.
\end{definition}
	
In this work we assume any initial configuration $R(0)$ to be a leader configuration. In Figure~\ref{fig:conf-types} relationships among different kinds of configurations are shown.

\begin{figure}[h]
\begin{center}
  \resizebox{0.40\textwidth}{!}{\input{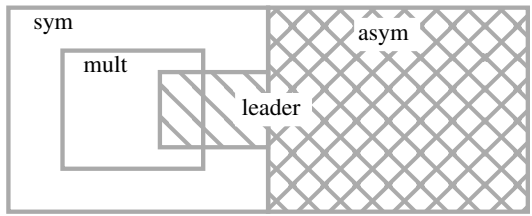}}
\caption{ A visualization of the relationships between the different kinds of configurations addressed in the paper, where ``sym'', ``asym'', ``mult'', and ``leader'' stand for symmetric configurations, asymmetric configurations, configurations with multiplicities, and leader configurations, respectively. In this work initial configurations are leader configurations without multiplicities, which include all the asymmetric configurations.}
\label{fig:conf-types}
\end{center}
\end{figure}

We now redefine the concept of clockwise (and hence of counter-clockwise) direction for a multiset of points $P$ so as to make it independent of a global coordinate system.\footnote{Indeed, this is not necessary when chirality is assumed.}
If $P$ is asymmetric and $p=\minview(P\setminus c(P))$, then the direction used to compute $V(p)$ during the analysis of all the rays starting from $c(P)$ is the \emph{clockwise direction} of $P$. If $P$ is symmetric, there might be many multipoints of minimum view. Symmetries we take care of are of three types: rotations, reflections and those due to multiplicities. If $P$ is reflexive, any direction can be assumed as the clockwise direction of $P$ since they are indistinguishable. In any other case the direction used to compute $V(p)$ from any point $p\neq c(P)$ of minimum view determines the \emph{clockwise direction} of $P$.

\begin{figure*}[h]
\begin{center}
  \resizebox{0.55\textwidth}{!}{\input{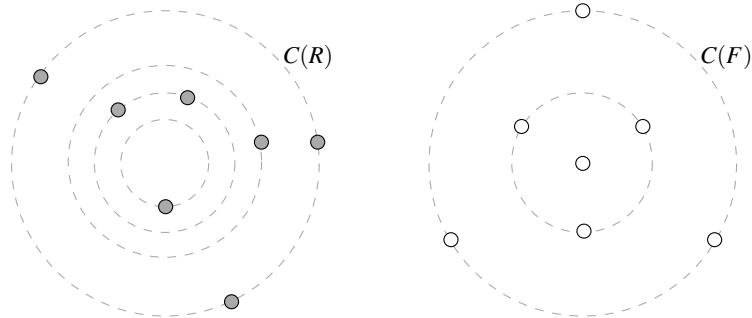}}
\caption{ An example of input for the \pf\ problem: in $(a)$, robots $R$; in $(b)$, the pattern $F$. }
\label{fig:sym}
\end{center}
\end{figure*}

We can now apply the redefined concept of clockwise direction to a configuration $R$ and/or a pattern $F$. For instance, in Figure~\ref{fig:sym}.$(a)$, the multiset $R$ is asymmetric, hence its clockwise direction coincides with that used to compute $V(r)$, with $r=\minview(R)$; whereas in Figure~\ref{fig:sym}.$(b)$, the set $F$ is rotational and reflexive, hence its clockwise direction does not distinguish a left-right orientation.

\section{The strategy}\label{sec:strategy}
In this section, we provide a general description of the strategy underlying our algorithm for the resolution of \apf starting from initial leader configurations. It is based on a functional decomposition approach: the problem is divided into five sub-problems denoted as $\RS$ (Reference System), $\PForm$ (Partial Formation), $\Fin$ (Finalization), $\SB$ (Symmetry Breaking),  and $\Member$. For each sub-problem but $\Member$ an algorithm is provided. Each algorithm is defined in a way that its execution consists of a sequence of phases, whereas $\Member$ crosses different phases. The whole strategy is then realized by composing the algorithms of each phase.

We now provide a high-level description of the four sub-problems.
 
\bloccobis{Problem $\RS$} It concerns the main difficulty arising when the pattern formation problem is addressed: the lack of a unique embedding of $F$ on $R$ that allows each robot to uniquely identifying its target (the final destination point to form the pattern).\footnote{In the literature, this is sometime realized by inducing a common coordinate system. This method can be effective only if $F$ is specified by coordinates and not by distances.} In particular, $\RS$ can be described as the problem of moving some (minimal number of) robots into specific positions such that they can be used by any other robot as a common reference system. Such a reference system should imply a unique mapping from robots to targets, and should be maintained along all the movements of robots (except for the \emph{finalization phase}, where the algorithm moves the robots forming the reference system). 

Our strategy solves $\RS$ by using three robots (called \emph{guards}). Such guards are positioned as described in Figure~\ref{fig:guards}: the \emph{boundary guards} are denoted as $g'$ and $g''$ and are two antipodal robots on $C(R)$, the \emph{internal guard} is denoted as $g$ and is the unique robot on $C^g(R)$ (the \emph{guard circle}). The internal guard is placed so that the angle $\angolo(g,c(R),g')$ is equal to $\alpha$ (the \emph{reference angle}) whose value only depends on $F$. Once $\RS$ is solved, each robot can use the guards to univocally determine its target position in the subsequent phases. 

Notice that in specific cases the presence of $C^g(R)$ implies a refinement to the strategy defined for the sub-problem $\RS$. In fact, in case of a multiplicity in $c(F)$, all the robots in $c(R)$ (but one) must be placed before $C^g(R)$ (and hence the internal guard $g$) is used. Then, two distinct phases are designed for addressing $\RS$: 
\begin{itemize}
\item phase $\Funo$, responsible for setting the external guards $g'$ and $g''$ and, 
      if required, for placing the multiplicity in $c(R)$;
\item phase $\Fdue$, responsible for setting the internal guard $g$.
\end{itemize}

\bloccobis{Problem $\PForm$} This sub-problem concerns moving all the non-guards robots (i.e., $n-3$ robots) toward the targets. In our strategy, phase $\Ftre$ is designed to solve this problem. 
 The difficulties in this phase are the following: (1) during the phase, the reference system must be preserved, (2) the movements must be performed by avoiding undesired multiplicities (collision-free routes), and (3) the movements must be performed without entering into the guard circle (routing through the teleporter circle).

Once the guards are placed, a unique robot per time is chosen to be moved toward its target: it is the one not on a target, closest to an unoccupied target, and of minimum view in case of tie. We are ensured that always one single robot $r$ will be selected since the configuration is maintained a leader configuration by the guards. 
The selected robot is then moved toward one of the closest targets until it reaches such a point. 
All moves must be performed so as to avoid the occurrence of undesired multiplicities; hence, it follows that sometimes the movements are not straightforward toward the target point but robots may deviate their trajectories. To this aim, the strategy makes use of a procedure called \DistMin\ designed ad-hoc for computing alternative trajectories. Moreover, according to the role of the internal guard $g$, robots cannot enter into the circle $C^g(R)$, otherwise the reference system induced by the guards is lost. The latter implies that, in case of a possible route passing through $C^g(R)$, a robot avoids entering into the guard circle by deviating along the boundary of $C^t(R)$ (the \emph{teleporter circle}); such a circle is centered in $c(R)$ and its radius is opportunely chosen so that in its interior there is only $g$ and in $\partial C^t(R)$ there is at most one robot.  

\bloccobis{Problem $\Fin$} It refers to the so-called finalization phase, where the last three robots (the guards) must be moved to their targets to complete the formation of pattern $F$.   In our strategy, a phase $\Fquattro$ is designed to solve this problem. This phase must face a complex task, since (1) moving the guards leads to the loss of the common reference system, and (2) moving without a common reference system makes hard to finalize the pattern formation.

In order to solve \apf, two additional problems must be faced. The first one concerns the case where the initial configuration in symmetric.

 \bloccobis{Problem $\SB$} Consider the case in which the initial leader configuration admits a symmetry while the pattern to be formed is asymmetric. In this situation it is mandatory for each solving algorithm 
 to break the symmetry (i.e., to transform the initial configuration into an asymmetric one). In fact, without breaking the symmetry, any pair of symmetric robots may perform the same kind of movements and this prevent the arbitrary pattern formation. In our strategy, phase $\Fcinque$ is used to address this sub-problem. It carefully moves the robot away from the center (in case of rotational symmetry) or one robot away from the unique axis (in case of reflections) until to obtain a stationary asymmetric configuration. 
 
The main difficulties in this sub-problem are: (1) to avoid the formation of other kind of symmetries that could prevent the pattern formation (e.g., rotational symmetries), and (2) to correctly face the situation in which multiple steps are necessary to reach the target. In the latter case, the algorithm must recognize the obtained asymmetric configuration as a ``configuration originated from a symmetric one where some robot has not yet reached a designed target''.

The second problem that crosses different phases is that described in the Introduction. 

\bloccobis{Problem $\Member$}
This is the problem of avoiding that moves defined in the algorithm can change the membership of the current configuration \emph{while robots are moving}.
 In fact, if such a situation happens, other robots believing to be in a different phase may start moving and then the situation becomes sometimes intractable or even they prevent the algorithm to accomplish the $\pf$ 
(for instance, this is the case for the algorithm proposed in~\cite{FYOKY15}).  
In contrast, our strategy correctly addresses such a general problem by making use of two ingredients: safe moves and an ad-hoc procedure called \SMove. 
By using only moves that are safe, we obtain an algorithm which is transition-safe.  This means that while robots are moving according to a move $m$, only robots planned by $m$ can move until all robots correctly reached the target.
By using Procedure \SMove, the algorithm can control each move that potentially could lead to non-stationary configurations, and hence to change the membership of the current configuration in an uncontrolled manner.\footnote{The technique adopted by means of Procedure \SMove is similar to the so-called \emph{cautious move} in~\cite{FPSV14,MV16}.}
Among others, Procedure  \SMove exploits two main properties of our algorithms that are ensured during the most of the computations: there is at most one moving robot and $C(R(0))=C(R(t))=C(F)$, $t> 0$. Such properties do not hold only in a few of the cases handled by phase $\Fcinque$. For instance, while solving $\RS$, the embedding of $F$ into $R$ is not yet defined but the property that $C(R)=C(F)$ might be exploited. In fact, in some cases \SMove can force the moving robots crossing circles $C^i_\downarrow(F)$ to stop on such circles (hence potentially on a point of $F$). In this way, we force the configuration to become stationary while its membership may have changed, because perhaps an embedding of $F$ that leads to a different phase holds. 

\section{Counter-examples}\label{sec:counter-ex}
In this section we provide fundamental arguments affecting the correctness of~\cite{FYOKY15},~\cite{BT15,BT16b} and~\cite{DPVarx09}. It is worth to remark that our algorithm does not fix all the problems arisen in this section, however our technique might be helpful as a guideline for approaching such inconsistencies.

\subsection{A counter-example to the correctness of the algorithm presented in~\cite{FYOKY15}}\label{ssec:ce-1}

\begin{figure*}[h]
\begin{center}
  \resizebox{0.62\textwidth}{!}{\input{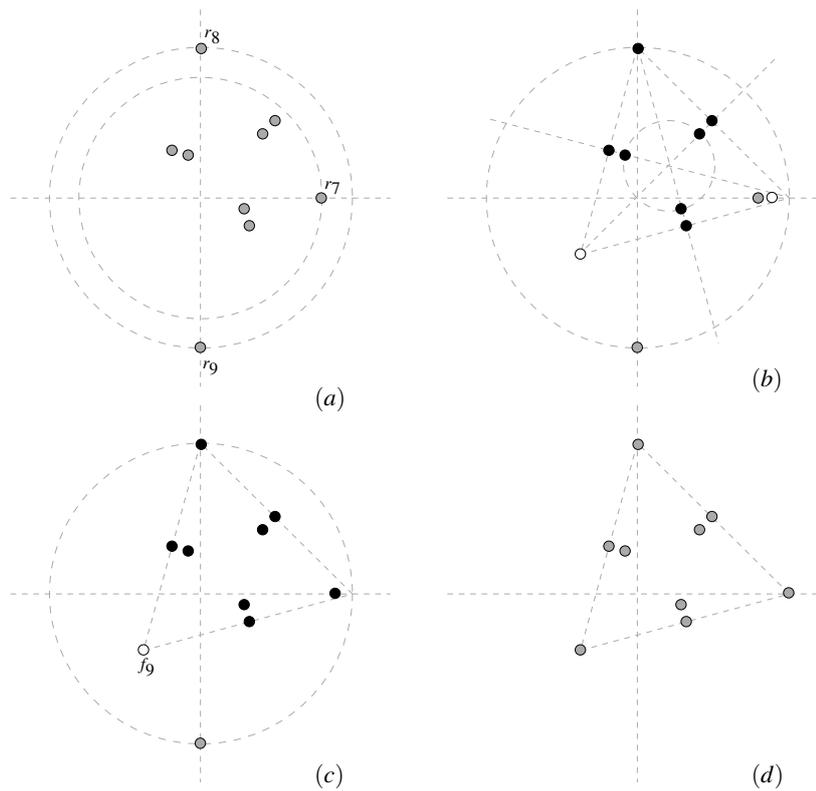}}
\caption{ A counter-example to the correctness of the algorithm $\FORM$ presented in~\cite{FYOKY15}. Grey circles represent robots, white circles represent points in the pattern $F$, black circles represent both robots and points in $F$. (a) An initial configuration $I$; (b) A visualization of a pattern $F$ to be formed along with points of $I$; (c) A configuration $R$ obtainable from $I$ during the movement of robot $r_7$ in the first phase (execution of the algorithm $A_1$ to form a $T$-stable configuration). From it, since invariant $\INV_{\!\!c}$ holds, also robot $r_9$ might start moving (toward $f_9$); (d) A symmetric configuration $R'$ obtainable if both $r_7$ and $r_9$ complete their scheduled moves.}
\label{fig:controesempio}
\end{center}
\end{figure*}

In this section, we provide a counter-example to the correctness of the algorithm $\FORM$ presented in~\cite{FYOKY15} to solve $\pf$ assuming chirality. $\FORM$ is designed to form a pattern $F$ (possibly with multiplicities) from an initial configuration $I$ (without multiplicity) each time $\rho(I)$ divides $\rho(F)$, where $\rho(\cdot)$ is the parameter that captures the symmetricity of a multiset of points in the plane (cf.~\cite{FYOKY15}). Since the algorithm assumes robots empowered with chirality, function $\rho(\cdot)$ only measures the number of possible rotations that the input set of points admits. In fact, dealing with chirality overcomes managing reflections.

An execution of algorithm $\FORM$ is partitioned into four phases called the pattern embedding ($\EMB$), the embedded pattern formation ($\FOR$), the finishing ($\FIN$), and the gathering ($\GAT$) phases. Phase $\EMB$ embeds a given pattern $F$, phase $\FOR$ forms a substantial part $\tilde{F}$ of $F$, and phase $\FIN$ forms the remaining $F\setminus\tilde{F}$ of $F$ to complete the formation. Phase $\GAT$ treats a pattern $F$ with multiplicities. These phases occur in this order, but some of them may be skipped. 

For each of the phases, the authors present an algorithm and an invariant (i.e., predicate) that every configuration in the phase satisfies. Each invariant $\INV$ is considered as a set too; a configuration $R$ satisfies $\INV$ if and only if $R\in \INV$. Authors show that the defined invariants are pairwise disjoint, so exactly one of the algorithms implementing the phases is executed.

We start by briefly recalling both the invariant $\INV_{\!\!\EMB}$ and the algorithm $A_{\EMB}$ for the first phase $\EMB$. For a formal understanding of the arguments below, the reader is invited to refer to~\cite{FYOKY15} for the definition of $\ell$-stable configurations, which in turn defines a set $\Lambda \subseteq \partial C(R)$: 
\begin{itemize}
\item $\INV_{\!\!\EMB} = \neg ( \INV_{\!\!\FOR} \vee \INV_{\!\!\FIN} ) $.     
    Informally, $\INV_{\!\!\FOR}$ is the set of the so called $\ell$-stable configurations $R
    $ such that not all robots in $R \setminus \Lambda$ are located at their 
    final positions in $F$, while $\INV_{\!\!\FIN}$ is the set of 
    configurations $R$ such that all robots in $R \setminus \Lambda$ are 
    located at their final positions in $F$.

    Figure~\ref{fig:controesempio}.(a) shows an initial configuration $I$ such 
    that $I\in \INV_{\!\!\EMB}$. Notice that the number of robots in $I$ is odd and $\rho(I)=1$ (since $I$ is asymmetric).
\item $A_{\EMB}$ consists of three algorithms $A_1$, $A_2$, and $A_3$, that are 
    devoted to three different cases. Such algorithms are responsible of 
    forming a $\ell$-stable configuration. In particular $A_1$ is responsible for forming a $T$-stable configuration, that is a configuration $R$ with 
    exactly three robots in $\partial C(R)$ such that two of them are antipodal 
    and the third one is a midpoint of them. Actually, $A_1$ 
    is invoked when $|\partial C(R)| = 2$ (and this is the case in $I$), and it moves an additional robot on 
    $C(R)$ to get a $T$-stable configuration.

    From $I$,
    algorithm $A_1$ moves robot $r_7$ straightly toward the point 
    $[c(I),r_7]  \cap C(I)$.
\end{itemize}
Both the invariant $\INV_{\!\!\FIN}$ and the algorithm $A_{\FIN}$ for the third phase $\FIN$ are also necessary to build our counter-example. The finishing phase consists of different invariants depending on the kind of $\ell$-stable configuration obtained in the first phase. In particular, when the $T$-stable configuration has been built according to algorithm $A_1$, $\INV_{\!\!\FIN}$ consists of three invariants $\INV_{\!\!a}$, $\INV_{\!\!b}$, and $\INV_{\!\!c}$, each associated to an algorithm that moves one of the three robots on $C(R)$. What we need to explore for the counter-example is  $\INV_{\!\!c}$.
\begin{itemize}
\item A configuration $R$ satisfies $\INV_{\!\!c}$  if and only if $R\setminus F = \{r\}$ and $r$ is on $\tau_c$, where $\tau_c$ is a route designed as follows:
\begin{itemize}
\item  Let $R = \{r_1, r_2, \ldots, r_n\}$, $F = \{f_1, f_2, \ldots, f_n\}$ be the set of robots and pattern points ordered according to their distance from $c(R)$ and $c(F)$, respectively. Assume that $c(R)=c(F)$  and that $R \setminus \{r_n\}$ is similar to $F\setminus \{f_n\}$. In such a 
    case, pattern $F$ can be formed from $R$ by moving just $r_n$ toward $f_n$ along a route $
    \tau_c$ defined as any route such that, for any point $p$ on $\tau_c
    $, still $r_n$ is the unique robot to be moved (toward $f_n$) to form $F$ from 
    $R'$, where $R'$ is constructed from $R$ by replacing $r_n$ with $p$. 
\end{itemize}
Concerning the algorithm executed when $\INV_{\!\!c}$ holds, it simply requires that robot $r_n$ traces the path $\tau_c$.
\end{itemize}
We have recalled all the details necessary to describe the counter-example. Assume now that algorithm $A_1$ is moving $r_7$ to form a $T$-stable configuration, and consider the pattern $F$ depicted in Figure~\ref{fig:controesempio}.(b) composed by the black and white circles. Notice that $\rho(F)=1$ as $F$ is asymmetric. Since $A_1$ moves robot $r_7$ straightly toward the point $[c(I),r_7]  \cap C(I)$, at a certain time it is possible that robot $r_9$ observes (during a Look phase) the configuration $R$ depicted in Figure~\ref{fig:controesempio}.(c). This means that during the movement of $r_7$, robot $r_9$ starts moving toward $f_9$ according to the algorithm associated to the invariant $\INV_{\!\!c}$. If this happens, the following properties hold:
\begin{enumerate}
\item there are two moving robots;
\item each move is due to a different phase;
\item if both moving robots reach their current targets -- see Figure~\ref{fig:controesempio}.(d) -- then the obtained configuration $R'$ admits $\rho(R')=3$ and from there it would be impossible to form $F$. In fact, as proved in~\cite{FYOKY15}, $F$ is formable from $R'$ if and only if $\rho(R')$ divides $\rho(F)$, but here $\rho(F)=1$.
\end{enumerate}

By personal communications, the authors of~\cite{FYOKY15} confirmed us that the provided counter-example and most importantly the rationale behind it represents a main issue that requires deep investigation. While it is possibly easy to find a patch to the counter-example by slightly modifying algorithm $A_1$, 
it is not straightforward to provide general arguments that can ensure the correctness of the whole algorithm. The main question left is: how can be guaranteed among all the phases that the membership of a configuration to a phase does not change while a robot is moving? 

As we are going to show in the correctness section, our algorithm does not suffer of such arguments as it prevents such scenarios. Actually, the authors of~\cite{FYOKY15} are working to devise an erratum to solve the posed problems. The current version of such an erratum is available from~\cite{FYOKY17}.

\subsection{A counter-example to the correctness of the algorithm presented in~\cite{BT15,BT16b}}\label{ssec:ce-2}
In this section, we show how missing arguments affect the correctness of~\cite{BT15,BT16b}. After our personal communications to the authors, a new version~\cite{BT16}\_v3 (i.e., version v3 of~\cite{BT16}) of the extended paper~\cite{BT16}\_v2 has been recently released which includes the case discussed here. Still we believe it is important to report the counter-example as it highlights how the lack of formalism may cause the missing of possible instances.

Similarly to our approach, the algorithm in~\cite{BT16}\_v2 selects a specific robot $r_1$ closest to $c(R)$ to serve as what we call guard. Such a robot is moved in a specific placement in order to be always recognized as such until the very last step that finalizes the formation of the input pattern $F$.

Another ingredient of the algorithm presented in~\cite{BT16}\_v2 we need to describe for our purpose is how the authors get rid of any form of multiplicity detection. When $F$ contains multiplicities, the robots first form a different pattern $\tilde{F}$ obtained by the robots from $F$ as follows:  for each point $p$ of multiplicity $k$, $k - 1$ further points $p_1$, $\ldots$,  $p_{k-1}$ are added such that the distance of each $p_i$ from $c(F)$ equals the distance of $p$ from $c(F)$, $\angolo(p, c(F), p_i) < p_i$, and  $|p_i-p|  = \frac d {4i}$, with $d =\min_{f, f'\in F} | f -f'|$. 

For all algorithm phases except \emph{Termination}, $\tilde{F}$ is used instead of $F$.
To execute the Termination phase, the configuration must be totally ordered (using the set of points, excluding multiplicity information), and at least one robot must be located at each point of $F$ (except maybe the smallest one). 
If there exists a robot $r\neq r_1$ not located at a point in $F$, then $r$ chooses the closest point in $F$ as its destination and rotates toward it while remaining in its circle. The global coordinate system remains unchanged because $r_1$ does not move. Eventually, $r_1$ becomes the only robot not located at its destination, then it moves toward it, and the pattern $F$ is formed.

As shown in the description of our strategy, and in particular in phase $\Fdue$, before creating our internal guard $g$, we ensure to move $k-1$ robots in $c(R)$ if $F$ admits a multiplicity of $k$ elements in $c(F)$. Such type of patterns (with a multiplicity in $c(F)$) are completely ignored in~\cite{BT16}\_v2. In such a case, the general description provided above does not apply.
First of all, $\tilde{F}$ is not well-defined. Secondly, once $r_1$ is correctly placed, $k$ robots should be moved close or perhaps in $c(R)$. This is not clear but in any case there will be robots getting closer to $c(R)$ with respect to $r_1$, hence affecting the global coordinate system.

The lack of details as well as of a rigorous approach led the authors to design a partial algorithm that cannot cope with all possible input. It is worth to remark that in our strategy, the case of a multiplicity in $c(F)$ required a separate phase due to the arisen difficulties. The design of such a phase might not be easy since, in general, several robots must be moved at or close to $c(R)$ \emph{before} the global coordinate system is established.

\begin{figure}[h]
\begin{center}
  \resizebox{0.22\textwidth}{!}{\input{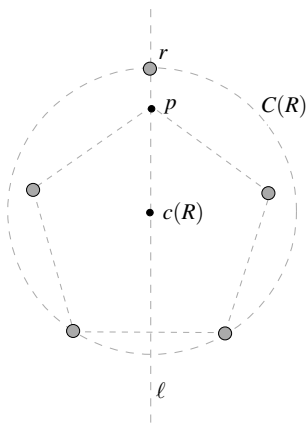}}
\caption{A counter-example to the correctness of the algorithm presented in~\cite{DPVarx09}.}
\label{fig:controesempio-fr}
\end{center}
\end{figure}

\subsection{A counter-example to the correctness of the algorithm presented in~\cite{DPVarx09}}\label{ssec:ce-3}
In~\cite{DPV10}, the authors show the equivalence of \apf and Leader Election for a number of robots greater than or equal to 4, endowed with chirality and for patterns without multiplicities. In~\cite{DPVarx09}, the same authors claim to extend the results to a number of robots greater than or equal to 5 but without chirality. The need of one robot more is justified by the way they want to build a common reference system for all robots. This is done by making use of two robots instead of one, the closest two robots to the center. They claim that such two robots can be always selected since they start from initial leader configurations. Hence a leader and a ‘second leader' can always be detected and moved close to the center so as they can be always recognized by all robots. Even though the description is rather sketchy, their strategy is clear. Unfortunately it is not correct. Just to provide a counter-example, we refer to Figure~\ref{fig:controesempio-fr}. The represented configuration is symmetric, admitting one axis of reflection $\ell$ with a robot $r$ on it. Such a robot can be elected as leader, that is, this is a leader configuration and the proposed algorithm should solve any pattern from it. Following the arguments in~\cite{DPVarx09}, since the configuration admits an axis of reflection with one robot on it, $r$ should move close to $c(R)$ in order to become the unique robot closest to $c(R)$. The authors wrongly argue that during its motion $r$ cannot lose its leadership because it is the unique robot on $\ell$, and by moving along $\ell$ it always remains the unique robot on $\ell$.

Even though during the movement it is true that $r$ remains the unique robot on $\ell$, it cannot move toward $c(R)$ without creating further axes of reflection, (five axes in this case) once position $p$ is reached. Once $r$ is in position $p$, all robots look the same as the configuration admits a rotation, that is the configuration is not a leader configuration anymore, hence preventing the resolution of \apf. As we are going to show in the correctness section, our algorithm carefully addresses the case by means of a sub-phase of phase $\Fcinque$.

\section{The algorithm}\label{sec:algorithm}
In this section, a robot always means an oblivious \async\ robot. Recall that initially a configuration is an initial leader configuration, hence it does not contain multiplicities by definition. Concerning the number of robots $n$, for $n=1$ the \apf problem is trivial. When $n=2$ the problem is unsolvable as the gathering of two robots has been shown to be unsolvable,  see~\cite{CFPS12}.

Concerning the pattern to form, it might contain multiplicities. The case of point formation (Gathering) is delegated to~\cite{CFPS12}, so we do not consider such a case as input for our algorithm.

Similarly to~\cite{FYOKY15}, for $n=3$, we design an ad-hoc algorithm.

\begin{figure*}[t]
\begin{center}
  \resizebox{0.55\textwidth}{!}{\input{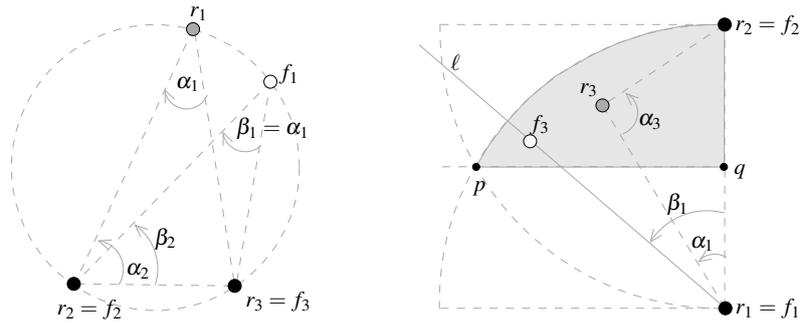}}
\caption{ Pictures showing arguments for the proof of 
          Theorem~\ref{teo:3robots}.
}
\label{fig:3robots}
\end{center}
\end{figure*}

\begin{theorem}\label{teo:3robots}
Let $R = \{r_1,r_2,r_3\}$ be an initial leader configuration with three robots, and let $F=\{f_1,f_2,f_3\}$ be any pattern. Then, there exists an algorithm able to form $F$ from $R$.
\end{theorem} 
\begin{proof}
We assume that $F$ does not contain a point of multiplicity 3, otherwise the algorithm in~\cite{CFPS12} for the Gathering problem can be used. Assume that the point in $F$ and in $R$ are not collinear and let $T_R$ and $T_F$ be the triangles formed by points in $R$ and $F$ respectively. Degenerate cases with collinear points or with $F$ admitting a multiplicity of cardinality two will be considered later. 

Let $\alpha_i$ ($\beta_i$, resp.) be the internal angle of $T_R$ ($T_F$, resp.) associated with the vertex $r_i$ ($f_i$, resp.), for each $i=1,2,3$. 
Since $R$ is a leader configuration, the case with $\alpha_1 = \alpha_2 = \alpha_3=60^\circ$ cannot occur.
Without loss of generality, we assume $\alpha_1 \le \alpha_2 \le \alpha_3$, and  $\beta_1 \le \beta_2 \le \beta_3$ (hence $\alpha_1 \leq 60^\circ$ and $\beta_1 \leq 60^\circ$). We prove the thesis by showing an algorithm that forms $F$ by modifying angles of $T_R$ so that eventually $\alpha_i=\beta_i$, for each $i=1,2,3$. 

According to such angles, the algorithm distinguishes among the following cases when $R$ is not similar to $F$:
\begin{enumerate}
\item[$(a)$] 
$\alpha_2 = \alpha_3, \alpha_1\neq \beta_1$. In this case, robot $r_1$ lies on an axis of symmetry $\ell$ and the algorithm performs the following step: $r_1$ moves along $\ell$ toward the point that makes true the condition $\alpha_1=\beta_1$. Notice that during the movement of $r_1$ no further symmetries can be created as $\alpha_1$ remains the smallest angle being $\beta_1\leq 60^\circ$. Hence $r_1$ is always recognized as the robot to move until its target is reached.  After this step, either $R$ and $F$ are similar (when $\alpha_i=\beta_i$ for each $i=2,3$) or the next case $(b)$ occurs.

\item[$(b)$]
$\alpha_2 = \alpha_3$, $\alpha_1=\beta_1$. Like in case (a), robot $r_1$ lies on an axis of symmetry. Then robots embed $f_2$ and $f_3$ on $r_2$ and $r_3$, respectively, and $r_1$ moves along the unique circle $C$ passing through all the three robots toward a point $p$ that makes the configuration similar to $F$ (since the configuration is symmetric there are two equivalent points toward which $r_1$ can move). For the analysis see case $(c)$, since immediately after $r_1$ starts moving we have $\alpha_2 < \alpha_3$ whereas $\alpha_1$ remains equal to $\beta_1$.

\item[$(c)$]
$\alpha_2 < \alpha_3$, $\alpha_1=\beta_1$.
We analyze two sub-cases, according to the comparison of $\alpha_2$ with $\beta_2$, and of $\alpha_3$ with $\beta_3$.

If $\alpha_2 \ge \beta_2$ and $\alpha_3 \le \beta_3$, the algorithm is based on the following steps: (1) robots embed $f_2$ and $f_3$ on $r_2$ and $r_3$, respectively, (2) robots elect $r_1$ as the only robot allowed to move, (3) like in case $(b)$, $r_1$ moves along the unique circle $C$ passing through all the three robots toward the unique point $p$ that makes the configuration similar to $F$ (here $p$ is unique as the configuration is asymmetric). This case is described by Figure~\ref{fig:3robots} (left side). During the movement: $\alpha_1$ does not change, $\alpha_2$ decreases (without exceeding $\beta_2$), and $\alpha_3$ increases (without exceeding $\beta_3$). As a consequence, the relationships between angles remain the same, the obtained configuration is recognized as belonging to case $(c)$, and $r_1$ is always recognized as the moving robot. Once $r_1$ reaches $p$,  $R$ and $F$ become similar and the pattern is formed. Note that during the movement the configuration is asymmetric as $\alpha_1$ remains strictly less than $\alpha_2$. Angle $\alpha_1$ can match $\alpha_2$ only if $F$ is symmetric and $r_1$ has reached its target $p$.

If $\alpha_2 \le \beta_2$ and $\alpha_3 \ge \beta_3$, the algorithm remains the same. The only difference is the analysis of the angles during the movement of $r_1$: $\alpha_1$ does not change, $\alpha_2$ increases, and $\alpha_3$ decreases. Anyway, as in the previous sub-case, $r_1$ is always recognized as the moving robot and the pattern is eventually formed.

\item[$(d)$]
$\alpha_2 < \alpha_3, \alpha_1\neq \beta_1$. This case concerns different sub-cases that can be described by referring to Figure~\ref{fig:3robots} (right side). Without loss of generality, in the figure we assume $dist(r_1,r_2)\ge dist(r_1, r_3)\ge  dist(r_2, r_3)$, and hence robot $r_3$ is inside  the shaded area depicted in Figure~\ref{fig:3robots}. Notice that $r_3$ cannot be on the arc of the circle delimited by points $p$ and $r_2$, otherwise $\alpha_3=\alpha_2$ against hypothesis, and cannot be in the segment $[q,r_2]$ as otherwise the three robots would be collinear. If $r_3$ is in the interior of the shaded area, then $\alpha_1 < \alpha_2 < \alpha_3$ and the initial configuration is asymmetric; if $r_3$ belongs to the open segment $(p,q)$, we get $\alpha_1 = \alpha_2 < \alpha_3$. 

Assuming $dist(f_1,f_2)\ge dist(f_1, f_3)\ge dist(f_2, f_3)$, the algorithm is based on the following steps: (1) robots embed $f_1$ and $f_2$ on $r_1$ and $r_2$, respectively, (2) robots elect $r_3$ as the only robot allowed to move, (3) $r_3$ embeds $f_3$ as closer as possible to itself, and (4) $r_3$ moves straightly toward $f_3$. 
According to step 3, it follows that $f_3$ is in the segment formed by intersecting the half-line $\ell=\halfline(r_1f_3)$ with the shaded area. In particular, it is in the interior of the shaded area when $F$ is asymmetric, whereas it is on the boundary of that area if $F$ is symmetric. To conclude the analysis of $f_3$, notice that this point cannot be on the half-line starting from $r_1$ and passing through $r_3$ since we assumed $\alpha_1\neq \beta_1$. 

It is clear that, during the movements, $r_3$ will be always elected as the only robot allowed to move and hence it eventually reaches the target to form $F$. 
\end{enumerate}
Summarizing, when the algorithm handles a leader configuration in case $(a)$, the generated configurations remain in the same case $(a)$ until a stationary leader configuration belonging to case $(b)$ is created or the pattern is formed. When the algorithm handles a leader configuration in cases $(b)$, as soon as the moving robot starts its movement the configuration belongs to case $(c)$ where the same move is continued until forming the pattern. If a configuration is in case $(c)$, or $(d)$ the generated configurations remain in the same case until the pattern is formed.

Now, let us consider the degenerated cases when $\alpha_3=180^\circ$ or $\beta_3=180^\circ$ and $F$ has no multiplicities. If $\alpha_3=180^\circ$, the embedding is like in case $(d)$ and, referring to Figure~\ref{fig:3robots}, robot $r_3$ is in the segment $[q, r_2)$ whereas $f_3$ is either in the same segment or in the interior of the shaded area, according to whether $\beta_3=180^\circ$ or not. The algorithm simply moves $r_3$ toward $f_3$.

The case $\alpha_3\not =180^\circ$ and $\beta_3=180^\circ$ requires a deeper analysis.
Obviously, $\alpha_1\not = \beta_1$ as $\beta_1=0$ whereas $\alpha_1\not =0$. If $\alpha_2 =\alpha_3$ the configuration is symmetric, but the algorithm cannot act as in cases $(a)$ since $r_1$ should go to infinity to make $\alpha_1=\beta_1=0$. Then, if $\alpha_1 < \alpha_2$, $r_1$ moves along the circle $C$ until $\alpha_1 = \alpha_2$. If the robot is stopped during the movement, $\alpha_1 <\alpha_2 < \alpha_3$ holds, and $r_1$ is recognized as the robot to be moved on $C$ until $\alpha_1 = \alpha_2$.  When $\alpha_1 = \alpha_2$, still referring to Figure~\ref{fig:3robots}, $T_R$ is symmetric and $r_3$ is in the segment $(p,q)$. Robot $r_3$ moves toward $q$ maintaining $\alpha_1 = \alpha_2$, and then it is always recognizable if stopped. When $r_3$ reaches $q$ then $\alpha_3=180^\circ$ and the algorithm acts as above by moving $r_3$ toward $f_3$.

Finally, if $F$ admits a multiplicity it is assumed $f_2\equiv f_3$, hence $f_1$ is not in the multiplicity. Then the algorithm simply acts as in the above degenerate cases when $\beta_3=180^\circ$.
\qed
\end{proof}

In the remainder, we will provide an algorithm able to form any pattern $F$ (which is not a single point) starting from any initial leader configuration with $n\ge 4$. A possible input for the algorithm is shown in Figure~\ref{fig:input}. We construct our algorithm in such a way that the execution consists of a sequence of phases $\Funo$, $\Fdue$, $\Ftre$, $\Fquattro$, and $\Fcinque$ not necessarily in this order (some of which might be skipped whereas some other might be repeated a finite number of times). To each phase, we assign an invariant such that every configuration satisfies exactly one of the invariants (hence robots can correctly recognize the phase in which they are). Since each algorithm associated to a phase is composed of different kinds of moves, each phase is divided into sub-phases. Each sub-phase is characterized by a single move. Moreover, apart for a few exceptions, each time robots switch to a different phase/sub-phase, the current configuration is stationary, that is the move performed in a sub-phase is initiated from a stationary configuration (this property is crucial to prove the correctness of our algorithm). 
Basically, whenever a robot becomes active, it can deduce from the acquired snapshot to which phase and sub-phase the observed configuration belongs to, and whether it is a robot designated to move. In case it is its turn to move, it applies the move associated to the sub-phase detected. As it will be shown in the proof of correctness, our algorithm always maintains the current configuration as a leader configuration until the pattern is formed. The algorithm also assures there will always be exactly one robot moving in phases $\Funo - \Fquattro$ and at most two robots moving concurrently in phase $\Fcinque$, however, all moves are safe. In turn we prove the algorithm is transition-safe.

Table~\ref{tab:basic-variables} contains all predicates required by our algorithm, whereas Table~\ref{tab:phases} describes in which phase a configuration is, according to the specified properties. Notice that for each arbitrary phase/sub-phase $\mathcal{X}$ we need three predicates $\mathcal{X}_{s}$, $\mathcal{X}_{d}$, and $\mathcal{X}_{e}$ to distinguish between the invariant that the configuration satisfies at the beginning of the phase (start), once robots start to move (during), and once the moving robots have terminated to apply the same move (end), respectively. 
In the most cases, we have $\mathcal{X}_{d}= \mathcal{X}_{s}$; in the remaining cases, as it will be clarified in the correctness section, when $\mathcal{X}_{d}\neq \mathcal{X}_{s}$, $\mathcal{X}_{d}$ and $\mathcal{X}_{e}$ always coincide. For this reason we omit $\mathcal{X}_{d}$ in the presented tables.

About moves, Table~\ref{tab:F1-phases}--\ref{tab:F5-phases} contain a description of all moves applied by our algorithm for each phase. As described in Section~\ref{sec:strategy}, sometimes the trajectories defined by the proposed moves are opportunely manipulated so as to guarantee stationarity (by means of Procedure \SMove) and to avoid collisions (by means of Procedure \DistMin). 

\begin{table*}
\caption{ The basic Boolean variables used to define all the phases' invariants. In the first column, the phase in which the corresponding variables are mainly used. Missing notations can be found in the corresponding sections. 
The table contains a formal definition for all variants except $\q$ and $\iuno, \ldots, \isei$. Such variables are formally defined in Section~\ref{ssec:f4}, where they are used to recognize the guards in different scenarios when guards have to be moved in the finalization phase.
}
\label{tab:basic-variables}
\bgroup
\def\arraystretch{1.4}
\setlength{\tabcolsep}{5pt}
\begin{center}
  \begin{tabular}{ | c | c | p{0.75 \textwidth} | }
    \hline
    \textit{phase}  &\textit{var}  &  \textit{definition} \\ \hline \hline
      
      *      & $\a$ & $R$ is a leader configuration \\ \hline\hline

    \multirow{9}{*}{$\Funo$} & $\sdue$ & $|\partial C(R)|=2$ \\ \cline{2-3}
     & $\stre$ & $|\partial C(R)|=3$ \\ \cline{2-3}     
     & $\spiu$ & $|\partial C(R)|> 3$ \\ \cline{2-3}     
     & $\c$ & $|\{c(R)\}\cap R|=1$ \\ \cline{2-3}    
     & $\tzero $ & $\stre\ \wedge$ points in $\partial C(R)$ form a 
                triangle of angles all different from $90^{\circ}$  \\ \cline{2-3} 
     & $\tuno $ &  $\stre\ \wedge$ points in $\partial C(R)$ form a 
                triangle of angles equal to $30^{\circ}$, 
                $60^{\circ}$, and $90^{\circ}$  \\ \cline{2-3} 

     & $\l$ & $|\partial C_{\downarrow}^{1}(R)|\ge 2 \vee \delta_1\le \delta_{0,2}$ \\ \cline{2-3}     

     & $\mzero$ & $|\{c(F)\}\cap F|=k$, $k>1$ \\ \cline{2-3}  
     & $\muno$ & $\mzero \wedge |\{c(R)\}\cap R|=k-1$  \\ \hline\hline  
    
    \multirow{3}{*}{$\Fdue$} 
     & $\gzero $ & $|\partial C^g(R)| = 1$   \\ \cline{2-3} 
     & $\guno $ & $\gzero \wedge \exists! g'\in \partial C(R): \angolo(g,c(R),g')=\alpha$  \\ \cline{2-3}
     & $\gdue $ & $\guno \wedge \exists g'' \in \partial C(R)$ antipodal to $g’$ \\ \cline{2-3} 
     & $\fdue$  & $(\mzero \Rightarrow \muno) \wedge \sdue \wedge \l$ \\ \hline\hline

    \multirow{3}{*}{$\Ftre$} & $\dzero$ & $\partial C^t(R)=\{r\}$ \\ \cline{2-3} 
     & $\duno$ & $\dzero \wedge \exists f^*\in F^*:~ r = [c(F),f^*]\cap C^t(R)$ \\ \cline{2-3}
     & $\ddue$ & $C^t(R)\cap (r, \mu(r)] \neq \emptyset$, where $r = \minview( R^{\neg m}_{\eta} )$  \\ \cline{2-3}           
     & $\ftre$  & $(\mzero \Rightarrow \muno) \wedge \gdue$ \\ \hline\hline

    \multirow{3}{*}{$\Fquattro$}
         & $\q$ & $\partial C(F) = F$, $F$ without multiplicities, and $|F|-1$ points of $F$ are on the same semi-circle \\ \cline{2-3} 
     & $\iuno, \ldots,\isei$ & guards $g$, $g'$ and $g''$ are detectable $\wedge\ R\setminus\{g,g''\}$  is similar to $F\setminus\{\mu(g),\mu(g'')\}$ \\ \cline{2-3}        
     & $\fquattro$ & $\neg \q \wedge (\iuno \vee \idue)~ \vee ~\q \wedge (\itre \vee \iquattro \vee \icinque \vee \isei) $ \\ \hline\hline

    \multirow{9}{*}{$\Fcinque$} 
    & $\bzero$ &  $R$ is symmetric \\ \cline{2-3}    
    & $\buno$ &  $\exists$ a reflection axis $\ell$ for $R$:                             
                                           $|R\cap \ell| = 1$ \\ \cline{2-3}
     & $\bdue$ & $\exists$ a reflection axis $\ell$ for $R$: $|R\cap \ell|\ge 2$ $\wedge$ 
                 at least two robots in $R\cap \ell$ are not critical \\ \cline{2-3}    
 
     & $\btre$ & $\exists$ a reflection axis $\ell$ for $R$: $|R\cap \ell|= 2$ $\wedge$
                 exactly one robot in $R\cap \ell$ is critical \\ \cline{2-3}      
     & $\bquattro$ & $\exists$ a reflection axis $\ell$ for $R$: $|R\cap \ell|=3$ $\wedge$
                     exactly two robots in $R\cap \ell$ are critical \\ \cline{2-3}      
     & $\bcinque$ & $\exists$ a reflection axis $\ell$ for $R$: $|R\cap \ell|=2$ $\wedge$
                     both robots in $R\cap \ell$ are critical \\ \cline{2-3}                  
     & $\e$ & $R$ does not contain multiplicities \\ \cline{2-3}      
     & $\zuno$ & let $\partial C(R) = \{r, r1, r2\}$ with $r$ faraway robot, and  
                 let $t$ be the point on $C(R)$ closest to $r$ s.t. 
                 points $t$, $r_1$, and $r_2$ form a right-angled triangle: 
                  $\exists$ a reflection axis $\ell$ for $R'=R\setminus\{r\}$
                 that reflects $r_1$ to $r_2$ $\wedge$ $r$ belongs to the smallest arc of $C(R)$ 
                 between $\ell$ and $t$ (excluded) $\wedge$ $R'\cap \ell=\emptyset$\\ \cline{2-3}      
     & $\zdue$ &  
                 $\exists r\in R$: 
                 $\partial C_{\uparrow}^{1}(R) = \{r\}$ 
                 $\wedge$ 
                 $R\setminus \{r\}$ is symmetric  
                 $\wedge$ 
                 [ either $\c$ or 
                   ($\neg\bzero$ $\wedge$ $\neg\idue$ $\wedge$ 
                    $d(c(R),r) < \delta(C^g(R))$ $\wedge$ 
                    $\forall r'\in \partial C(R), \angolo(r,c(R),r')\neq \alpha$)] \\ \cline{2-3}      
     & $\uuno$ & $\exists i>0$: $\partial C_{\uparrow}^{i}(R) = \{\robotdue\} $ $\wedge$ $d( \robotdue,c(R) ) = 1/2 [ \delta(C_{\uparrow}^{i-1}(R)) + \delta(C_{\uparrow}^{i+1}(R)) ]$  \\ \cline{2-3}        
     & $\udue$ & $\robotdue \not \in [t^{60},t^{55}]  \vee \robotdue=t^{x} \vee \robotdue=t^{55}$ \\ \cline{2-3}        
     & $\fcinque$ & $(\bzero \vee \zuno \vee \zdue)\wedge \e$ \\ \hline\hline 
  
          *      & $\w $ & $R$ is similar to $F$   \\ \hline 

  \end{tabular}
\end{center}
\egroup
\end{table*}

\begin{table*}
\caption{ Each label on the first column specifies a different phase. In column ‘start', for each phase it is specified the invariant that holds while a configuration belongs to the corresponding phase. In column ‘end', it is specified the possible phases outside the considered one that can be reached or whether $\w$ may hold.
}
\label{tab:phases}
\def\arraystretch{1.4}
\setlength{\tabcolsep}{5pt}
\begin{center}
  \begin{tabular}{ | c | r | r | }
    \hline
    \textit{phase}  & \textit{start}    &  \textit{end} \\ \hline\hline
    
    $\Funo$ & $\FunoS$ &  $\FunoE$ \\ \hline 
            
    $\Fdue$ & $\FdueS$ &  $\FdueE$ \\ \hline 
            
    $\Ftre$    & $\FtreS$ &  $\FtreE$ \\ \hline 
            
    $\Fquattro$ & $\FquattroS$ &  $\FquattroE$ \\ \hline 
    
    $\Fcinque$ & $\FcinqueS$ &  $\FcinqueE$ \\ \hline 
                         
  \end{tabular}
  \end{center}
\end{table*}

It is important to keep in mind that during the whole algorithm it is assumed a first embedding of $F$ such that $C(R)=C(F)$. Actually $C(R)$ never changes in phases $\Funo - \Fquattro$ whereas in phase $\Fcinque$ it might change, however within $\Fcinque$ robots do not exploit the embedding, hence the assumption considering $C(R)=C(F)$ does not affect any reasoning.
We are now ready to consider each phase separately to see how the desired behavior is obtained.


\begin{figure*}[t]
\begin{center}
  \resizebox{0.53\textwidth}{!}{\input{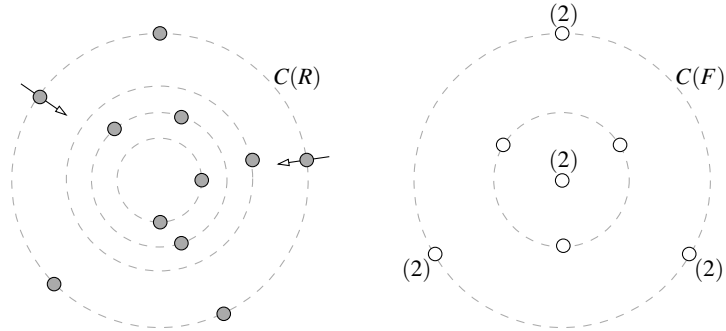}}
\caption{ An example of input for the \pf\ problem: robots $R$ (left) and the pattern $F$ (right). The number closed to a point denotes a multiplicity. Notice that $R$ belongs to the sub-phase $\Auno$ (i.e., predicates $\Funo$ and $\Auno_s=\AunoS$ hold), hence the algorithm applies move $m_1$ (cf Table~\ref{tab:F1-phases}).}
\label{fig:input}
\end{center}
\end{figure*}

\subsection{Phase $\Funo$}
As described in Section~\ref{sec:strategy}, this phase is responsible for setting the external guards $g'$ and $g''$ and, if required, for placing robots in $c(F)$. Due to its complexity, the algorithm for this phase is composed of many different moves, and each move is referred to a sub-phase. Table~\ref{tab:F1-phases} describes all such sub-phases, and also the corresponding invariants and moves.

\noindent
\begin{table*}
\caption{ (left) Invariants and moves for all the sub-phases of $\Funo$. 
Each label on the first column specifies a different sub-phase to which a configuration belongs to. Then, with respect to each sub-phase, in the upper (shaded) side it is specified the invariant that the configuration satisfies at the beginning of the phase (start), and which sub-phase within $\Funo$ can be reached (end). In the lower side, it is specified the corresponding move performed by the algorithm, and on the last column the possible phases outside $\Funo$ that can be reached or whether $\w$ may hold. (right) Description of the moves performed by the algorithm in $\Funo$.
}
\label{tab:F1-phases}
\begin{center}
\begin{minipage}[t]{.45\linewidth} 
\definecolor{grigio}{rgb}{0.9,0.9,0.9}
\bgroup
\rowcolors{1}{white}{grigio}
\def\arraystretch{1.4}
\setlength{\tabcolsep}{5pt}
\begin{center}
  \begin{tabular}{ | c | c | c | c |}
    \hline
    \textit{phase}  & \textit{start}  &  \textit{end} \\ \hline\hline
    $\Auno$ & $\AunoS$ & $\AunoE$ \\ \hline 
            & $m_1$ & $\Fquattro_s \vee \w$ \\ \hline
            
    $\Adue$ & $\AdueS$ &  $\AdueE$ \\ \hline 
            & $m_2$ & $\Fquattro_s \vee \w$ \\ \hline
            
    $\B$  & $\BS$ &  $\BE$ \\ \hline 
            & $m_3$ & \\ \hline         
            
    $\Cuno$ & $\CunoS$ & $\CunoE$ \\ \hline 
            & $m_4$ & $\Fquattro_s \vee \Fcinque_s \vee \w$ \\ \hline         
            
    $\Cdue$ & $\CdueS$ & $\CdueE$ \\ \hline 
            & $m_5$ & \\ \hline         
            
    $\D$    & $\DS$ & $\DE$ \\ \hline 
            & $m_6$ & $\Ftre_s$ \\ \hline         
            
    $\Euno$ & $\EunoS$ & $\Edue$ \\ \hline 
            & $m_1$ & $\Fdue_s \vee \Fquattro_s \vee \w$ \\ \hline 

    $\Edue$ & $\EdueS$ & \\ \hline 
            & $m_2$ & $\Fdue_s \vee \Fquattro_s \vee \w$\\ \hline         
                    
  \end{tabular}
\end{center}
\egroup
\end{minipage}
\hspace{3mm}
\begin{minipage}[b]{.50\linewidth}
\bgroup
\def\arraystretch{1.3}
\setlength{\tabcolsep}{5pt}
\begin{center}
  \begin{tabular}{ | c | p{0.85 \columnwidth} | }
    \hline
    \textit{name}  &  \textit{description} \\ \hline\hline
    $m_1$ & 
Let $r$ be the not critical robot on $C(R)$ with minimum view. $r$ moves according to \SMove toward $t=[r,c(R)] \cap  C^{0,1}(R)$.
    \\ \hline 

    $m_2$ & 
Let $r$ be the robot on $C_{\downarrow}^{1}(R)$. $r$ moves according to \SMove toward $t=[r,c(R)] \cap C^{0,2}(R)$.
    \\ \hline     

    $m_3$ & 
Let $\partial C(R)= \{r_1,r_2\}$, $r\in C_\downarrow^1(R)$ of minimum view and $p_1$, $p_2$ be the intersections of $C(R)$ with the line perpendicular to $[r_1,r_2]$ passing through $c(R)$. $r$ moves radially toward $t$ on $C(R)$ if $t\not \in \{r_1,r_2,p_1,p_2\}$, else toward $C(R)$ along the line tangent to $C_\downarrow^1(R)$ in $r$. \\ \hline 

    $m_4$ & 
The three robots on $C(R)$ form a triangle with angles $\alpha_1\geq \alpha_2\geq \alpha_3$ and let $r_1$, $r_2$ and $r_3$ be the three corresponding robots. For equal angles, the role of the robot is selected according to the view, i. e. if $\alpha_1=\alpha_2$ then the view of $r_1$ is smaller than that of $r_2$. $r_2$ rotates  according to \SMove toward the closest point $t$ such that $\alpha_1$ equals $90^\circ$.
    \\ \hline 

    $m_5$ & 
Let $r$ be the robot on $C(R)$ such that its antipodal point is not in $R$. $r$ rotates toward the closest point $t$ such that the triangle formed by $t$ and the two antipodal robots admits angles of $30^\circ$, $60^\circ$ and $90^\circ$.
    \\ \hline 
    
    $m_6$ & 
The robot closest to $c(R)$ but not in $c(R)$, of minimum view, moves toward $c(R)$. 
    \\ \hline                        
   \end{tabular}
\end{center}
\egroup
\end{minipage}
\end{center}
\end{table*}

The first sub-phase $\A$ (logically divided into $\Auno$ and $\Adue$) leaves exactly three robots on $C(R)$ when more than three robots are there. This is realized by selecting any not \emph{critical} robots for $C(R)$ but three (see Property~\ref{prop2}). We remind that a robot is critical for $C(R)$ if its removal makes $C(R)$ changing. The selected robots are moved one by one inside $C(R)$ so as to maintain the configuration being asymmetric. 
This is realized by moving each of such robots toward a specific new circle $C_{\downarrow}^1(R)$ (hence concentric to $C(R)$ and inside it). 

As an example, the configuration $R$ in Figure~\ref{fig:input} belongs to the sub-phase $\Auno$ (i.e., predicate $\Auno_s=\AunoS$ holds), and hence move $m_1$ (cf Table~\ref{tab:F1-phases}) is applied. As soon as a robot $r$ leaves $\partial C(R)$, the obtained configuration switches to $\Adue$ and hence move $m_2$ is applied until $r$ reaches the desired circle. In fact, if $r$ ends its movement before reaching its target, it is selected again by the algorithm and again  move $m_2$ is applied. Once $r$ reaches its target, predicate $\Auno_s=\AunoS$ holds. Such moves are repeated so that the configuration with $|\partial C(R)|=3$ in Figure~\ref{fig:fase_A-C1} left is obtained: this configuration belongs to $\Cuno$ since $\Cuno_s=\CunoS$ holds.


\begin{algorithm}[ht]{{
\SetKwComment{Comment}{/*}{*/}
\SetKwInput{Proc}{Procedure}
\Proc{\SMove}
\SetKwInOut{Input}{Input}
\Input{target $t$.}
\BlankLine
Let $m$ be the current move,\\ 
$tmp=t$\;
    \If{$m\in \{m_1, m_2, m_7\}$}{\If{$\exists p\in (r,t)$ being the closest point to $r$ intersecting a circle $C_{\downarrow}^j(F)$\label{algo:mFstart}}         
      {$tmp=p$\label{algo:mFend}}
        }   
    \If{$m=m_4$}{
    \If{the number of distinct points in $\partial C(F)$ is 3 and $\exists p$ along the way toward $tmp$ s. t. $(\partial C(R)\setminus \{r\})\cup \{p\}$ form a triangle similar to that formed by $\partial C(F)$\label{algo:mFquattros}}{$tmp=p$\label{algo:mFquattroe}}
    \If{$\exists r'\in \partial C^g(R) \wedge \exists p$ along the way toward $tmp$ s. t. $\angolo(p,c,r')=\alpha$\label{algo:mAquattros}}{$tmp=p$\label{algo:mAquattroe}}
    \If{$\exists p$ along the way toward $tmp$ s.t. in $R\setminus \{r\}\cup \{p\}$ predicate $\zuno$ holds\label{algo:mWquattros}}{$tmp=p$\label{algo:mWAquattroe}}

}
    move toward $tmp$\; 
}}
\caption{ Procedure \SMove\ performed by any robot $r$ when moves $m_1$, $m_2$, $m_4$ or $m_7$ are executed.}
\label{alg:stationary-move}
\end{algorithm}


Actually, both $m_1$ and $m_2$ are performed by invoking Procedure \SMove (see Algorithm~\ref{alg:stationary-move}). This is done because while moving, a robot $r$ may incur along the way in a point $p$ that makes the current configuration belonging to the finalization phase $\Fquattro$ or even final, that is predicate $\w$ holds. Such situations may happen when according to some embedding of $F$ on $R$, difficult to detect at this stage, $p$ coincides with a point in $F$. If other robots perform their Look phase while $r$ is on $p$, they may start moving according to a different rule specified by our strategy, hence violating the desired property to maintain stationarity among phases. In order to avoid such a behavior, we simply force $r$ to stop on points that potentially may cause the described situations. For $m_1$ and $m_2$, such points belong to some $C_{\downarrow}^i(F)$, see Lines~\ref{algo:mFstart}-\ref{algo:mFend} of \SMove. If after a stop, still the configuration belongs to $\Funo$, then the same robot $r$ will be selected again to keep on moving.

The sub-phase $\B$ is concerned with the case of just two robots in $\partial C(R)$ and $F$ admits a multiplicity in $c(F)$. In such a case, a third robot is moved from $\Int(C(R))$ to $C(R)$ (see move $m_3$ of Table~\ref{tab:F1-phases}).

Sub-class $\C$ (logically divided into $\Cuno$ and $\Cdue$) processes configurations with exactly three robots in $\partial C(R)$, and moves them so that they form a triangle with angles of $30^\circ$, $60^\circ$ and $90^\circ$. Now, assume that the three robots form a triangle with angles $\alpha_1\geq \alpha_2\geq \alpha_3$ and let $r_1$, $r_2$ and $r_3$ be the three corresponding robots. $\Cuno$ takes care of the case in which all the angles are different from $90^\circ$ (see Figure~\ref{fig:fase_A-C1}, left side): by move $m_4$, robot $r_2$ rotates on $C(R)$ in such a way that $\alpha_1$ becomes of $90^\circ$ (see Figure~\ref{fig:fase_A-C1}, right side). 

\begin{figure*}[t]
\begin{center}
  \resizebox{0.53\textwidth}{!}{\input{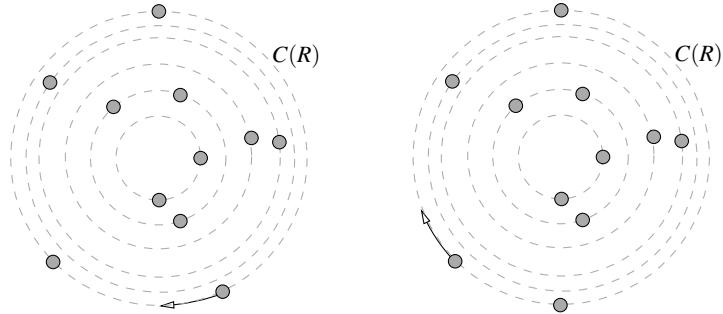}}
\caption{ Configurations obtained from $R$ as described in Figure~\ref{fig:input} after sub-phases $\A$ (left) and $\Cuno$ (right), respectively. The configuration on the right side is in phase $\Cdue$.}
\label{fig:fase_A-C1}
\end{center}
\end{figure*}

\begin{figure*}[t]
\begin{center}
  \resizebox{0.53\textwidth}{!}{\input{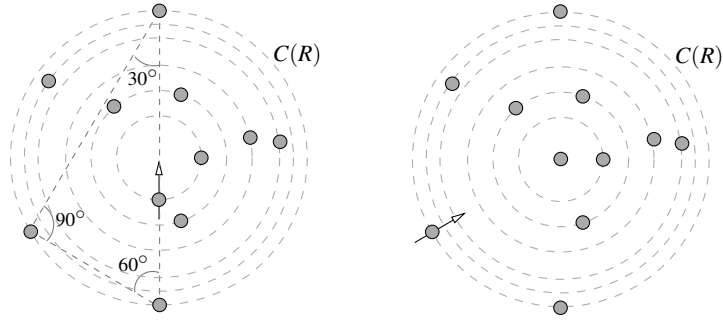}}
\caption{ Configurations obtained from $R$ as described in (the right side of) Figure~\ref{fig:fase_A-C1} after sub-phases $\Cdue$ (left) and $\D$ (right), respectively. }
\label{fig:fase_C2-D}
\end{center}
\end{figure*}

\begin{figure*}[t]
\begin{center}
  \resizebox{0.53\textwidth}{!}{\input{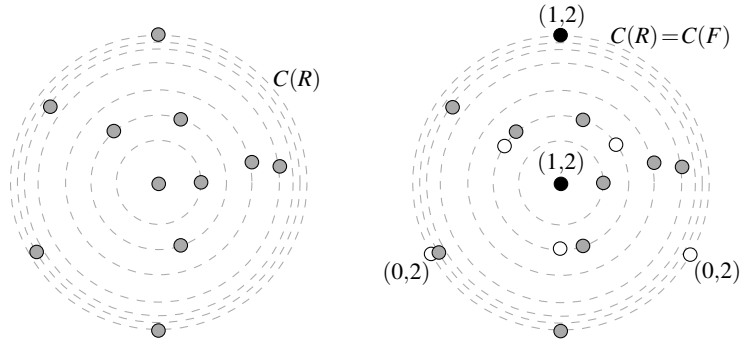}}
\caption{ Configurations obtained from $R$ as described in (the right side of) Figure.~\ref{fig:fase_C2-D} after sub-phases $\E$ (left); such a configuration will be processed by phase $\Fdue$ (in particular, by the sub-phase $\Guno$). On the right, a possible embedding of $F$ onto $R$: black points denote robots on targets (points of $F$), a pair of numbers denotes multiplicities of robots and targets, respectively. The real embedding will be fixed only when the internal guard $g$ will be correctly placed (cf Figures~\ref{fig:fase_G} and~\ref{fig:fase_F3}, left side).}
\label{fig:fase_E}
\end{center}
\end{figure*}

Similarly to $m_1$ and $m_2$, move $m_4$ is performed by invoking Procedure \SMove. The points in $F$ on which the moving robot $r_2$ must stop are now a bit different than before, as $m_4$ rotates $r_2$ along $C(R)$ rather than moving straightly. This may cause different situations: still incurring in points in $F$ (Lines~\ref{algo:mFquattros}-\ref{algo:mFquattroe}),
or creating an unexpected reference angle of $\alpha$ degrees (Lines~\ref{algo:mAquattros}-\ref{algo:mAquattroe}), or creating an asymmetric configuration belonging to $\Fcinque$ (Lines~\ref{algo:mWquattros}-\ref{algo:mWAquattroe}). 

If $r_2$ completes move $m_4$ by making $\alpha_1=90^\circ$ like in Figure~\ref{fig:fase_A-C1}, then there are two antipodal robots on $C(R)$ that will be detected as $g'$ and $g''$ in the subsequent phases. The third robot $r$ on $C(R)$ is now moved either along $C(R)$ or toward a position inside $C(R)$, depending on whether $F$ admits a multiplicity in $c(F)$ of $k>1$ elements or not. In the former case (see Figure~\ref{fig:fase_A-C1}, right side), the algorithm is in sub-phase $\Cdue$ and by move $m_5$ it rotates $r$ on $C(R)$ along the shortest path in such a way the composed triangle admits the required angles of $30^\circ$, $60^\circ$ and $90^\circ$ (see Figure~\ref{fig:fase_C2-D}, left side). In the latter case $r$ is moved inside $C(R)$ by means of sub-phase $\E$.

Once the robots in $\partial C(R)$ have formed the required triangle (case in which $F$ admits a multiplicity of $k$ elements in $c(F)$), sub-phase $\D$ can move  the $k-1$ robots closest to $c(R)$ toward it (cf move $m_6$ of Table~\ref{tab:F1-phases} and Figure~\ref{fig:fase_C2-D}). The specific triangle formed by the robots in $\partial C(R)$ assures that during sub-phase $\D$ the configuration remains always a leader configuration. 

The last sub-phase of $\Funo$ (sub-phase $\E$) is applied after creating a multiplicity of $k-1$ robots in $c(R)$ if $F$ requires $k$ elements in $c(F)$, and when exactly three robots are in $\partial C(R)$. Among such robots, there are two antipodal ones $g'$ and $g''$ plus a third one $r$. Robot $r$ is moved inside $C(R)$. Sub-phase $\E$ is logically divided into $\Euno$ and $\Edue$; this is because the same moves of $\Auno$ and $\Adue$ are used to perform the required task (cf Figure~\ref{fig:fase_C2-D} right side with Figure~\ref{fig:fase_E} left side).


\begin{figure*}[t]
\begin{center}
  \resizebox{0.53\textwidth}{!}{\input{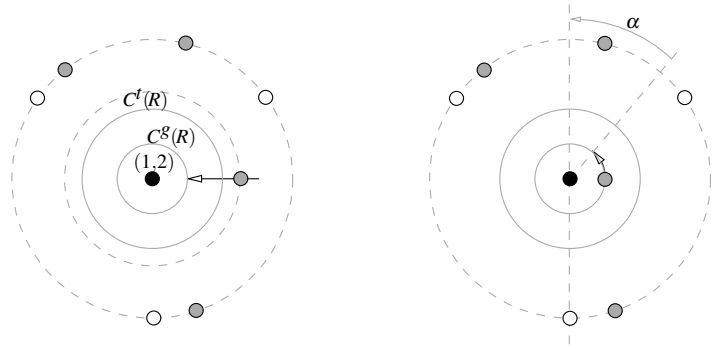}}
\caption{ Zooming on the circles $C_{\uparrow}^{i}(R)$, $0\le i\le 2$, of the configuration shown in Figure~\ref{fig:fase_E} (right side). On the left side, the guard circle $C^g(R)$ and the teleporter circle $C^t(R)$ are shown. On the right side, the configuration obtained at the end of sub-phase $\Guno$ after move $m_7$ lead a robot on the guard circle. The obtained configuration will be processed by the sub-phase $\Gdue$ to place the internal guard at the reference angle $\alpha$.}
\label{fig:fase_G}
\end{center}
\end{figure*}

\subsection{Phase $\Fdue$}
This phase is responsible for setting the internal guard $g$. In particular, $g$ is initially identified as the closest robot to $c(R)$ (excluding possible robots in $c(R)$ required to compose a multiplicity), and hence moved on the guard circle $C^g(R)$ (cf Definition~\ref{def:guard-disk}). Such a move is performed either by sub-phase $\G$ (logically divided into $\Guno$ and $\Gdue$) or by sub-phase $\H$. Table~\ref{tab:F2-phases} describes all such sub-phases, and also the corresponding invariants and moves. 

\noindent
\begin{table*}
\caption{ Invariants and moves for all the sub-phases of $\Fdue$.}
\label{tab:F2-phases}
\begin{center}
\begin{minipage}{.4\linewidth} 
\definecolor{grigio}{rgb}{0.9,0.9,0.9}
\bgroup
\rowcolors{1}{white}{grigio}
\def\arraystretch{1.4}
\setlength{\tabcolsep}{5pt}
\begin{center}
  \begin{tabular}{ | c | c | c | c |}
    \hline
    \textit{phase}  & \textit{start}  &  \textit{end} \\ \hline\hline
    $\Guno$ & $\GunoS$ & $\GunoE$ \\ \hline 
            & $m_7$ & $\Ftre_s \vee \Fquattro_s \vee \w$ \\ \hline
                        
    $\Gdue$ & $\GdueS$ &  \\ \hline 
            & $m_8$ & $\Ftre_s \vee \Fquattro_s$ \\ \hline         
                             
    $\H$    & $\HS$ & $\HE$ \\ \hline 
            & $m_9$ &  \\ \hline                            
                                        
  \end{tabular}
\end{center}
\egroup
\end{minipage}
\begin{minipage}{.50\linewidth}
\bgroup
\def\arraystretch{1.3}
\setlength{\tabcolsep}{5pt}
\begin{center}
  \begin{tabular}{ | c | p{0.9 \linewidth} | }
    \hline
    \textit{name}  &  \textit{description} \\ \hline\hline
    
    $m_7$ & 
Let $r$ be the robot on $C_{\uparrow}^{1}(R)$ of minimum view. $r$ moves  according to \SMove toward  $t= (r,c(R)]\cap  C^g(R)$  
     \\ \hline 
     
    $m_8$ & 
Let $r$ be the robot on $C^g(R)$. $r$ rotates along $C^g(R)$ toward the closest point $t$ such that $\angolo(g',c,t)$ (taken in any direction) is equal to $\alpha$
     \\ \hline 
     
    $m_9$ & 
Let $g’$ be the robot on $C(R)$ of minimum view. The robot in $c(R)$ moves toward a point $t$ on $C^g(R)$ such that $\angolo(g',c,t) = \alpha/2$.
     \\ \hline 
                       
   \end{tabular}
\end{center}
\egroup
\end{minipage}
\end{center}
\end{table*}

Sub-phase $\G$ starts when the invariant $\Guno_s = \GunoS$ holds, that is when $\gzero$ is false (i.e., the guard $g$ is not yet on $C^g(R)$) and $\c \Rightarrow \muno$ is true (i.e., if there is one robot on the center $\c(R) = c(F)$ then this is due to the presence of a multiplicity in $c(F)$). An example of such a case is described in Figure~\ref{fig:fase_E} (right side). Sub-phase $\Guno$ then repeatedly applies move $m_7$ (cf Table~\ref{tab:F2-phases}) to move the closest robot to $c(R)$ on $C^g(R)$ (cf Figure~\ref{fig:fase_G}). Note that, similarly to moves $m_1$ and $m_2$, move $m_7$ is performed by invoking Procedure \SMove.

When $\Guno$ is terminated, the sub-phase $\Gdue$ starts with the aim of rotating the robot $g$ placed on $C^g(R)$ so that $g$, $c(R)$, and one of the antipodal robots on $C(R)$, now detected as $g'$, form an angle equal to the reference angle $\alpha$ (cf Definition~\ref{def:alpha}). Move $m_8$ performs this task.
At the end of $\Gdue$, the three guards $g$, $g'$ and $g''$ compose the required common reference system for all robots. 

Sub-phase $\H$ handles the specific configurations in which the invariant $\H_s = \HS$ holds, that is the cases where there is one robot in $c(R)=c(F)$ and there is no multiplicity in $c(F)$. In such a cases, $m_9$ moves the robot away from the center $c(R)$, so that the obtained configuration is subsequently  managed by $\G$.


\begin{figure*}[t]
\begin{center}
  \resizebox{0.53\textwidth}{!}{\input{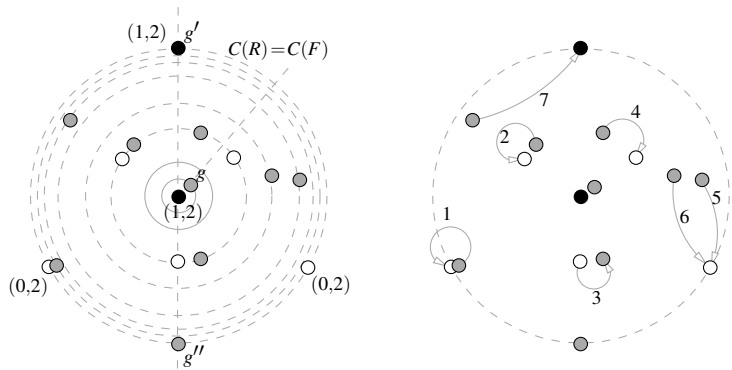}}
\caption{ On the left side, the whole configuration of Figure~\ref{fig:fase_G} (right side). On the right side, the mapping between non-guard robots and targets produced during phase $\Ftre$ (the numbers show the order in which the mapping is produced according to distances).}
\label{fig:fase_F3}
\end{center}
\end{figure*}

\subsection{Phase $\Ftre$}
This phase is responsible for moving all the non-guards robots (i.e., $n-3$ robots) toward the targets. It is composed of three sub-phases, as described in Table~\ref{tab:F3-phases}.

We remark that at the end of $\Fdue$, the required common reference system for all robots has been established (based on the three guards $g$, $g'$ and $g''$). This implies that all robots can now embed $F$ on $C(R)$. This embedding is obtained as follows: 
\begin{itemize}
\item as already observed, $C(F)$ is superimposed on $C(R)$;
\item the counter-clockwise direction for $R$ is assumed to be the one such that 
      $g$ becomes collinear with $g'$ and $c(R)$ by rotating of $\alpha$ degrees;
\item the counter-clockwise direction for $F$ is that defined in Section~\ref{sec:notation} for any multiset of points;
\item let $f'$ be a point in $\partial C(F)$ such that $f'$ is the second point 
      appearing in $V(f)$, being $f$ any point in $\minview(\partial C(F))$;
      $f'$ is superimposed on $g'$, and any other point in $F$ is superimposed
      such that the clockwise direction of $F$ coincides with that of $R$.
\end{itemize}
This embedding is shown in Figure~\ref{fig:fase_F3} (left side). According to this embedding, each robot uses the following mapping $\mu:\{g',g'',g\}\to F$ for determining the final target of each guard:
\begin{itemize}
\item $\mu(g')=f'$;
\item $\mu(g'')=f''$, where $f''$ is the point in $\partial C(F)$ closest to $g''$ 
      (in case of tie, that reachable from $g''$ in the counter-clockwise direction); 
\item $\mu(g)=f$, where $f\in F\setminus \{f',f''\}$ and if $c(F)\in F$ then  
      $f=c(F)$, else $f$ is the point in $\partial C_{\uparrow}^{1}(F)$ closest to $g$ 
      (in case of tie, the first of such points reached by $\halfline(c(R),g)$ when 
      this half-line is turned counter-clockwise around the center).
\end{itemize}

This embedding is maintained along all phase $\Ftre$; we remark that $g$, $g'$ and $g''$ are not moved during this phase. Any other robot, one by one, is moved toward its closest point of $F \setminus \{\mu(g), \mu(g'), \mu(g'')\}$ (see Figure~\ref{fig:fase_F3}, right side). At any time, each robot must determine (1) whether it is already on its target or not (i.e., whether it is \emph{matched} or not), (2) if it is not matched, which is its target, and (3) whether it is its turn to move or not. To this aim, each robot computes the following data:
\begin{itemize}
\item  the sets of matched robots and matched targets, that is 
       $R^m=R\cap(F\setminus \{\mu(g), \mu(g'), \mu(g'')\})$ and 
       $F^m= F\cap R^m$;
\item  the sets of unmatched robots and unmatched targets, that is 
       $R^{\neg m}= R\setminus (R^m \cup \{g, g',g''\})$ and 
       $F^{\neg m}=F\setminus (F^m\cup \{\mu(g), \mu(g'), \mu(g'')\})$;
\item  the minimum distance between unmatched robots and unmatched targets, 
       that is $\eta= \min \{ d(r,f):~ r\in R^{\neg m},~ f\in F^{\neg m} \}$;
\item  the set of unmatched robots at minimum distance from unmatched targets, 
       that is $R^{\neg m}_{\eta} = \{r\in R^{\neg m}:~ d(r,F^{\neg m}) = \eta \}$;
\item  in a given turn, which is the robot $r$ that has to move toward its target, 
       that is $r = \minview( R^{\neg m}_{\eta} )$, and which is the corresponding 
       target $\mu(r)$, that is any point in $\{ f\in F^{\neg m}:\; d(r,f)=\eta\}$.
\end{itemize}  

According to the strategy described in Section~\ref{sec:strategy}, the robot that moves toward its target has two constraints: (1) avoiding undesired collisions (in case the trajectory meets an already matched robot), and (2) avoiding entering into the guard circle (to preserve the common reference system). 

For the former constraint, a Procedure \DistMin\ is designed. The procedure is given in Algorithm~\ref{alg:dist_min} while its description can be found in the corresponding correctness proof provided in Lemma~\ref{lem:corr-DistMin}.

 For the latter, in case the segment $[r,\mu(r)]$ meets the teleporter circle $C^t(R)$, then an alternative trajectory is computed. For this computation, robots also need the following data: the set $F^*=\{f\in F^{\neg m}~:~ d(f,c(F)) ~\text{is minimum} \}$.
\begin{figure}[h]
\begin{center}
  \resizebox{0.23\textwidth}{!}{\input{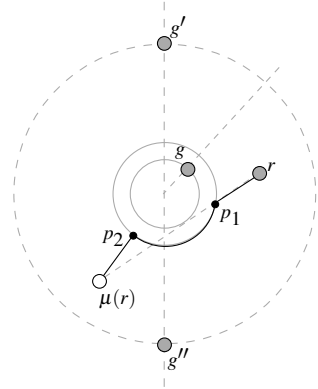}}
\caption{ Visualization of a robot trajectory through the teleporter circle. The robot traces the path represented by the black polygonal curve consisting of two segments and one arc.}
\label{fig:teleporter}
\end{center}
\end{figure}
The new trajectory is divided into three parts (see Figure~\ref{fig:teleporter}): 
\begin{enumerate}
\item[(a)] a collision free trajectory toward the closest point $p_1$ on the teleporter circle, 	
		   that is $p_1 = (r,\mu(r)] \cap C^t(R)$;
\item[(b)] a rotation along $C^t(R)$ toward the closest point 
           $p_2 =[c(F),f^*]\cap C^t(R)$, where $f^* \in F^*$;
\item[(c)] a collision free trajectory toward $f^*$. 
\end{enumerate}
Note that in case (c), the destination $f^*$ may differ from the original destination $\mu(r)$ computed in case (a), but this is not a problem since, as shown in the correctness section, no other robot is moved until $r$ reaches its final destination.

\noindent
\begin{table*}
\caption{ Invariants and moves for all the sub-phases of $\Ftre$.}
\label{tab:F3-phases}
\begin{center}
\begin{minipage}{.35\linewidth} 
\definecolor{grigio}{rgb}{0.9,0.9,0.9}
\bgroup
\rowcolors{1}{white}{grigio}
\def\arraystretch{1.4}
\setlength{\tabcolsep}{5pt}
\begin{center}
  \begin{tabular}{ | c | c | c | c |}
    \hline
    \textit{phase}  & \textit{start} &  \textit{end} \\ \hline\hline
    $\M$ & $\MS$ & $\ME$ \\ \hline 
            & $m_{10}$ &  \\ \hline
            
    $\N$ & $\NS$ & $\NE$ \\ \hline 
            & $m_{11}$ &  \\ \hline
            
    $\O$ & $\OS$ & $\OE$ \\ \hline 
            & $m_{12}$ &  $\Fquattro_s$ \\ \hline                                                                
                                        
  \end{tabular}
\end{center}
\egroup
\end{minipage}
\begin{minipage}{.55\linewidth}
\bgroup
\def\arraystretch{1.3}
\setlength{\tabcolsep}{5pt}
\begin{center}
  \begin{tabular}{ | c | p{0.90 \linewidth} | }
    \hline
    \textit{name}  &  \textit{description} \\ \hline\hline
     
    $m_{10}$ & 
    The unique robot $r$ on $C^t(R)$ rotates toward the closest point $p_2 =[c(F),f^*]\cap C^t(R)$, where $f^*\in F^*$.
     \\ \hline 
     
    $m_{11}$ & 
    Robot $r = \minview( R^{\neg m}_{\eta} )$ moves according to \DistMin\ toward the closest point $p_1$ on the teleporter circle, that is $p_1 = (r,\mu(r)] \cap C^t(R)$;.
     \\ \hline 
     
    $m_{12}$ & 
    Robot $r = \minview( R^{\neg m}_{\eta} )$ moves according to \DistMin\  toward $\mu(r)$.
      \\ \hline 
                    
   \end{tabular}
\end{center}
\egroup
\end{minipage}
\end{center}
\end{table*}

The algorithm (see Table~\ref{tab:F3-phases}) performs such a new trajectory by moving robots along the teleporter circle as follows. Sub-phase $\M$ concerns case (b) above, that is configurations fulfilling invariant $\M_s = \MS$ (i.e., configurations where there is a unique robot $r$ on the $C^t(R)$ but $r$ does not coincide with a point $p =[c(F),f^*]\cap C^t(R)$, where $f^* \in F^*$). Then, move $m_{10}$ rotates $r$ toward such a point $p$. Sub-phase $\N$ concerns case (a) above. Via move $m_{11}$ (that invokes \DistMin), $\N$ is in charge of leading $r$ on a point on $C^t(R)$ if needed (cf $\ddue$ in $\N_s = \NS$). Finally, sub-phase $\O$ is concerned with case (c) above, that is when the trajectory from $r$ to its target $\mu(r)$ does not meet the teleporter circle. Move $m_{12}$ (that invokes \DistMin) performs this final task.

\begin{algorithm}[ht]{{
\SetKwComment{Comment}{/*}{*/}
\SetKwInput{Proc}{Procedure}
\Proc{\DistMin}
\SetKwInOut{Input}{Input}
\Input{A target $f$.}
\BlankLine
    \If{there are no robots between $r$ and $f$}{move toward $f$\label{dm:move1}}
    \Else{
            $\bar{r}=\argmin_{x } \{ d(r,x):~ x \in R\cap [r,f] \}$ \label{dm:f-bar}\;
let $\ell$ be one of the half-lines starting from $\bar{r}$, perpendicular to $[r,f]$, and on an open half-plane that does not include $c(R)$\label{dm:ell2}\;	
            $P =\{ p = \ell\cap \halfline(r,x):~ p\neq \bar{r} ~\mbox{and}~ x\in R\setminus\{r\}\}$
             \label{dm:P}\;
            $p'=\ell \cap C(R)$ \label{dm:P1}\;
            let $p''$ be the intersection between $\ell$ and the circle centered in $f$ 
            of radius $[f,r]$\label{dm:P2}\;
            let $p'''$ be the intersection, if it exists, between $\ell$ and a side of 
            the cell of the Voronoi diagram induced by $F^{\neg m}$ where $f$ 
            lies\label{dm:P3}\;
            $\bar{p} = \argmin_{x} \{ d(\bar{r},x):~ x\in P\cup \{p',p'',p'''\} \}$ 
            \label{dm:p-bar}\;
            let $p$ be the median point in $[\bar{r},\bar{p}]$ \label{dm:p}\;
            move toward $p$ \label{dm:move2}\;
    }
}}
\caption{ Procedure \DistMin\ performed by any robot $r$ when moves $m_{11}$ or $m_{12}$ must be executed.}
\label{alg:dist_min}
\end{algorithm}

\begin{lemma}\label{lem:corr-DistMin}
Procedure $\DistMin$ performed by a robot $r$ with input a target $f$ always moves $r$ avoiding collisions with other robots either toward $f$, if there are no robots between $r$ and $f$, or toward a point $p$ fulfilling the following conditions:
\begin{enumerate}
 \item $p$ is inside both $C(R)$ and the cell $D_{f}$ of the Voronoi diagram induced by $F^{\neg m}$ where $f$ lies;
 \item there is no robot between $p$ and $f$;
 \item $d(f,p)<d(f,r)$.
\end{enumerate}
 Moreover, all the points $x$ reached by $r$ during its movement share the same properties of $p$.
\end{lemma}
\begin{proof}
 At Line~\ref{dm:move1}, Procedure $\DistMin$ moves $r$ toward $f$ when there are no robots between $r$ and $f$. As the movement is straightforward and since both  $D_{f}$ and the disk enclosed by $C(R)$ are convex, all the points $x$ reached by $r$ during its movement are inside $C(R)$ and $D_{f}$.
 
 If there are robots between $r$ and $f$, among such robots the procedure, at Line~\ref{dm:f-bar}, identifies as $\bar{r}$ the closest to $r$. The point $p$ is calculated on one of the two half-lines perpendicular to $[r, f]$ in $\bar{r}$, in accordance to the position of $c(R)$, see Line~\ref{dm:ell2}.
 On $\ell$, a set $P =\{ p = \ell\cap \halfline(r,x):~ p\neq \bar{r} ~\mbox{and}~ x\in R\setminus\{R\} \}$ is calculated at Line~\ref{dm:P}. The target $p$ is different by any point in $P$: being these points on the lines between $r$ and any another robot, this will assure that the movement will be free by further collisions.
 To set the exact position of $p$ on $\ell$, three other points are calculated at Lines~\ref{dm:P1},~\ref{dm:P2}, and~\ref{dm:P3}. The first one is $p'$, that is the intersection of $\ell$ and $C(R)$. The target $p$ is such that $d(\bar{r},p)< d(\bar{r},p')$, then all the points $x$ reached by $r$ during its movement are inside $C(R)$. The second one is $p''$, that is the intersection between $\ell$ and the circle centered in $f$ 
 of radius $[f,r]$. The target $p$ is such that $d(\bar{r},p)< d(\bar{r},p'')$, then all the points $x$ reached by $r$ during its movement are closer and closer to $f$.  The third one, if exists, is $p'''$, that is the intersection of $\ell$ and a side of cell $D_{f}$. The target $p$ is such that $d(\bar{r},p)< d(\bar{r},p''')$, then all the points $x$ reached by $r$ during its movement are inside this cell. This assures that the points $x$ reached by $r$ are closer to $f$ than any other target. Moreover, there is no robot between $x$ and $f$. In fact, if such robot $r'$ exists, the line $\halfline(r,x)$ would intersect $\ell$ in a point closest to $\bar{r}$ than each point in $P$, a contradiction as the target $p$ is the closer one. 
 
Finally, at Line~\ref{dm:p}, the point $p$ is set at a position fulfilling all the above constraints. In addition, notice that the choice of half-line $\ell$ ensures that $p$ cannot be on or inside $C^t(R)$; this implies that the procedure can be safely applied. In turn, all such properties prove that the claim holds.
\qed 
\end{proof}

\subsection{Phase $\Fquattro$}\label{ssec:f4}

This phase concerns the finalization steps, where the last three robots (the guards) must be moved to their targets to complete the formation of pattern $F$. In particular, since in this phase we use the embedding defined in phase $\Ftre$, $g'$ is already matched and hence at most $g$ and $g''$ remain to be moved. Moving the guards leads to the loss of the common reference system, and hence ad-hoc moves must be designed to complete the pattern. 

The algorithm for this phase is composed of two sub-phases denoted as $\P$ and $\Q$  (see Table~\ref{tab:F4-phases}).  The former handles the majority of configurations by moving first $g''$ and then $g$ toward the respective targets. The latter manages some special cases where $g''$ is critical for $C(R)$. In particular, predicate $\q$ (informally introduced in Table~\ref{tab:basic-variables}) is used to characterize such special cases. Formally:

\noindent
\begin{table*}
\caption{ Invariants and moves for all the sub-phases of $\Fquattro$.}
\label{tab:F4-phases}
\begin{center}
\begin{minipage}{.30\linewidth} 
\definecolor{grigio}{rgb}{0.9,0.9,0.9}
\bgroup
\rowcolors{1}{white}{grigio}
\def\arraystretch{1.4}
\setlength{\tabcolsep}{5pt}
\begin{center}
  \begin{tabular}{ | c | c | c | c |}
    \hline
    \textit{phase}  & \textit{start}  &  \textit{end} \\ \hline\hline
    $\Puno$ & $\PunoS$ & $\PunoE$ \\ \hline 
            & $m_{13}$ &  \\ \hline
            
    $\Pdue$ & $\PdueS$ & \\ \hline 
            & $m_{14}$ &  $\w$ \\ \hline
            
    $\Quno$ & $\QunoS$ & $\QunoE$ \\ \hline 
            & $m_{15}$ &  \\ \hline
            
    $\Qdue$ & $\QdueS$ & $\QdueE$ \\ \hline 
            & $m_{16}$ &  \\ \hline
            
    $\Qtre$ & $\QtreS$ & $\QtreE$ \\ \hline 
            & $m_{17}$ &  $\w$ \\ \hline                                                              

    $\Qquattro$ & $\QquattroS$ & \\ \hline 
            & $m_{13}$ & $\w$ \\ \hline                                                              
                                               
  \end{tabular}
\end{center}
\egroup
\end{minipage}
\begin{minipage}{.60\linewidth}
\bgroup
\def\arraystretch{1.3}
\setlength{\tabcolsep}{5pt}
\begin{center}
  \begin{tabular}{ | c | p{0.90 \linewidth} | }
    \hline
    \textit{name}  &  \textit{description} \\ \hline\hline

    $m_{13}$ & 
$g''$ rotates toward $\mu(g'')$. 
     \\ \hline 
     
    $m_{14}$ & 
$g$ moves toward $\mu(g)$. 
     \\ \hline 
     
    $m_{15}$ & 
$g$ moves toward $C(R)\cap \halfline(c(R),g)$.
     \\ \hline 

    $m_{16}$ & 
$g''$ rotates along $C(R)$ toward the closest point among $\mu(g'')$ and the antipodal point to $g$. 
     \\ \hline 
            
     $m_{17}$ & 
$g$ rotates toward $\mu(g)$. 
     \\ \hline  
                              
   \end{tabular}
\end{center}
\egroup
\end{minipage}
\end{center}
\end{table*}

\begin{definition}\label{def:q}
Let $F=\{f_1,f_2,\ldots,f_n\}$ be a pattern to be formed. We say that the predicate $\q$ holds if $F$ fulfills the following conditions:
\begin{enumerate}
\item $\partial C(F) = F$;
\item $F$ does not contain multiplicities;
\item assuming $f_1 = \minview(F)$ and $(f_1,f_2,\ldots,f_n)$ as the counter-clockwise sequence of points on $C(F)$, then 
$\angolo(f_n,c(F),f_2) > 180^{\circ}$, where such an angle is obtained by rotating  $\halfline(c(F),f_n)$ counter-clockwise.
\end{enumerate}
\end{definition}
Figure~\ref{fig:fase_Q}.$(a)$ shows an embedding of a pattern $F$ in which $\q$ holds: note that $r_1$ is $g$, $r_2$  is $g''$ and $f_2$ is matched with $g'$. 

Sub-phase $\P$ is divided into $\Puno$ and $\Pdue$ to manage the moves of $g''$ and $g$, respectively. We now describe each sub-phase: its invariant, the task it performs, and the corresponding move. 

In $\Puno$, $g''$ rotates along $C(R)$ toward $\mu(g'')$. Each robot can recognize this phase by performing, 
in order, the following steps:
\begin{enumerate}
\item test whether $\guno$ holds; 
\item if the previous test is passed, it uses the same embedding defined in phase 
      $\Ftre$: this embedding allows to recognize the guard $g'$, and, in turn, 
      to determine $g''$ as the robot on $C(R)$ closest to the antipodal point of $g'$;
\item\label{p1-test}  finally, it tests whether such an embedding makes 
      $R\setminus \{g,g''\}$ similar to $F\setminus \{\mu(g),\mu(g'')\}$ and $g''\not = \mu(g'')$, 
      where also the targets $\mu(g)$ and $\mu(g'')$ are those defined in phase $\Ftre$.
\end{enumerate}
The result of test at Item~\ref{p1-test} above can be seen as the value of an invariant $\iuno$ (cf phase $\Puno$ at Table~\ref{tab:F4-phases}). When a robot checks that $\iuno$ holds and recognizes itself as $g''$, it simply applies move $m_{13}$ to complete the rotation along $C(R)$ to reach its target $\mu(g'')$.
For instance, referring to Figure~\ref{fig:fase_F3} right side, once all non-guard robots are correctly placed during $\Ftre$, the configuration belongs to $\Puno$, and $g''$ rotates in the clockwise direction along $C(R)$ to compose the multiplicity on the left. This is in fact the closest point on $C(R)$ to $g''$ (in the clockwise direction as there is a tie to break) with respect to the defined embedding of $F$.
Once also $g''$ is correctly positioned on its target, as we are going to see, $\idue$ holds and only $g$ remains to move toward $\mu(g)$ to finalize $F$. In the specific example of  Figure~\ref{fig:fase_F3}, $\mu(g)=c(R)$.
 
The movement of $g$ is realized in $\Pdue$ via a straight move of $g$ toward $\mu(g)$. 
Each robot can recognize this phase by performing, in order, the following steps:
\begin{enumerate}
\item compute the set $E = \{ (r,f):~(\{r\}=c(R) \vee \{r\}=\partial C^1_{\uparrow}(R)) \wedge 
      R\setminus \{r\} \mbox{ is similar to } F\setminus \{f\} \}$; 
\item\label{p2-test} test whether there exists a pair $(r,f)\in E$ such that 
      $f=c(F) \vee (c(F)\not \in F) \wedge f \in C^1_{\uparrow}(F) \wedge  d(c(R),r) < d(c(F),f)$;
\item if there are many pairs $(r,f)$ that fulfill the previous test, it selects 
      one for which $d(r,f)$ is minimum;
\item the pair $(r,f)$ selected in the previous step allows robots to recognizes 
      the internal guard $g$ (i.e., $g$ coincides with $r$) and the target 
      of $g$ (i.e., the point $f$). 
\end{enumerate}
The result of test at Item~\ref{p2-test} above can be seen as the value of an invariant $\idue$. When a robot checks that $\idue$ holds and recognizes itself as $g$, it simply applies move $m_{14}$ to complete the movement toward the remaining unmatched target $f$. 

\smallskip
To ensure a correct finalization even when $\q$ holds and hence when $g''$ might be critical for $C(R)$, the sub-phase $\Q$ is divided into four sub-phases. They are responsible for: 
\begin{itemize}
\item $\Quno$: moving $g$ radially toward $C(R)$; 
\item $\Qdue$: rotating $g''$ of at most $\alpha$ along $C(R)$; 
\item $\Qtre$: rotating $g$ from the position acquired at $\Quno$ along $C(R)$ toward its target; 
\item $\Qquattro$: rotate again $g''$ if necessary  along $C(R)$ toward its target.
\end{itemize}
The movement of $g$ toward its target performed in two steps guarantees to avoid possible symmetries.
We now describe each sub-phase: its invariant, the task it performs, and the corresponding move. To recognize each phase among $\Quno,\ldots,\Qquattro$, robots perform the following test:
\begin{itemize}
\item[1)] test whether there exists an embedding of $F$ such that $r_2,r_3,\ldots,r_{n-1}$ are matched with $f_2,f_3,\ldots,f_{n-1}$, respectively. We recall that here predicate $\q$ holds, hence points in $F$ fulfill all the conditions in Definition~\ref{def:q} (see Figure~\ref{fig:fase_Q}.$(a)$).
\end{itemize}

\begin{figure*}[h]
\begin{center}
  \resizebox{0.99\textwidth}{!}{\input{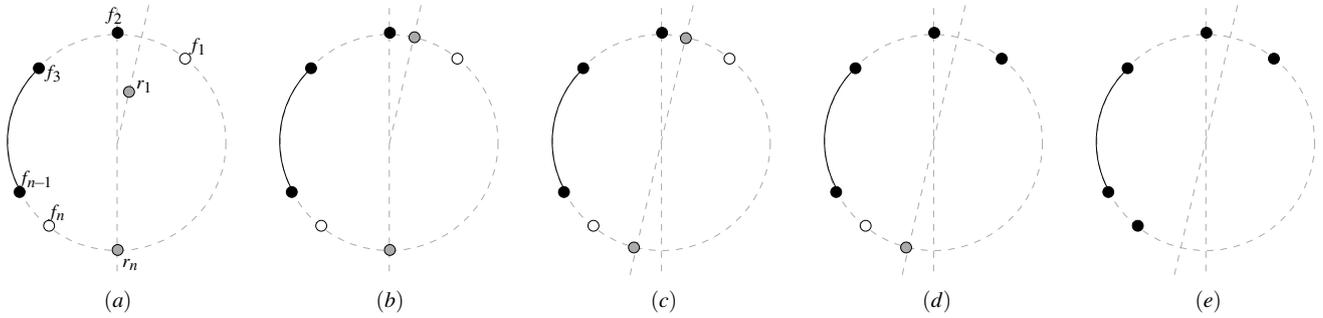}}
\caption{ Visualization of some configurations belonging to $\Q$ (points $f_4, \ldots, f_{n-1}$, if any, all lie on the black arc). In (a), a configuration where $\Quno_s$ holds, in (b) $\Qdue_s$ holds, in (c) $\Qtre_s$ holds, in (d) $\Qquattro_s$ holds, and finally in (e) $\w$ holds. }
\label{fig:fase_Q}
\end{center}
\end{figure*}

\noindent
As we are going to show in the correctness section, in order to fulfill the required conditions $R$ must is asymmetric. Hence the above test always determines a unique ordering for the robots. After, robots perform two additional tests, according to the specific sub-phases to be recognized. 
%
Concerning $\Quno$, the following additional tests are needed:
\begin{itemize}
\item[2)] test whether both $\partial C^1_{\uparrow}(R) = \{r_1\}$ and
          $\angolo(r_1,c(R),f_2)=\alpha$ hold; 
\item[3)] test whether $r_n$ is antipodal to $r_2$ (which coincides with $f_2$).
\end{itemize}
If all the previous tests are passed, robots recognize $g$ as $r_1$, $g'$ as $r_2$, and $g''$ as the antipodal robot to $g'$. The result of such a process can be seen as the value of an invariant $\itre$. When a robot recognizes itself as $g$ and checks that $\q \wedge \itre$ holds, it simply applies move $m_{15}$ to move radially toward $C(R)$ (cf cases $(a)$ and $(b)$ of Figure~\ref{fig:fase_Q}). 

Concerning $\Qdue$, the following additional tests are needed:
\begin{itemize}
\item[2)] test whether $r_1$ is on $C(R)$ between $f_1$ and $f_2$ such that 
          $\angolo(r_1,c(R),f_2)=\alpha$;
\item[3)] test whether $r_n$ is on $C(R)$ such that $r_n\neq f_n$ and 
          $\angolo(r_n,c(R),p) < \alpha$, where $p$ is the antipodal point 
          to $f_2$.
\end{itemize}
If all the previous tests are passed, robots recognize $g$ as $r_1$, $g'$ as $r_2$, and $g''$ as the closest robot to $p$. The result of such a process can be seen as the value of an invariant $\iquattro$. When a robot checks that $\q \wedge \iquattro$ holds and recognizes itself as $g''$, it applies move $m_{16}$ to rotate of at most an angle $\alpha$ from $p$ or to reach its target $f_n$ if residing before $p$. This is done to maintain $C(R)$, being $g''$ critical (cf cases $(b)$ and $(c)$ of Figure~\ref{fig:fase_Q}).

Now, notice that according to the definition of view of $F$ given in Section~\ref{ssec:view} and according to Definition~\ref{def:q}, $\angolo(f_1,c(F),f_2)=3\alpha$. Hence, for sub-phase $\Qtre$ the following additional tests are needed:
\begin{itemize}
\item[2)] test whether $r_1$ is on $C(R)$ between $f_1$ and $f_2$ such that 
          $\alpha \le \angolo(r_1,c(R),f_2) < 3\alpha$;
\item[3)] test whether $r_n$ is on $C(R)$ such that $r_n = f_n$ or
          $\angolo(r_n,c(R),p)=\alpha$, where $p$ is the antipodal point to $f_2$.
\end{itemize}
If all the previous tests are passed, robots recognize $g$ as $r_1$, $g'$ as $r_2$, and $g''$ as the closest robot to $p$. The result of such a process can be seen as the value of an invariant $\icinque$. When a robot checks that $\q \wedge \icinque$ holds and recognizes itself as $g$, it applies move $m_{17}$ to complete the rotation along $C(R)$ to reach its target $f_1$  (cf cases $(c)$ and $(d)$ of Figure~\ref{fig:fase_Q}).

Finally, for $\Qquattro$ the following additional tests are needed:
\begin{itemize}
\item[2)] test whether $r_1$ is matched with $f_1$, according to the embedding 
          defined at test 1;
\item[3)] test whether $r_n$ is on $C(R)$ such that $r_n\neq f_n$ and 
          $\alpha \le \angolo(r_n,c(R),p) < \angolo(f_n,c(R),p)$, where $p$ 
          is the antipodal point to $f_2$.
\end{itemize}
If all the previous tests are passed, robots recognize $g$ as $r_1$, $g'$ as $r_2$, and $g''$ as the closest robot to $p$. The result of such a process can be seen as the value of an invariant $\isei$. When a robot checks that $\q \wedge \isei$ holds and recognizes itself as $g''$, it applies move $m_{13}$ to reach $f_n$  (cf cases $(d)$ and $(e)$ of Figure~\ref{fig:fase_Q}).


\subsection{Phase $\Fcinque$}
As remarked in Section~\ref{sec:strategy}, to solve \apf it is necessary to face the  $\SB$ sub-problem, that is to break possible symmetries in the initial leader configuration provided to the algorithm. We recall that the initial symmetric configurations handled by the algorithm consist of any configuration $R$ without multiplicities fulfilling one of the following conditions:
\begin{itemize}
\item[(1)] there exists a unique reflection axis $\ell$ for $R$ such that $|R\cap \ell| \ge 1$;
\item[(2)] there exists a rotational symmetry in $R$ and $c(R)\in R$.
\end{itemize}
The strategy used to address this sub-problem is to carefully move one robot away from the axis (if case (1) occurs) or the robot away from the center (if case (2) occurs), in order to obtain a stationary asymmetric configuration. In fact, during this phase multiplicities are not created. In particular, configurations with a reflection axis are always preliminary transformed by moving one robot along the axis $\ell$ toward $c(R)$, if possible; otherwise (cf Figure~\ref{fig:controesempio-fr}), a robot on $\ell$ is moved `sufficiently faraway' along $\ell$ and then rotated on $C(R)$.
Whereas, if $c(R)$ is reached or case (2) occurs, the robot in $c(R)$ is moved radially away from the center. 
In all the cases, in order to maintain stationarity, we aim to obtain a configuration different from those addressed in the other phases $\Funo, \ldots,\Fquattro$ until the moving robot does not reach a specific target.  

The strategy for breaking the symmetries presents some complexities that lead the algorithm for this phase to be composed of many different moves. As in the previous phases, each of such moves is referred to a sub-phase; Table~\ref{tab:F5-phases} describes all such sub-phases, and also the corresponding invariants and moves. We now provide a description of each sub-phase.
%
\noindent
\begin{table*}
\caption{ Invariants and moves for all the sub-phases of $\Fcinque$. 
}
\label{tab:F5-phases}
\begin{center}
\begin{minipage}[t]{.40\linewidth} 
\definecolor{grigio}{rgb}{0.9,0.9,0.9}
\bgroup
\rowcolors{1}{white}{grigio}
\def\arraystretch{1.4}
\setlength{\tabcolsep}{5pt}
\begin{center}
  \begin{tabular}{ | c | c | c | c |}
    \hline
    \textit{phase}  & \textit{start}  &  \textit{end} \\ \hline\hline
    $\T$ & $\TS$ & $\TE$ \\ \hline 
            & $m_{18}$ & $\w$ \\ \hline
            
    $\U$ & $\US$ &  $\UE$ \\ \hline 
            & $m_{19}$ & $\Funo_s\vee \w$ \\ \hline
                        
    $\Vuno$ & $\VunoS$ &  $\VunoE$ \\ \hline 
            & $m_{20}$ & $\w$ \\ \hline
                    
    $\Vdue$ & $\VdueS$ & $\VdueE$ \\ \hline 
            & $m_{21}$ & $\w$ \\ \hline
            
    $\Vtre$ & $\VtreS$ & $\VtreE$ \\ \hline 
            & $m_{22}$ & $\w$ \\ \hline    
            
    $\Vquattro$ & $\VquattroS$ & $\VquattroE$ \\ \hline 
            & $m_{23}$ & $\w$ \\ \hline  

    $\W$ & $\WS$ & $\WE$ \\ \hline 
            & $m_{24}$ & $\Funo_s \vee \Fdue_s \vee \Fquattro_s$ \\ \hline                                                            
                                            
  \end{tabular}
\end{center}
\egroup
\end{minipage}
\hspace{2mm}
\begin{minipage}[c]{.50\linewidth}
\bgroup
\def\arraystretch{1.3}
\setlength{\tabcolsep}{5pt}
\begin{center}
  \begin{tabular}{ | c | p{0.9 \columnwidth} | }
    \hline
    \textit{name}  &  \textit{description} \\ \hline\hline
    $m_{18}$ & 
    Let $r$ be the only robot on the axis of symmetry $\ell$. 
    If $\exists$ a rotational-free path for $r$ then 
    $r$ moves toward $c(R\setminus \{r\})$, 
    else $r$ moves toward the closest point $t\in \ell$ such that 
    $R\setminus\{r\}\cup\{t\}$ is a faraway configuration. \\ \hline 

    $m_{19}$ & 
    the faraway robot $r$ rotates along $C(R)$ toward the point $t$ defined in predicate $\zuno$.
    \\ \hline     

    $m_{20}$ & 
    If $\exists$ a rotational-free path for $\robotuno$ then $\robotuno$ moves toward $c(R)$; concurrently, $\robotdue$ moves toward a 
    point that makes predicate $\uuno$ true, without swapping its role with $\robotuno$ nor crossing $C_\downarrow^j(R)$, for any $j$. \\ \hline 

    $m_{21}$ & 
    If $\exists$ a rotational-free path for $\robotuno$ then $\robotuno$ moves toward $c(R)$; Else let $C_\downarrow^i(R)$ be the circle where $\robotuno$ resides, then $\robotuno$ moves toward a point $t$ that halves its distance from $C_\downarrow^{i+1}(R)$.
If $\robotdue \in [t^{60},t^{55})$, then it moves toward $t^x$ if such a point exists else toward $t^{55}$.
    \\ \hline 
    
    $m_{22}$ & 
    $\robotuno$ moves toward $c(R)$. 
    \\ \hline 

    $m_{23}$ & 
    Let $r\in\{\robotuno,\robotdue\}$ be the robot closest to $C(R\setminus \{\robotuno,\robotdue\})$, $\robotuno$ in case of ties. Let $t$ be the closest point to $r$ on $C(R\setminus \{r\} \cup t ) \cap \ell$ with $t$ not critical for $C(R\setminus \{r\} \cup t ) \cap \ell$. $r$ moves toward $t$.
\\ \hline 

    $m_{24}$ & Let $r$ be the robot closest to $c(R)$, and $t$ be a point on $C^g(R)$ s.t. $t$ is not on an axis of symmetry and $\angolo(t,c(R),r')\neq \alpha$ for any $r'\in \partial C(R)$.
        $r$ moves radially toward $t$.
    \\ \hline                           
   \end{tabular}
\end{center}
\egroup
\end{minipage}
\end{center}
\end{table*}

Sub-phase $\T$ concerns any configuration $R$ where predicate $\TS$ holds. Informally, this means that $R$ is symmetric (cf $\bzero$), there is no robots in $c(R)$ (cf $\neg\c$) and hence $R$ must admit one single axis $\ell$ of reflection with robots on it. Actually, there exists a unique robot $r$ on $\ell$ (cf $\buno$),  and no robot has started to move away from $\ell$ (cf $\zuno$). Figure~\ref{fig:fase_F5}.(a) shows an example for such a configuration $R$. Move $m_{18}$ moves $r$ toward $c(R\setminus \{r\})$ if possible. This movement can be performed only when $r$ admits a \emph{rotational-free path} toward $c(R\setminus \{r\})$. Such a path is defined as follows:
\begin{itemize}
\item
given a robot $r\in \ell$, if in the segment $(r,c(R\setminus \{r\}))$ there are no robots and there is no point $t$ such that $R\setminus\{r\}\cup\{t\}$ has a rotational symmetry, then we say that there exists a rotational-free path for $r$ toward $c(R\setminus \{r\})$.
\end{itemize}
If it is not possible for $r$ to reach $c(R\setminus \{r\})$, that is $r$ does not admit a rotational-free path, then move $m_{18}$ moves $r$ along $\ell$ so that a \emph{faraway configuration} $R'$ is created. 

A configuration $R$ is said to be a faraway configuration if in $R$ the following conditions hold:

\begin{itemize}
	\item $\stre$, that is $|\partial C(R)|=3$;
	\item there exists a robot $r$ among the three on $C(R)$ such that $R\setminus \{r\}$ admits one axis of reflection $\ell$ that reflects to each other the other two robots on $C(R)$, referred to as $r_1$ and $r_2$;
	\item $r_1$ and $r_2$ are the furthest robots from $r$;
	\item $\angolo(r_1,r,r_2)\leq 59^\circ$.
\end{itemize}

In such a case, robot $r$ is said to be a \emph{faraway robot}.

Concerning the configuration shown in Figure~\ref{fig:fase_F5}.(a), since $r$ does not admit a rotational-free path, then move $m_{18}$ moves $r$ to create a faraway configuration $R'$. In Theorem~\ref{lem:corr-F5}, it is shown that $r$ becomes a faraway robot and that $\partial C(R')=\{r, r_1,r_2\}$. 

Sub-phase $\U$ handles any faraway configuration $R$. These configurations are characterized by predicate $\zuno$. In particular, since $\partial C(R)=\{r, r_1,r_2\}$, with $r$ being a faraway robot, then move $m_{19}$ moves robot $r$ along $C(R)$ so that eventually $r$ is antipodal to either $r_1$ or $r_2$. It is worth to note that as soon as $r$ leaves the axis $\ell$, the obtained configuration is no longer symmetric. Anyway, predicate $\zuno$ characterizes not only faraway configurations that are symmetric, but also configurations in which robot $r$ is stopped before reaching its target.
\begin{figure*}[h]
\begin{center}
  \resizebox{0.80\textwidth}{!}{\input{fase_F5-2}}
\caption{ Visualization of some configurations belonging to $\Fcinque$. In (a), a configuration where $\T_s$ holds (notice that it corresponds to the counter-example provided in Section~\ref{ssec:ce-3}), in (b) and (c) $\Vuno_s$ holds, in (d) and (e) $\Vdue_s$ holds, in (f) $\Vtre_s$ holds, and in (g) $\Vquattro_s$ holds.
}
\label{fig:fase_F5}
\end{center}
\end{figure*}

Sub-phases $\Vuno, \ldots,\Vquattro$ are responsible for transforming configurations with one axis of reflection and at least two robots on it into configurations having one robot in $c(R)$. In these sub-phases we use the following additional notation. Let $R$ be a configuration with one axis of reflection $\ell$ such that $|R\cap \ell|\ge 2$. Then: 
\begin{itemize}
\item
	It is possible to total order the elements in $R\cap \ell$ by exploiting their distance from $c(R)$ and, in case of ties, by giving priority to robots being not critical, to those admitting rotational-free paths, and then of minimum view;
\item	
	According to the above ordering, we denote by $\robotuno$ the first robot, and by $\robotdue$ the second robot. For instance, robot $\robotuno$ is the only robot moved by $m_{22}$ in sub-phase $\Vtre$ (cf Figure~\ref{fig:fase_F5}.(f)).  Robots $\robotuno$ and $\robotdue$ will be possibly moved concurrently by move $m_{20}$ in sub-phase $\Vuno$; 
	
\item
When $|R\cap \ell|= 2$ and exactly one robot in $R\cap \ell$ is critical (i.e., predicate $\btre$ holds - cf Figure~\ref{fig:fase_F5}.(e)), we need some additional notation. Note that the unique critical robot on $\ell$ is by definition $\robotdue$. 
Let $r'$, $r''$ be the closest robots to $\robotdue$ that belong to $\partial C(R)$. We denote by $t^{60}$ and $t^{55}$ the points on $\ell$ such that $\angolo(r',t^{60},r'')=60^\circ$ and $\angolo(r',t^{55},r'')=55^\circ$, respectively. Moreover, if it exists, let $t^x$ on $\ell$ be the point between $t^{60}$ and $t^{55}$, closest to $t^{60}$ such that $| C(R\setminus \{r_1\}\cup \{t^x\}) \cap R| > 3$. Move $m_{21}$ uses points $t^{60}$, $t^{55}$, and $t^x$ to define the movement of robot $\robotdue$.
\end{itemize}

Sub-phase $\Vuno$ concerns any configuration $R$ where predicate $\VunoS$ holds. Informally, this means that $R$ is symmetric (cf $\bzero$), there exist at least two robots on the axis of reflection $\ell$ and at least two of them are not critical for $C(R)$ (cf $\bdue$), there is no robot in $c(R)$ (cf $\neg\c$), or $\robotdue$ has not yet reached its target (cf $\neg\uuno$). Figures~\ref{fig:fase_F5}.(b)-(c) show examples for such a configuration $R$. The move planned for this phase is $m_{20}$. If $\robotuno$ admits a rotational-free path then $m_{20}$ moves $\robotuno$ toward the center of the current configuration, like in Figure~\ref{fig:fase_F5}.(b). Concurrently, move $m_{20}$ makes $\robotdue$ moving along the axis toward a point such that $\robotdue$ remains the unique robot on $C^i_\uparrow(R)$ at a distance form $c(R)$ which is in the middle between $C^{i-1}_\uparrow(R)$ and $C^{i+1}_\uparrow(R)$. Concerning Figure~\ref{fig:fase_F5}.(c), notice that as soon as $\robotdue$ starts moving, a rotational-free path for $\robotuno$ is created; hence the two robots move concurrently. During the concurrent movements, $\robotdue$ must take care to not swap its role with $\robotuno$, hence maintaining bigger its distance from $c(R)$ with respect to $\robotuno$. 

Sub-phase $\Vdue$ concerns any configuration $R$ where predicate $\VdueS$ holds. Informally, this means that $R$ is symmetric (cf $\bzero$), there are two robots on the reflection axis $\ell$ and exactly one of them is critical for $C(R)$ (cf $\btre$), there is no robots in $c(R)$ (cf $\neg\c$), or $\robotdue$ has not yet reached its target (cf $\neg\udue$). Figures~\ref{fig:fase_F5}.(d)-(e) show examples for such a configuration $R$. The move planned for this phase is $m_{21}$. If $\robotuno$ admits a rotational-free path then $m_{21}$ moves $\robotuno$ toward $c(R)$ (cf Figure~\ref{fig:fase_F5}.(d)), else $\robotuno$ moves still in the direction of $c(R)$ but without reaching any circle $C^i_\downarrow$ where other robots resides. Concurrently, if $\robotdue$ lies in $[t^{60},t^{55})$ then move $m_{21}$ makes $\robotdue$ move along the axis toward $t^x$ if it exists, or $t^{55}$ like in Figure~\ref{fig:fase_F5}.(e). Notice that, like for $\Vuno$, there might be two robots that move concurrently even though now they do not risk to swap their roles.

Sub-phase $\Vtre$ concerns any configuration $R$ where predicate $\VtreS$ holds, that is symmetric configurations (cf $\bzero$) with three robots on the reflection axis $\ell$, and with exactly two of such robots critical for $C(R)$ (cf $\bquattro$). Moreover, in $R$ there is no robot in $c(R)$ (cf $\neg\c$). An example of such configurations is provided in Figure~\ref{fig:fase_F5}.(f). In such a case it is easy to see that $\robotuno$ admits a rotational-free path. Consequently, move $m_{22}$ moves $\robotuno$ in $c(R)$.

Sub-phase $\Vquattro$ concerns any configuration $R$ where predicate $\VquattroS$ holds, that is symmetric configurations (cf $\bzero$) with two robots only on the reflection axis $\ell$ that are both critical for $C(R)$ (cf $\bcinque$), and with no robots in $c(R)$ (cf $\neg\c$). Assuming $R\cap \ell =\{r_1,r_2\}$, move $m_{23}$ moves the robot $r\in \{r_1,r_2\}$ closest to $C(R\setminus \{r_1,r_2\})$. The target is the point $t\in \ell$ closest to $r$ that is not critical for $C(R\setminus \{r\} \cup t )$. An example of such a case is provided in Figure~\ref{fig:fase_F5}.(g). Notice that when $r$ reaches the target $t$, a configuration belonging to $\Vdue$ is obtained.

Finally, sub-phase $\W$ concerns any configuration $R$ where predicate $\WS$ holds. This means that there are two possibilities for $R$: either $R$ is symmetric with a robot $r$ in $c(R)$, or $R$ is asymmetric, potentially obtained from a symmetric configuration by moving a robot $r$ away from $c(R)$. The movement of $r$ toward $C^g(R)$ is dictated by move $m_{24}$ in a direction that does not form the reference angle $\alpha$ (used in Phase $\Fdue$) nor leaves symmetries. It follows that as soon as $r$ moves from $c(R)$, the configuration is asymmetric. It is possible that the obtained configuration belongs to phase $\Fquattro$ and in particular sub-phase $\Pdue$ in case $r$ is the unique robot left to be correctly moved toward the final target in order to form $F$, that is predicate $\idue$ holds. Since in phase $\W$ predicate $\idue$ must be false when the configuration is asymmetric, robots can always recognize which phase the configuration belongs to, and that $r$ is the only moving robot.

\section{Correctness}\label{sec:correctness}

In this section, we provide all the results necessary to assess the correctness of our algorithm.
To this end, we have to show that for each non-final configuration (that is configurations not satisfying $\w$) exactly one  of the start predicates (predicates for stationary configurations) in Table~\ref{tab:phases} is true. We will show that a stationary configuration satisfying a starting predicate in Table~\ref{tab:phases} will be transformed by robots' moves into a stationary configuration of another phase or in a stationary configuration satisfying $\w$. 
To this end, for each phase (that is for each configuration that satisfies a starting predicate of Table~\ref{tab:phases}), we have to show that exactly one of the starting predicates in the corresponding table (one among Table~\ref{tab:F1-phases},~\ref{tab:F2-phases},~\ref{tab:F3-phases},~\ref{tab:F5-phases},~\ref{tab:F4-phases},~and~\ref{tab:F5-phases}) is true. 
For each applied move, we have to show that during its implementation no undesired symmetry (and hence multiplicity) is created, and all robots but those involved by the move remain stationary. This assures that at the end of the move the configuration is necessarily stationary. Finally, we will show that during a move the starting predicate of Table~\ref{tab:phases} indicating the phase remains unchanged, with the exception of a few remarked situations which do not affect the correctness of the algorithm. 

\begin{lemma}\label{lem:phases-disjoint}
Given an initial leader configuration $R$ and a pattern $F$, if $\w$ does not hold then exactly one of the predicates defining a starting phase of Table~\ref{tab:phases} is true.   
\end{lemma}

\begin{proof}
First, we show that each phase manages a different set of configurations, that is the logical conjunction of any two predicates among those defining the five starting phases in Table~\ref{tab:phases} is false. Then, we show that the logical disjunction of all the predicates defining the five phases of Table~\ref{tab:phases} along with predicate $\w$ is a tautology.

The conjunction of $\Funo_s$, with $\Fdue_s$, $\Ftre_s$, $\Fquattro_s$, and $\Fcinque_s$ is false because of variables $\fdue$, $\ftre$, $\fquattro$, and $\fcinque$, respectively.
 Similarly, the conjunction of $\Fdue_s$ with $\Ftre_s$, $\Fquattro_s$, and $\Fcinque_s$ is false because of variables $\ftre$, $\fquattro$, and $\fcinque$, respectively. The conjunction of $\Ftre_s$ with $\Fquattro_s$ and $\Fcinque_s$ is false because of $\fquattro$ and $\fcinque$, respectively. Finally, the conjunction of $\Fquattro_s$ with $\Fcinque_s$ is false because of $\fcinque$.

For the second part of the proof, let us consider 
$\Funo_s \vee \Fdue_s \vee \Ftre_s \vee \Fquattro_s \vee \Fcinque_s$. This expression is equivalent to 
$\neg \w \wedge [
  (\neg \fdue \wedge \neg \ftre \wedge \neg\fquattro \wedge \neg\fcinque )  \vee 
  ( \fdue \wedge \neg \ftre \wedge \neg\fquattro  \wedge \neg\fcinque)   \vee 
  ( \ftre \wedge \neg\fquattro  \wedge \neg\fcinque)   \vee 
  ( \fquattro  \wedge \neg\fcinque )   \vee
                                           \fcinque    ]$. 
Since the sub-expression in square brackets is true, the whole expression is equivalent to $\neg \w$ which is clearly true in disjunction with $\w$.
\qed  
\end{proof}

Lemma~\ref{lem:phases-disjoint} basically shows that our algorithm can take as input any initial configuration for which it has been designed for, that is leader configurations. 

\begin{lemma}\label{lem:F1-disjoint}
Given a configuration $R$ and a pattern $F$, if $\Funo_s$ is true then exactly one of the predicates for the starting phases in Table~\ref{tab:F1-phases} is true.      
\end{lemma}
\begin{proof}
We first show that at least one predicate among $\Auno_s$, $\Adue_s$, $\B_s$, $\Cuno_s$, $\Cdue_s$, $\D_s$, $\Euno_s$, and $\Edue_s$ is true. 
By simple algebraic transformations, we obtain  $\Auno_s \vee \Adue_s \vee \B_s \vee \Cuno_s \vee \Cdue_s \vee \D_s \vee \Euno_s \vee \Edue_s = \spiu \vee \stre \vee (\sdue \wedge \neg ((\mzero \Rightarrow \muno ) \wedge \l))$.
Note that  $\spiu \vee \stre \vee \sdue$ is true for each configuration, since that expression is referred to all the possibilities about the number of robots on $C(R)$ as we assumed $|R|\geq 4$.  So it is sufficient to show that $\Funo_s \Rightarrow \neg ((\mzero \Rightarrow \muno ) \wedge \l)$ when $\sdue$ is true.
As $\Funo_s$ implies $\neg \fdue = \neg (( \mzero \Rightarrow \muno ) \wedge \sdue \wedge \l) $ (see Table~\ref{tab:phases}), when $\sdue$ holds we trivially get that $\neg (\l \wedge (\mzero \Rightarrow \muno ))$ holds.

We now show that at most one of the predicates for starting phases in Table~\ref{tab:F1-phases} is true. To this end, it is sufficient to show that the logical conjunction of any two predicates is false. In most cases, this is obtained by showing that both the predicates imply the same variable, but with opposite logical values.

\begin{itemize}
\item 
Concerning $\Auno_s = \AunoS$, it is disjoint with $\Adue_s$ because of $\l$. Since $\spiu$ implies $\neg (\sdue \vee \stre)$, then $\Auno_s$ is disjoint with any of the remaining predicates, as for them either $\sdue$ or $\stre$ is true.
\item
Concerning $\Adue_s = \AdueS$,  either it implies $\spiu$ (and then, as above, it differs from all the remaining predicates) or it implies $\stre \wedge \neg \l$. However, $\l$ is positive in all the remaining predicates where $\stre$ holds; predicates $\B_s$ and $\Edue_s$ are both disjoint with $\Adue_s$ because of $\sdue$.
\item
Predicate $\B_s$ is disjoint with all the remaining predicate but $\Edue_s$ because of $\sdue$ (the others require $\stre$). $\B_s$ and $\Edue_s$ are disjoint because of $\mzero \Rightarrow \muno$.
\item
Concerning $\Cuno_s = \CunoS$, it is disjoint with both $\Cdue_s$ and $\Euno_s$ because of $\tzero$, with $\D_s$ because of $\tuno$ (since $\tzero$ implies $\neg \tuno$), and with $\Edue_s$ because $\stre$ implies $\neg \sdue$.
\item
Predicate $\Cdue_s$ is disjoint with $\D_s$ because of $\tuno$, with $\Euno_s$ because of $\mzero \Rightarrow \muno$, and with $\Edue_s$ because $\stre$ implies $\neg \sdue$.
\item 
Predicate $\D_s$ is disjoint with $\Euno_s$ because of $\mzero \Rightarrow \muno$, and with $\Edue_s$ because $\stre$ implies $\neg \sdue$.
\item
Predicate $\Euno_s$ is disjoint with $\Edue_s$ because $\stre$ implies $\neg \sdue$.
\end{itemize}
Summarizing, we get that exactly one of the predicates for the starting phases in Table~\ref{tab:F1-phases} is true when $\Funo_s$ holds.
\qed
\end{proof}

\begin{lemma}\label{lem:F2-disjoint}
Given a configuration $R$ and a pattern $F$, if $\Fdue_s$ is true then exactly one of the predicates for the starting phases in Table~\ref{tab:F2-phases} is true.   
\end{lemma}
\begin{proof}
We first show that at most one of the predicates for starting phases in Table~\ref{tab:F2-phases} is true. To this end, it is sufficient to show that the logical conjunction of any two predicates is false. This is obtained by showing that both the predicates imply the same variable, but with opposite logical values. In particular, $\Guno_s$ and $\Gdue_s$ are disjoint because of $\gzero$. Since $\muno \Rightarrow \mzero$, we can assume that both $\Guno_s$ and $\Gdue_s$ imply $\c \Rightarrow \mzero = \neg \c \vee \mzero = \neg( \c \wedge \neg \mzero)$. Then, both $\Guno_s$ and $\Gdue_s$ are disjoint with $\H_s$ because of $( \c \wedge \neg \mzero)$.

We now show that exactly one of the predicates for starting phases in Table~\ref{tab:F2-phases} is true. To this end, we show that $\Guno_s \vee \Gdue_s \vee \H_s$ is true when $\Fdue_s$ holds. We first analyze $\Guno_s \vee \Gdue_s$; it corresponds to $[\GunoS] \vee [\GdueS] = [\neg \gzero \vee (\gzero \wedge \neg \guno)]\wedge(\c \Rightarrow \muno) = [\neg \gzero \vee \neg \guno]\wedge(\c \Rightarrow \muno)$. Since $\neg \gzero \Rightarrow \neg \guno$, the last expression can be simplified into  
\begin{equation}\label{eq:1}
\Guno_s \vee \Gdue_s = \neg \guno \wedge(\c \Rightarrow \muno). 
\end{equation}

Now, observe that $\Fdue_s$ implies both $\fdue = (\mzero \Rightarrow \muno) \wedge \sdue \wedge \l$ and $\neg \ftre = \neg [ (\mzero \Rightarrow \muno) \wedge \gdue] = \neg (\mzero \Rightarrow \muno) \vee \neg \gdue $. In turn, it follows that $\Fdue_s$ implies $(\mzero \Rightarrow \muno)$, $\sdue$, and $\neg \gdue$. According to the definition of $\gdue$, it  follows that either $\guno$ is false or $\not\exists g'' \in \partial C(R)$ antipodal to $g’$. The latter condition cannot hold since $\Fdue_s$ implies $\sdue$, and hence $\guno$ is false. Concluding, Eq.~\ref{eq:1} is equivalent to $\Guno_s \vee \Gdue_s = \c \Rightarrow \muno$. 

Finally, we get $\Guno_s \vee \Gdue_s \vee\H_s = (\c \Rightarrow \muno)\vee (\c \wedge \neg \mzero) =  \neg \c \vee (\mzero \Rightarrow \muno)$.
 Since in $\Fdue_s$ predicate $\mzero \Rightarrow  \muno$ holds, then the claim follows.

Summarizing, we get that exactly one of the predicates for the starting phases in Table~\ref{tab:F2-phases} is true when $\Fdue_s$ holds.
\qed
\end{proof}

\begin{lemma}\label{lem:F3-disjoint}
Given a configuration $R$ and a pattern $F$, if $\Ftre_s$ is true then exactly one of the predicates for the starting phases in Table~\ref{tab:F3-phases} is true.  
\end{lemma}
\begin{proof}
By observing that $\dzero \wedge \neg \duno$ is equivalent to $\neg(\dzero \Rightarrow  \duno)$, it is easy to see that the logical conjunction of any two predicates among $\M_s$, $\N_s$,  and $\O_s$ is false. Then at most one of these predicates is true for $R$. On the other hand the logical disjunction of predicates $\M_s$, $\N_s$, and $\O_s$ is also trivially true. Then, exactly one of the predicates for the starting phases in Table~\ref{tab:F3-phases} is true.
\qed
\end{proof}

\begin{lemma}\label{lem:F4-disjoint}
Given a configuration $R$ and a pattern $F$, if $\Fquattro_s$ is true then exactly one of the predicates for the starting phases in Table~\ref{tab:F4-phases} is true.   
\end{lemma}
\begin{proof}
We first show that at least one predicate among $\Puno_s$, $\Pdue_s$, $\Quno_s$, $\Qdue_s$, $\Qtre_s$, and $\Qquattro_s$ is true. 
Since $\Fquattro_s = \FquattroS$ holds, then $\fquattro$ holds as well. Since $\fquattro = 
\neg \q \wedge (\iuno \vee \idue)~ \vee ~\q \wedge (\itre \vee \iquattro \vee \icinque \vee \isei) $, then the logical disjunction $\Puno_s \vee \Pdue_s \vee \Quno_s \vee \Qdue_s \vee \Qtre_s \vee \Qquattro_s$ is true. 

We now show that $R$ is processed by exactly one sub-phase of $\Fquattro$. 
$\Puno_s$ is disjoint with $\Pdue_s$ since $\iuno$ implies that exactly two robots are unmatched, while $\idue$ implies that exactly one robot is unmatched. $\Puno_s$ and $\Pdue_s$ are disjoint with (any sub-phase of) $\Q$ because of $\q$. According to the formal definitions of predicates $\itre,\ldots,\isei$, it follows that $\Quno_s$ is disjoint with any other sub-phase of $\Q$, since $\itre$ implies $r_1$ inside $C(R)$ while $\iquattro,\ldots,\isei$ all imply $R=\partial C(R)$. $\Qquattro_s$ is disjoint with both $\Qdue_s$ and $\Qtre_s$ since $\isei$ implies that exactly one robot is unmatched, while both $\iquattro$ and $\icinque$ imply that exactly two robots are unmatched. Finally, $\Qdue_s$ and $\Qtre_s$ are disjoint because of the third items in the definitions of $\iquattro$ and $\icinque$.
\qed
\end{proof}

\begin{lemma}\label{lem:F5-disjoint}
Given a leader configuration $R$ and a pattern $F$, if $\Fcinque_s$ is true then exactly one of the predicates for the starting phases in Table~\ref{tab:F5-phases} is true.      
\end{lemma}

\begin{proof}
We first show that at least one predicate among $\T_s$, $\U_s$, $\Vuno_s$, $\Vdue_s$, $\Vtre_s$, $\Vquattro_s$ and $\W_s$ is true. 
$\Fcinque_s$ implies $(\bzero  \vee \zuno \vee \zdue) \wedge \e $. Hence, $R$ does not contain multiplicities. If $\zuno$ is true, then $\U_s$ is true. If $\zdue$ is true, then one predicate among $\W_s$, $\Vuno_s$, or $\Vdue_s$ is true. If both $\zuno$ and $\zdue$ are false, then necessarily $\bzero$ holds and $\c$ must be false as otherwise $\zdue$ is true. It follows that $R$ admits exactly one axis of reflection $\ell$ with robots on it. In such a case $\T \vee \Vuno_s \vee \Vdue_s \vee \Vtre_s \vee \Vquattro_s=\buno \vee \bdue \vee \btre \vee \bquattro \vee \bcinque$. Actually, such an expression is true as $\buno$, $\ldots$, $\bcinque$ cover all the possible cases concerning robots lying on $\ell$.

We now show that at most one of the predicates for starting phases in Table~\ref{tab:F5-phases} is true. To this end, it is sufficient to show that the logical conjunction of any two predicates is false. 

\begin{itemize}
\item 
Concerning $\T_s = \TS$, it is disjoint with $\U_s$ because of $\zuno$. It is disjoint with $\Vuno_s$, $\Vdue_s$, $\Vtre_s$ and $\Vquattro_s$ because of $\buno$ ($\Vuno_s$, $\Vdue_s$, $\Vtre_s$ and $\Vquattro_s$ require that that at least one predicate among $\bdue,\ldots,\bcinque$ holds, while predicates $\buno,\ldots,\bcinque$ are pairwise disjoint). It is disjoint with $\W_s$ because $\T_s$ implies that $\bzero \wedge \neg \c$ holds, whereas $\W_s$ implies that $\neg \bzero \vee \c$ holds.

\item
Concerning $\U_s = \US$, it is disjoint with $\Vuno_s$, $\Vdue_s$, $\Vtre_s$ and $\Vquattro_s$ because such predicates require that that at least one among $\bdue,\ldots,\bcinque$ holds, and this implies that $R$ must have a reflection axis with at least two robots on the axis. Conversely $\US$ implies that $R$ is either asymmetric or reflexive with at most one robot on the axis. To prove that $\U_s$ is disjoint with $\W_s$ we show that $\zuno\Rightarrow \neg\zdue$. Consider a configuration $R$ that fulfills $\zuno$: 
	\begin{itemize}
	\item
	if $R$ is symmetric (i.e., $\bzero$ holds) then $c(R)$ is not occupied (i.e., $\neg\c$ holds) -- since $\zdue$ implies $\c \vee \neg\bzero$ holds, then $\zdue$ is false;
	\item
	if $R$ is asymmetric (i.e., $\neg \bzero$ holds), let $r$ be the faraway robot detected by predicate $\zuno$. Since  $R$ is asymmetric then $r$ is not on the axis $\ell$ given in the definition of predicate $\zuno$. If $r'=\partial C_{\uparrow}^{1}(R)$, then $R\setminus \{r'\}$ cannot be symmetric as the asymmetry of $R$ is due to $r$. This means that $\zdue$ cannot hold. 
    \end{itemize}

\item
Concerning $\Vuno_s = \VunoS$, it is disjoint with $\Vdue_s$, $\Vtre_s$ and $\Vquattro_s$ because of $\bdue$. Moreover it is trivially disjoint with $\W_s$ ($\W_s$ requires $\neg\Vuno_s$).

\item
Concerning $\Vdue_s = \VdueS$, it is disjoint with $\Vtre_s$ and $\Vquattro_s$ because of $\btre$. Moreover it is trivially disjoint with $\W_s$ ($\W_s$ requires $\neg\Vdue_s$).

\item
Concerning $\Vtre_s = \VtreS$, it is disjoint with $\Vquattro_s$ because of $\bquattro$. It is disjoint with $\W_s$ because if $R$ belongs to $\W$ then it fulfills $\bzero \wedge \neg \c$ holds, while we have already observed that $\zdue$ implies $\neg \bzero \vee \c$.

\item
Concerning $\Vquattro_s = \VquattroS$, it is disjoint with $\W_s$ because we are in the same situation of $\Vtre_s$ vs $\W_s$.
\end{itemize}
\qed
\end{proof}

We are now ready to provide the correctness proof of our algorithm for each phase, and then we combine all phases by means of the final theorem that provides the correctness of the whole algorithm.
For each phase we consider all possible sub-phases. For each sub-phase we show all the possible scenarios where the corresponding moves lead. In particular, for each phase among $\Funo$, $\Fdue$, $\Ftre$, $\Fquattro$, $\Fcinque$, and for each move $m$ defined in the algorithm, we need to show several properties that guarantee to our algorithm to safely evolve until pattern $F$ is formed. For the first four phases $\Funo$, $\Fdue$, $\Ftre$, and $\Fquattro$ all moves involve only one robot and we are going to prove the following properties:

\begin{itemize}
	\item[$\h_0$:] at the beginning, $m$ involves only one robot;
	\item[$\h_1$:] while a robot is moving according to $m$, the configuration is a leader configuration;
	\item[$\h_2$:] $m$ is safe, and in particular that while a robot is moving according to $m$, all other robots are stationary;
	\item[$\h_3$:] while a robot is moving according to $m$, no collisions are created;
	\item[$\h_4$:]  if $m$ is associated to phase $\mathcal{X}$, then the predicate $\mathcal{X}_e$ holds once a robot has terminated to apply $m$;  
	\item[$\h_5$:] $m$ preserves stationarity.
\end{itemize}

Basically property $\h_2$ is needed to correctly address the $\Member$ problem described in Section~\ref{sec:strategy}.

About Property $\h_4$, the stop of a robot $r$ is due to three events. First, the adversary may stop $r$ before reaching its target. Second, the move might be subject to Procedures $\SMove$ or $\DistMin$, hence $r$ reaches an intermediate target. Third, $r$ reaches the real target imposed by the current move. From the proofs, we omit the analysis of the first condition because the situation obtained once the adversary stops the moving robot $r$ always equals what can happen while $r$ is moving, that is the analysis of property $\h_2$ holds.

About Property $\h_5$, we always omit this property form our proofs because it comes for free from the other properties once we have shown that there is always only one robot $r$ moving. So whenever $r$ stops moving, the configuration is stationary.

\begin{lemma}\label{lem:corr-F1}
Let $R$ be a stationary configuration in $\Funo$. From $R$ the algorithm eventually leads to a stationary 
configuration belonging to $\Fdue$, $\Ftre$, $\Fquattro$, $\Fcinque$ or where $\w$ holds.
\end{lemma}
\begin{proof}
By Lemma~\ref{lem:F1-disjoint}, exactly one of the predicates for the starting phases in Table~\ref{tab:F1-phases} is true. In turn, this implies that exactly one of the moves associated to the sub-phases of $\Funo$ is applied to $R$. We show that the properties $\h_0,\ldots,\h_4$ hold for each possible move applied to $R$.

\noindent
Let us consider sub-phase $\Auno$ where move $m_1$ is performed. 

\begin{description}
    \item[$\h_0$:] Move $m_1$ only concerns the not critical robot $r$ on $C(R)$ of minimum view. 

    \item[$\h_1$:] During the movement of $r$, the configuration remains a leader configuration. Actually it is asymmetric as there are no multiplicities yet and $r$ cannot participate to neither a rotation, being the only robot on $C_{\downarrow}^1(R)$, or a reflection as the axis of symmetry should pass through $r$, but then the starting configuration $R$ was symmetric, a contradiction.

    \item[$\h_2$:] We show that $m_1$ is safe. 
    As soon as the robot moves, predicate $\Adue=\AdueS$ holds, since from $\spiu$ by moving one robot either $\spiu$ or $\stre$ holds and $\neg \l$ holds as the robot has not yet reached the target.
    
Any configurations $R'$ observed during the move of $r$ cannot belong to $\Fdue$ as $\sdue$ does not hold. It cannot belong to $\Ftre$ as well because $\ftre$ does not hold. In fact, $\ftre$ does not hold in $\Auno_s$ and move $m_1$ cannot change this status. Possibly, the $R'$ falls in $\Fquattro$, in particular sub-phases $\Pdue$ and $\Quno$. In fact, in $\Puno$ predicate $\guno$ should hold, but $r$ is certainly not on $C^g(R)$; in $\Qdue$, $\Qtre$, and $\Qquattro$ all robots should belong to $\partial C(R)$, but $r$ does not. If $\Pdue_s$ holds, then only $r$ can be the remaining unmatched robot that moves in $\Pdue$ since $r$ is guaranteed to not meet a point in $F$ according to the use of Procedure $\SMove$. It follows that during the movement, in case $r$ is stopped by the adversary, it will be selected again by the algorithm as the unique robot that performs move $m_{14}$. Similar arguments can be applied if $\Quno$ holds, where there is only one robot inside $C(R)$. 

Finally, we prove that $R'$ cannot belong to $\Fcinque$ by showing that $\neg \fcinque= \neg \e \vee (\neg \bzero \wedge \neg \zuno \wedge \neg \zdue)$ holds in $R'$. As the moving robot is alone on $C_\downarrow^1(R')$, $R'$ is asymmetric, that is $\bzero$ does not hold. If predicate $\spiu$ holds then $\zuno$ is false. If $\stre$ holds, then there should be a faraway robot $r'$ on $C(R')$ such that $R'\setminus \{r'\}$ admits an axis where $r$ resides (being alone on $C_\downarrow^1(R')$), that is $\zuno$ is false. Now we show that $\zdue$ is false in $R'$. If $r$ is the closest robot to $c(R')$ then its distance from $c(R')$ is not less than $\delta(C^g(R'))$, that is $\zdue$ is false; else let $r'$ be the robot closest to $c(R')$. It follows that if $R''=R'\setminus \{r'\}$ is symmetric, that is $r$ is on an axis of symmetry of $R''$, then also $R\setminus \{r'\}$ is symmetric. This contradicts the fact that $\zdue$ was false in $R$ as $\neg \fcinque$ and $\e$ hold.

The above arguments also ensure that no other robot than $r$ can move from the reached configurations.

    \item[$\h_3$:] Move $m_1$ guarantees there are no robots between $r$ and its target.
    
    \item[$\h_4$:] Assume that $r$ stops moving because it reaches an intermediate target dictated by Procedure $\SMove$. In this case, predicates $\w$, $\Puno_s$, or $\Pdue_s$ might hold because $\iuno$ or $\idue$ become true (clearly, $\itre,\ldots ,\isei$ cannot become true as $F$ should equal $\partial F$). If this is not the case, $r$ is unmatched and the same considerations given for $\h_2$ hold, that is the configuration belongs to phases $\Adue$ or $\Pdue$ or $\Quno$.

Assume that $r$ reaches the target $t=[r,c(R)] \cap C^{0,1}(R)$. If the configuration remains in $\Funo$, then $\l$ holds and the configuration is either in $\Auno$ if $\spiu$ holds, or in $\Adue$, $\C$, $\D$, or $\Euno$ if $\stre$ holds. If the configuration is not in $\Funo$, then by the above analysis it belongs to $\Fquattro$ or $\w$ holds.
\end{description}

Sub-phase $\Adue$, where move $m_2$ is performed, is the continuation of sub-phase $\Auno$ in case the moving robot $r$ stops before reaching its target on $C^{0,1}(R)$ to make predicate $\l$ newly true. Then the same analysis of move $m_1$ applies. Hence moves $m_1$ and $m_2$ are repeatedly applied in order to remove not critical robots from $\partial C(R)$ until there remain exactly three robots on $C(R)$ (unless the configuration reaches phase $\Fquattro$ or satisfies $\w$ before). This can be done according to Properties~\ref{prop1} and~\ref{prop2}. Once the resulting configuration satisfies $\stre \wedge \l$, it does not belong to sub-phase $\A$ anymore. This situation occurs in a finite number of steps. In fact, as shown above, move $m_1$ along with move $m_2$ bring not critical robots one by one to their designed targets. Since by assumption a robot is guaranteed to traverse at least distance $\nu$ each time it moves, then a finite number of steps suffices to reach the desired configuration.
 
Three robots on $C(R)$ are necessary to maintain the configuration a leader one in case a multiplicity should be formed in $c(R)$, that is when predicate $\mzero \wedge \neg \muno$ is true. If there are only two robots in $\partial C(R)$, a third one from $\Int(C(R))$ is moved on $C(R)$. This is done in sub-phase $\B$ by means of move $m_3$. 	

 \begin{description}
	\item[$\h_0$:] The move only concerns robot $r$ on $C_\downarrow^1(R)$ with minimum view. 
	\item[$\h_1$:] During the movement of $r$ the configuration remains   a leader configuration. Actually it is asymmetric as there are no multiplicities yet, $r$ cannot participate to neither a rotation, being the only robot on $C_{\downarrow}^1(R)$, or a reflection as the axis of symmetry should pass through $r$, but then the starting configuration $R$ was symmetric too, a contradiction.
	\item[$\h_2$:] We show that $m_3$ is safe. Assuming that the configuration observed during the move of $r$ still belongs to $\Funo$, then $r$ is always detected as the unique robot on $C_\downarrow^1(R)$ and hence $\B_s$ still holds. 
	The configuration observed while $r$ moves cannot belong neither to $\Fdue$ nor to $\Ftre$ as $\mzero\Rightarrow \muno$ does not hold (we have $\mzero \wedge \neg \muno$ by hypothesis). Similarly, it cannot belong to $\Fquattro$ 
	as a multiplicity in $c(R)$ must be created. 
Finally, it cannot belong to $\Fcinque$ as the configuration is asymmetric, that is $\bzero$ is false; $\sdue$ holds hence $\zuno$ is false; $\zdue$ remains false because, by removing the closest robot to $c(R)$, the configuration cannot be symmetric due to the movement of robot $r$, as observed in $\h_1$.

The above arguments ensure that during the movement all robots but $r$ are stationary as the configuration remains in $\B$.
	\item[$\h_3$:] No collision is possible as, by definition, there are no robots between $C_\downarrow^1(R)$ and $C(R)$, and the target $t$ on $C(R)$ cannot coincide with one of the positions of the two robots on $C(R)$.
	\item[$\h_4$:] Once $r$ reaches its target, the configuration cannot belong to $\Fdue$, $\Ftre$, $\Fquattro$ for the same reasons as in $\h_2$. The configuration cannot belong to $\Fcinque$ as it is asymmetric, that is $\bzero$ is false; $\stre$ holds, but the robots on $C(R)$ form a right-angle triangle and then $\zuno$ is false; $\zdue$ remains false for the same reasons as in $\h_2$.
	It follows that the configuration remains in $\Funo$. In particular, the configuration is in $\C_1$, $\C_2$ or $\D$, depending on the kind of triangle formed by the robots on $C(R)$, that is the status of predicates $\tzero$ and $\tuno$. 
\end{description}

When $|\partial C(R)|=3$, we have to guarantee that two robots on $C(R)$ are antipodal before removing the third one, otherwise $C(R)$ could change its radius. This is done in sub-phase $\Cuno$ by means of move $m_4$. The move involving one of the three robots on $C(R)$ makes the triangle they form containing a $90^\circ$ angle, and hence two antipodal robots.

\begin{description}
	\item[$\h_0$:] As specified by the definition of $m_4$, the only robot $r$ involved in this move is that on $C(R)$ corresponding to angle $\alpha_2$ if the triangle has angles $\alpha_1\geq \alpha_2\geq \alpha_3$ (in case of ties, the uniqueness of $r$ is guaranteed by using the view of robots). 
	\item[$\h_1$:] $r$ remains always detectable because as soon as it starts moving we get $\alpha_1> \alpha_2> \alpha_3$, that is the triangle formed by the three robots on $C(R)$ is asymmetric. Moreover the triangle remains asymmetric as during the movement and until the end, $\alpha_1$ increases, $\alpha_2$ maintains its value and $\alpha_3$ decreases. 
	\item[$\h_2$:] We show that $m_4$ is safe. Assuming that the configuration observed during the move of $r$ still belongs to $\Funo$, as $\stre \wedge\tzero$ remains true during the movement, then the configuration remains in $\Cuno$. While $r$ moves the observed configuration cannot belong to $\Fdue$ as $\sdue$ does not hold. The same for $\Ftre$, as there are no two antipodal robots and then $\gdue$ is false. Predicates  $\Fquattro_s$ cannot hold as the movement is controlled by $\SMove$ which preventively calculates the possible points that could make $\Fquattro_s$ true at the end of the movement. So, during the movement, $\Fquattro_s$ remains false. Similarly, $\Fcinque_s$ is false as it is controlled by $\SMove$.
As the above defined angles are such that $\alpha_1> \alpha_2> \alpha_3$ during the move, then the robot $r$ on $\alpha_2$ is always uniquely determined;
	\item[$\h_3$:] No collisions are created as the robot moves on $C(R)$ having as target a point antipodal to one of the other two robots. The third robot cannot be on its trajectory as, otherwise, the smallest enclosing circle would be different from $C(R)$.
	\item[$\h_4$:] If $r$ stops on a target specified by Procedure $\SMove$, the configuration can remain in $\Cuno$.
	As in case $\h_2$, the configuration cannot belong to $\Fdue$ as $\sdue$ does not hold. The same for $\Ftre$, as there are no two antipodal robots and then $\gdue$ is false. If Procedure $\SMove$ computes a target specified at line~\ref{algo:mFquattroe}, it could be possible that there exists an embedding satisfying $\iuno$. Then the configuration can belong to sub-phase $\Puno$ of $\Fquattro$. Moreover,  predicate $\w$ might hold.
	If the computed target comes from line~\ref{algo:mAquattros} of Procedure $\SMove$, an angle of $\alpha$ degrees is formed with a robot on $C^g(R)$. In that case the configuration may belong to either sub-phase $\Puno$ or sub-phase $\Pdue$ of $\Fquattro$.
If the computed target comes from line~\ref{algo:mWquattros} of Procedure $\SMove$, an asymmetric configuration where $\zuno$ holds is reached that may belong to sub-phase $\U$ of $\Fcinque$.
	
	If the robot stops reaching the target of move $m_4$, forming an angle of $90^\circ$, $t_0$ is false and then the configuration can satisfy $\CunoE$. If the configuration is not in $\Funo$, as above, it can satisfy $\w$, $\Fquattro_s$, or $\Fcinque_s$.
\end{description}

In case a multiplicity in $c(R)$ must be formed, in order to guarantee that the configuration remains a leader configuration, the algorithm ensures the triangle formed by the three robots on $C(R)$ is asymmetric. To guarantee the stationarity of the configuration, we impose that the triangle has angles equal to $30^\circ$, $60^\circ$, and $90^\circ$ degrees.
This is done in sub-phase $\Cdue$ by means of move $m_5$.

\begin{description}
	\item[$\h_0$:] Among the three robots on $C(R)$, only the one that does not admit an antipodal robot is moved (notice that the status of $\tzero$ in $\Cdue_s$ implies two antipodal robots on $C(R)$).
	\item[$\h_1$:] The configuration is always asymmetric because the triangle formed by the three robots can be symmetric only at the beginning.
	\item[$\h_2$:] We show that $m_5$ is safe. As predicate $\mzero \wedge \neg \muno$ holds during the whole movement, that is a multiplicity in $c(R)$ must be formed, then the configuration observed during the movement remains in sub-phase $\Cdue$ until $\tuno$ becomes true. For the same reason, the configuration cannot belong to $\Fdue$, $\Ftre$, and $\Fquattro$. About $\Fcinque$, the configuration is asymmetric hence $\bzero$ is false; there are two antipodal robots, hence $\zuno$ is false; the configuration obtained by excluding the closest robot to $c(R)$ cannot be symmetric as the triangle on $C(R)$ is asymmetric, then $\zdue$ is false.
	
During the movement all other robots are stationary, for the same reason as in $\h_0$.
	\item[$\h_3$:] No collisions are created as the target of the moving robot is always between the antipodal robots on $C(R)$.
	\item[$\h_4$:] As discussed in $\h_2$, the configuration remains in $\Funo$ and in particular in $\Cdue$ or $\D$ depending on $\tuno$.

\end{description}

Once the three robots on $C(R)$ form a triangle having angles equal to $30^\circ$, $60^\circ$, and $90^\circ$ degrees, the multiplicity in $c(R)$ can be safely formed with respect to undesired symmetries. This is done in sub-phase $\D$ by means of move $m_6$.

\begin{description}
	\item[$\h_0$:] The only moving robot  $r$ is the closest to $c(R)$, not in $c(R)$, and of minimum view. 
	\item[$\h_1$:] The configuration remains a leader configuration as predicate $\tuno$ remains true during the movement. In fact, none of the three robots on $C(R)$ is moved because at least two points in $F$ are on $C(F)$, so the multiplicity to be formed is at most of $|F|-2$ elements. However the algorithm moves at most $|F|-3$ robots in $c(R)$ (see predicate $\muno$).  
	\item[$\h_2$:] We show that $m_6$ is safe. Predicate $\tuno \wedge \mzero \wedge \neg \muno$ holds during the whole movement, then the configuration remains in $\D$ until $\muno$ becomes true. As long as $\muno$ is false, the configuration cannot be in $\Fdue$, $\Ftre$, and $\Fquattro$. It cannot be in $\Fcinque$ because either $\e$ is false or $\bzero\vee \zuno \vee \zdue$ is false. In fact, if $\e$ is true, the configuration is asymmetric, that is $\bzero$ is false. During the movement predicate $\zuno$ remains false as the three robots on $C(R)$ form a right-angle triangle. The configuration without $r$ is asymmetric, that is $\zdue$ is false. 
	
During the movement all other robots are stationary, for the same reason as in $\h_0$. 
	\item[$\h_3$:] No collisions are created as the moving robot is the closest to the target.
	\item[$\h_4$:] 
	Once the robot reaches $c(R)$, either another robot must be moved in $c(R)$ because $\muno$ is still false and then the configuration is still in $\D$, or $\muno$ holds. In the latter case the configuration
	can be in $\Euno$. It cannot be in $\Fdue$, as $\stre$ holds, and $\w$ cannot be satisfied as at least one robot must be still moved in $c(R)$. The configuration can be in any sub-phase of $\Ftre$ as, by chance, $\gdue$ might hold. Moreover the configuration can belong to $\Fquattro$, sub-phase $\P$, but not $\Q$ as predicate $\q$ is false. Finally, the configuration cannot be in $\Fcinque$ for the same reasons as in $\h_2$. 

\end{description}

As soon as phase $\D$ terminates and the configuration is in $\Euno$, move $m_1$ is applied to remove one robot from $C(R)$ leaving only the two antipodal robots.

\begin{description}
	\item[$\h_0$:] The only not critical robot $r$ on $C(R)$ is moved. 
	\item[$\h_1$:] During the movement of $r$ the configuration remains a leader configuration as $r$ cannot participate to neither a rotation, being the only robot on $C_{\downarrow}^1(R)$, nor a reflection as the axis of symmetry should pass through $r$, but then the starting configuration $R$ was symmetric too, a contradiction.

	\item[$\h_2$:] We prove that $m_1$ in $\Euno$ is safe. Assuming that the configuration $R'$ observed during the movement of $r$ still belongs to $\Funo$, then as soon as $r$ starts moving, predicate $\Edue_s=\EdueS$ holds. In fact, in $\Euno$ predicate $\stre$ holds and by moving one robot, $\sdue$ becomes true while $\l$ becomes false as the robot has to reach its target.
	
	The configuration while $r$ moves cannot belong to $\Fdue$ as $\l$ does not hold. It cannot belong to $\Ftre$ as well because $\ftre$ does not hold in $\Euno$ and $m_1$ cannot change this status. Possibly, the configuration falls in $\Fquattro$, in particular only in sub-phase $\Pdue$. In fact, in $\Puno$ predicate $\iuno$ should hold, but since $\sdue$ holds, $g''$ should be on $\mu(g'')$. In $\Q$, there must be at least three robots on $C(R)$. If $\Pdue$ holds then only $r$ can be the remaining unmatched robot that moves in $\Pdue$ since $r$ is guaranteed to not meet a point in $F$ according to the use of Procedure $\SMove$. It follows that during the movement and once $r$ is stopped by the adversary, it will be selected again by the algorithm as the unique robot that performs move $m_{14}$. 
	
	The configuration $R'$ cannot belong to $\Fcinque$ as either $\e$ is false or we now show that $\bzero\vee \zuno\vee\zdue$ is false. As the moving robot is alone on $\C_\downarrow^1(R')$, the configuration is asymmetric, that is $\bzero$ is false. Since $\sdue$ holds $\zuno$ is false. Finally as $\zdue$ is false in $R$, then, by removing the closest robot to $c(R')$, the obtained configuration cannot be symmetric, hence $\zdue$ remains false.

The above arguments also ensure that no other robot than $r$ can move from the reached configurations.
	\item[$\h_3$:] Move $m_1$ guarantees there are no robots between $r$ and its target.
	\item[$\h_4$:]  If $r$ stops because it reaches an intermediate target dictated by Procedure $\SMove$, then $\w$ or $\Pdue$ might hold, because $\idue$ becomes true.  For the same reasons as above, $\Puno$ and $\Q$ cannot be reached. If $r$ reaches the target $t=[r,c(R)] \cap C^{0,1}(R)$, the obtained configuration is not in $\Funo$ anymore, as $\fdue$ holds. Then, it can be in $\Fdue$ as $\fdue$ is now true, and as above it can be in $\Fquattro$ or $\w$ holds.
\end{description}

Similarly to sub-phases $\Auno$ and $\Adue$, sub-phase $\Edue$ is the continuation of sub-phase $\Euno$ in case the moving robot $r$ stops before reaching its target on $C^{0,1}(R)$ to make predicate $\l$ newly true. Then the same analysis of move $m_1$ for $\Euno$ applies. Once $r$ reaches its target, predicate $\fdue$ holds, and the configuration can be in $\Fdue$, $\Fquattro$ or $\w$ holds.
\qed
\end{proof}

\begin{lemma}\label{lem:corr-F2}
Let $R$ be a stationary configuration in $\Fdue$. From $R$ the algorithm eventually leads to a stationary 
configuration belonging to $\Ftre$, $\Fquattro$ or where $\w$ holds.
\end{lemma}
\begin{proof}
Recall that the aim of $\Fdue$ is the formation of a configuration satisfying predicate $\gdue$ (cf definition of predicate $\ftre$ in $\Ftre_s$). That is, the target configuration has three robots acting as guards such that two of them, $g'$ and $g''$, are antipodal on $C(R)$ and the third one $g$ is on $C^g(R)$ forming an angle of $\alpha$ degree with $g'$.

By Lemma~\ref{lem:F2-disjoint}, exactly one of the predicates for the starting phases in Table~\ref{tab:F2-phases} is true. In turn, this implies that exactly one of the moves associated to the sub-phases of $\Fdue$ is applied to $R$. We show that the properties $\h_0,\ldots,\h_4$ hold for each possible move applied to $R$.

We analyze move $m_7$: it is performed in sub-phase $\Guno$ to bring a robot $r$ on $C^g(R)$.
\begin{description}
	\item[$\h_0$:] The move selects only one robot: the robot $r$ on $C^1_\uparrow(R)$ of minimum view. 
	\item[$\h_1$:] During the movement of $r$, the configuration remains a leader configuration as $r$ cannot participate to neither a rotation, being the only robot on $C_{\uparrow}^1(R)$, or a reflection as the axis of symmetry should pass through $r$, but then the starting configuration $R$ was symmetric too, a contradiction.
	\item[$\h_2$:] We show that $m_7$ is safe. During the movement of $r$, as $\fdue$ remains true, the observed configuration cannot belong to $\Funo$. Assuming the observed configuration still belongs to $\Fdue$, predicate $\Guno_s$ remains true.
	Since $r$ has not yet reached the target, $\gdue$ is still false and hence the observed configuration cannot belong to $\Ftre$.  
	Possibly, the configuration falls in $\Fquattro$, in particular sub-phase $\Pdue$. In fact, in $\Puno$ predicate $\guno$ should hold, but $r$ is certainly not on $C^g(R)$ (target of the current move). In $\Q$ there should be at least three robots in $\partial C(R)$, but $\sdue$ holds. If $\Pdue$ holds then only $r$ can be the remaining unmatched robot that moves in $\Pdue$ since $r$ is guaranteed to not meet a point in $F$ according to the use of Procedure $\SMove$. The configuration cannot belong to $\Fcinque$ as either $\e$ is false or $\bzero \vee \zuno \vee \zdue$ is false. In fact, if $\e$ is true then the configuration is asymmetric, that is $\bzero$ is false; $\zuno$ cannot hold as $\sdue$ holds; $\zdue$ remains false as $m_7$ moves $r$ toward $C^g(R)$. It follows that during the movement and once $r$ is stopped by the adversary, it will be selected again by the algorithm as the unique robot that performs move $m_{14}$. 
	
Summarizing,  while $r$ is moving the configuration can be in sub-phase $\Guno$ of $\Fdue$ or in sub-phase $\Pdue$ of $\Fquattro$. In both cases $r$ is always recognized as the only robot allowed to move.  
	\item[$\h_3$:] no collisions are created as there is no robot between $r$ and $C^g(R)$.
	\item[$\h_4$:] Assume that $r$ stops moving because it reaches an intermediate target dictated by Procedure $\SMove$. Then, $\w$ can hold or the configuration is still in $\Guno$. In fact, by the analysis in $\h_2$, the configuration might be in $\Pdue$, but this is excluded as the robot is on $C^i_\uparrow(F)$, for some $i>0$, while $\idue$ does not hold because it requires that $d(c(R),r)< d(c(F), f)$, where $f$ is a point on $C^1_\uparrow(F)\cap F$.
	
	Assume $r$ reaches its target on $C^g(R)$. Then, $\w$ cannot hold because there are no points of $F$ on $C^g(R)$ by definition. The configuration is not in $\Funo$ as $\fdue$ holds.  As $\gzero$ holds, the configuration can be in $\Gdue$ and in $\Ftre$ in case also $\guno$ holds. The configuration can be in $\Fquattro$, in particular in $\Puno$ or $\Pdue$ depending on $\iuno$ or $\idue$. It cannot be in $\Q$ because $\sdue$ holds.
	The configuration cannot be on $\Fcinque$ by the same analysis in $\h_2$.
\end{description}

Once there is a robot $r$ on $C^g(R)$, to make $\guno$ true, that is to correctly place guard $g$, it should be rotated on $C^g(R)$. This is done in sub-phase $\Gdue$ by move $m_8$.

\begin{itemize}
	\item[$\h_0$:] $r$ is the unique robot on $C^g(R)$;
	\item[$\h_1$:] During the movement of $r$, the configuration remains a leader configuration as $r$ cannot participate to neither a rotation, being the only robot on $C^g(R)$, or a reflection as the axis of symmetry should pass through $r$, and through the antipodal robots $g'$ and $g''$ or between them. These cases can happen only if $r$ is collinear with $g'$ and $g''$ or if it lies on the line perpendicular to the segment $[g',g'']$.
 In any other case, the movement of $g$ cannot generate an axis of reflection as it would imply the existence of a point $t'$, representing the reflection of $t$ with respect to the axis, such that $d(g,t')<d(g,t)$. This contradicts the hypothesis that $t$ is the required closest point.
	
	\item[$\h_2$:]  We show that $m_8$ is safe. During the movement of $r$, as $\fdue$ remains true, the observed configuration cannot belong to $\Funo$. Assuming the configuration still belongs to $\Fdue$, predicate $\Gdue_s$ remains true.
	As $\gdue$ is still not true, the configuration cannot belong to $\Ftre$.  
	The observed configuration does not fall in $\Fquattro$. In particular: as $\guno$ is false, it is not in $\Puno$;  
    it is not in $\Pdue$ otherwise $\idue$ true during the movement of $r$ implies $\idue$  true in the starting configuration $R$ too, a contradiction;
	it is not in $\Q$ as $\sdue$ still holds during the movement of $r$, while $\Q$ handles configurations with at least $n-1$ robots on $C(R)$.
	The configuration does not fall in $\Fcinque$. In particular: either $\e$ is false or $\bzero \vee \zuno \vee \zdue$ is false. In fact, if $\e$ is true then the configuration is asymmetric, that is $\bzero$ is false. $\zuno$ cannot hold as $\sdue$ holds. $\zdue$ does not hold as $r$ is on $C^g(R)$. 

As the configuration remains in $\Gdue$ by the above analysis, robot $r$ is always detected as the only moving one.
	\item[$\h_3$:] No collisions are created as there is only $r$ on $C^g(R)$.
	\item[$\h_4$:] If $r$ reaches its target on $C^g(G)$, $\w$ cannot hold because there are no points of $F$ on $C^g(R)$ by definition. The configuration is not in $\Funo$ as $\fdue$ holds.  As $\gdue$ holds, the configuration can be in $\Ftre$. The configuration can be in $\Fquattro$, in particular in $\Puno$ as $\iuno$ might hold. It cannot be in $\Pdue$ or in $\Q$ by the same analysis in $\h_2$. It cannot be in $\Fcinque$ by the same analysis in $\h_2$.
\end{itemize}

In case a robot $r$ is in $c(R)$, but there is no multiplicity in $c(F)$, then $r$ will be moved on $C^g(R)$ in sub-phase $\H$ by means of move $m_9$.

\begin{itemize}
	\item[$\h_0$:] $r$ is clearly the only robot to move;
	\item[$\h_1$:] the configuration is always a leader configuration by the same analysis provided for move $m_8$ in sub-phase $\Gdue$.
	\item[$\h_2$:] We show that $m_9$ is safe. As soon as $r$ starts moving, the observed  configuration can only belong in $\G$. In fact, during the movement of $r$, $C^g(R)$ changes, but $r$ is recognized as the unique robot on $C^1_\uparrow(R)$. It follows that the configuration is in $\Gdue$ or $\Guno$ depending whether $r$ is on the current circle $C^g(R)$ or not.
	The configuration cannot be in $\Funo$ as $\fdue$ holds. It cannot be in $\H$ as $\c$ is false, and it cannot be in $\Ftre$ as $\gdue$ is false. It cannot be in $\Fquattro$ as $\guno$ remains false and this excludes $\Puno$, $\idue$ remains false and this excludes $\Pdue$, and  $\Q$ is excluded by $\sdue$. It cannot be in $\Fcinque$ as the configuration is asymmetric with $c(R)$ occupied, which means both $\zuno$ and $\zdue$ are false.

In any of the reachable sub-phases described, $r$ is the only moving robot.
	\item[$\h_3$:] No collisions are created as there are no further robots between $c(R)$ and $C^g(R)$.
	\item[$\h_4$:] Once $r$ reaches $C^g(R)$, $\w$ cannot hold because there are no points of $F$ on $C^g(R)$ by definition. The obtained configuration is not in $\Funo$ as $\fdue$ holds, and is not in $\Ftre$ as $\guno$ is false. The obtained configuration is not in $\Fquattro$ nor in $\Fcinque$ as both $\fquattro$ and $\fcinque$ remain false, for similar reasons as in $\h_2$. It means the configuration can only be in phase $\Fdue$, in particular it cannot be in $\Guno$ as $\gzero$ is true, and it cannot be in $\H$ as $\c$ is false. So it can only be in $\Gdue$.
\end{itemize}
\qed
\end{proof}

\begin{lemma}\label{lem:corr-F3}
Let $R$ be a stationary configuration in $\Ftre$. From $R$ the algorithm eventually leads to a stationary configuration belonging to $\Fquattro$. 
\end{lemma}

\begin{proof}
By Lemma~\ref{lem:F3-disjoint}, exactly one of the predicates for the starting phases in Table~\ref{tab:F3-phases} is true. In turn, this implies that exactly one of the moves associated to the sub-phases of $\Ftre$ is applied to $R$. We show that the properties $\h_0,\ldots,\h_4$ hold for each possible move applied to $R$.

Let $P^* =\{[c(F),f^*]\cap C^t(R)~|~ f^*\in F^*\}$, and let $d$ be the distance from any point in $P^*$ to any target in $F^*$. Let us start the analysis of moves $m_{10}$, $m_{11}$, and $m_{12}$ by assuming there is exactly one robot $r$ on $C^t(R)$ and that $r$ is not on a point of $P^*$, that is $\dzero$ holds while $\duno$ is false. In this case, the configuration is in sub-phase $\M$ and move $m_{10}$ is applied.

\begin{itemize}
	\item[$\h_0$:] The only moved robot is that on $C^t(R)$, that by assumption is $r$.
	\item[$\h_1$:] The configuration is maintained a leader configuration by the position of the three guards $g$, $g'$, and $g''$. In fact, as $g$ is the only robot on $C^g(R)$, the configuration cannot be rotational. Moreover, the only possible reflection axis should pass through $g$, but there is no robot that can be reflected to $g'$ as, by predicate $\guno$, $g'$ is the only robot such that $\angolo(g,c(R),g')=\alpha$. 
	\item[$\h_2$:] We show that $m_{10}$ is safe. As $\ftre$ holds during the movement of $r$, the observed configuration cannot be in $\Funo$ and $\Fdue$. Moreover, during the move, both robots $r$ and $g$ do not stay neither on a target of $F$ nor on $C(R)$. As each predicate among $\iuno, \ldots, \isei$ requires that at most two robots are not on target and at least one of them is on $C(R)$, they are all false and then the observed configuration is not in $\Fquattro$. 
		The configuration does not fall in $\Fcinque$. In particular: either $\e$ is false or $\bzero \vee \zuno \vee \zdue$ is false. In fact, if $\e$ is true then the configuration is asymmetric due to guard $g$, that is $\bzero$ is false. $\zuno$ cannot hold as $\gdue$ holds which implies the existence of antipodal robots on $C(R)$. $\zdue$ does not hold as there is $g$ on $C^g(R)$. 
	Hence, as $\MS$ holds, the observed configuration remains in $\Ftre$, sub-phase $\M$.
Robot $r$ remains the only robot on $C^t(R)$, then all the other robots are stationary.
	\item[$\h_3$:] Collisions are impossible as $r$ rotates on $C^t(R)$ and it is the only robot on it.
	\item[$\h_4$:] Once $r$ reaches the target on $C^t(R)$, by the analysis done in $\h_2$, the configuration remains in $\Ftre$, but now $\duno$ holds. If still $r = \minview( R^{\neg m}_{\eta} )$, then $\ddue$ is false as the current target $\mu(r)\in F^*$ is reachable without intersecting $C^t(R)$. If another robot $r' = \minview( R^{\neg m}_{\eta} )$ (this can happen only the first time phase $\Ftre$ is applied, and there was already $r$ on $C^t(R)$), then again $\ddue$ is false because otherwise the distance from $r'$ to $\mu(r')$ would be greater than $d$, a contradiction. Hence, the obtained configuration is in sub-phase $\O$.
\end{itemize}

When $\dzero \Rightarrow \duno$ holds, the configuration is in sub-phase $\N$ or $\O$ depending on whether $\ddue$ is true or not. Let us assume that $\ddue$ is true, that is $C^t(R)\cap (r, \mu(r)] \neq \emptyset$, where $r = \minview( R^{\neg m}_{\eta} )$. Then, move $m_{11}$ is applied and $r$ is moved toward the closest point $p_1$ in  $C^t(R)\cap (r, \mu(r)]$ according to Procedure $\DistMin$. 

\begin{itemize}
	\item[$\h_0$:] The only moved robot is $r$ that is the one with minimum view in $R^{\neg m}_{\eta}$. It is unique as the configuration is a leader configuration and hence there cannot be two robots in $R^{\neg m}_{\eta}$ with the same view.
	\item[$\h_1$:] As in the analysis done for move $m_{10}$, the configuration is maintained a leader configuration by the  guards. 
	\item[$\h_2$:] We show that $m_{11}$ is safe. During the movement of $r$, the configuration remains in $\Ftre$ by the same analysis done for move $m_{10}$. In particular, it can still belong to sub-phase $\N$ or to sub-phase $\O$. In fact, if there are robots between $r$ and $p_1$, then $r$ receives an intermediate target $p$ by procedure $\DistMin$ and hence $\ddue$ may remain true or not. 
By Lemma~\ref{lem:corr-DistMin}, condition $3)$, robot $r$ reduces its distance to $p_1$ and then to $\mu(r)$, being $p_1$ an intermediate point between the initial position of $r$ and $\mu(r)$. Then the distance $\eta'$ of $r$ to $\mu(r)$ is such that $\eta' <\eta$, and hence it  remains the only robot in $R^{\neg m}_{\eta'}$. All the other robots remain stationary as their distance to any target is at least $\eta$.
	\item[$\h_3$:] Collisions are impossible as $r$ is moved according to Procedure $\DistMin$.
	\item[$\h_4$:] By the analysis provided for $m_{10}$, the configuration remains in $\Ftre$ while $r$ is moving. Here there are three possible cases: (i) $r$ reaches target $p_1$; (ii) $r$ has reached the new target $p$ (or it has been stopped before by the adversary) and the trajectory toward $\mu(r)$ does not intersect $C^t(R)$ anymore; (iii) $r$ has reached the new target $p$ (or it has been stopped before by the adversary) and still the trajectory toward $\mu(r)$ intersects $C^t(R)$. In case (i) $r$ is now the only robot on $C^t(R)$, and the configuration is in sub-phase $\M$. In case (ii), $\ddue$ is false and the configuration is in $\O$, and by Lemma~\ref{lem:corr-DistMin}, $r$ will be selected again to move. In case (iii), the configuration is still in $\N$ and robot $r$ will be selected again by the algorithm. In fact, by condition $3)$ of Lemma~\ref{lem:corr-DistMin}, $r=\minview( R^{\neg m}_{\eta'} )$, with $\eta'<\eta$. 
	Moreover, by condition $1)$ of Lemma~\ref{lem:corr-DistMin}, $C(R)$ and the target $\mu(r)$ do not change. Let $p'\neq p$ be the closest point to $r$ on $(r,\mu(r)]\cap C^t(R)$. By condition $2)$ of Lemma~\ref{lem:corr-DistMin}, there are no robots between $r$ and $p'$, that is there will not be a deviation by means of $\DistMin$ when $r$ applies again $m_{11}$. Note that, in case (iii), still $r$ applies $m_{11}$, so predicate $\N_e$ does not hold, but now $p'$ is assured to be reached within a finite number of steps because in each step $r$ moves of at least $\nu$.
\end{itemize}

Let us assume that $\ddue$ is false, that is $C^t(R)\cap (r, \mu(r)] = \emptyset$, where $r = \minview( R^{\neg m}_{\eta} )$. Then, the configuration is in $\O$, move $m_{12}$ is applied and $r$ is moved toward $\mu(r)$ according to Procedure $\DistMin$. 

\begin{itemize}
	\item[$\h_0$:] As in move $m_{11}$, the only moved robot is that with minimum view $r=R^{\neg m}_{\eta}$. Robot $r$ is unique as the configuration is a leader configuration.
	\item[$\h_1$:] As in the analysis done for move $m_{10}$, the configuration is maintained a leader configuration by the  guards. 
	\item[$\h_2$:] During the movement of $r$, the configuration remains in $\Ftre$ by the same analysis done for move $m_{10}$. In particular, it can only belong to $\O$ as the trajectory from $r$ to $\mu(r)$ cannot intersect $C^t(R)$, even if Procedure $\DistMin$ assigns a new target.  
By Lemma~\ref{lem:corr-DistMin} condition $3)$, robot $r$ reduces its distance to $\mu(r)$. Then the distance $\eta'$ of $r$ to $\mu(r)$ is such that $\eta' <\eta$, and hence it remains the only robot in $R^{\neg m}_{\eta'}$. All the other robots remain stationary as their distance to any target is at least $\eta$.
	\item[$\h_3$:] Collisions are impossible as $r$ is moved according to Procedure $\DistMin$.
	\item[$\h_4$:] 	If $r$ reaches a new target $p$ assigned by $\DistMin$ (or it has been stopped before by the adversary), the configuration is still in $\O$ but now, by condition $2)$ of Lemma~\ref{lem:corr-DistMin}, there are no robots between $r$ and $\mu(r)$. By condition $1)$ of the same lemma both $C(R)$ and $\mu(r)$ do not change, and by condition $3)$, robot $r$ reduces its distance to $\mu(r)$. So $r$ applies again $m_{12}$, predicate $\O_e$ does not hold, but now $\mu(r)$ is assured to be reached within a finite number of steps because in each step $r$ moves of at least $\nu$.
	
	If $r$ reaches $\mu(r)$ and there are still unmatched points in $F\setminus \{\mu(g),\mu(g'')\}$ then the configuration remains in $\Ftre$ as $\FtreS$ holds. In particular, it belongs either to $\N$ or to $\O$ depending on $\ddue$. It cannot belong to $\M$ as a robot on $C^t(R)$ would move before $r$. 

If $r$ reaches $\mu(r)$ and all points in $F\setminus \{\mu(g),\mu(g'')\}$ are matched, then $\w$ does not hold as $g$ is not on $\mu(g)$. The configuration is not in $\Fcinque$ by the same analysis done for $\h_2$ of move $m_{10}$. Hence, the configuration is in $\Fquattro$. In particular it can be in sub-phases $\Puno$, $\Pdue$, or $\Quno$, as $g$ is not on $C(R)$.
\end{itemize}

Clearly, the fact that from $\O$ the configuration can go back to $\N$ or to $\O$ can happen only a finite number of times, until all points in $F\setminus \{\mu(g),\mu(g'')\}$ become matched.

In conclusion, moves $m_{10}$, $m_{11}$, and $m_{12}$ can be applied only a finite number of times, then eventually the configuration leaves phase $\Ftre$ and, following the above analysis, phase $\Fquattro$ is reached.
\qed
\end{proof}

\begin{lemma}\label{lem:corr-F4}
Let $R$ be a stationary configuration in $\Fquattro$. From $R$ the algorithm eventually leads to a stationary configuration where $\w$ holds. 
\end{lemma}
\begin{proof}
By Lemma~\ref{lem:F3-disjoint}, exactly one of the predicates for the starting phases in Table~\ref{tab:F4-phases} is true. In turn, this implies that exactly one of the moves associated to the sub-phases of $\Fquattro$ is applied to $R$. We show that the properties $\h_0,\ldots,\h_4$ hold for each possible move applied to $R$.

We recall that in this phase only guards $g$ and $g''$ need to be moved to complete the formation of the pattern. As guards move, the embedding exploited in phase $\Ftre$ cannot be always recognized, at least not straightforwardly. Each sub-phase refers in fact to a different embedding that tries to reconstruct where the configuration comes from. 

Sub-phases $\Puno$ and $\Pdue$ manage the case where $\q$ is false. In sub-phase $\Puno$, guard $g''$ rotates along $C(R)$ in order to reach $\mu(g'')$, performing move $m_{13}$. Since $\q$ is false, we are guaranteed that $C(R)$ does not change while $g''$ moves. In fact, let $p$ be the antipodal point to $\mu(g'')$.
If $\q$ is false because of the first condition that is $\partial C(F)\neq F$, then all points in $\partial C(F)\setminus \mu(g'')$ are occupied by robots as $\mu(g)$ is inside $C(R)$.  Without loss of generality, let us assume that $g''$ needs to rotate in the clockwise direction to reach $\mu(g'')$. Since $C(F)=C(R)$, there must exist a point $f\in \partial C(F)$, which is occupied, that either coincides with $p$ or it can be met in the clockwise direction by rotating on $C(R)$ from $p$. Robot $g''$ moves in between $\mu(g'')$ and $f$, and hence it is not critical. 
Similar arguments hold if $\q$ is false because of the second condition that is $F$ contains multiplicities since $\mu(g)$ is either inside $C(R)$ or on a multiplicity. 
If $\q$ is false because of the third condition, then $\angolo(f_2,c(R),f_n)\leq 180^\circ$. If $\mu(g'')$ is reached by $g''$ in the counter-clockwise direction, then by the disposal of the points on $C(F)$ and the minimality of $f_1$, $g''$ can safely move without affecting $C(R)$. If $\mu(g'')$ is reached by $g''$ in the clockwise direction, then the condition on $\q$ assures that $f_n$ lies in the counter-clockwise direction from $g''$, it is occupied, and it is closer than $p$ to $\mu(g'')$.

\begin{description}
	\item[$\h_0$:] Only $g''$ is involved in move $m_{13}$. Even though $g''$ moves, $\iuno$ ensures to always recognize the same robot as $g''$.
	\item[$\h_1$:] As in $\Ftre$, the configuration is maintained always a leader configuration by the positioning of $g$ on $C^g(R)$.
	\item[$\h_2$:] We show that $m_{13}$ is safe. As $\q$ only depends on $F$, it is always false. During the movement of $g''$, $\iuno$ remains true hence, by Lemma~\ref{lem:F4-disjoint}, the observed configuration always belongs to $\Puno$ and it cannot belong to any other sub-phase of $\Fquattro$. The configuration cannot belong to $\Fcinque$ because either $\e$ is false or we have to show that $\bzero \vee \zuno \vee \zdue$ is false. Assuming $\e$ true, the configuration is asymmetric because of the guards, that is $\gdue$ holds. $\zuno$ is false because $\stre$ should hold and there should exist an axis of reflection $\ell$ as defined in $\zuno$. The axis $\ell$ cannot exist because of the chosen angle $\alpha$ formed by the guard $g$. Predicate $\zdue$ is false because guard $g$ is on $C^g(R)$. It follows that the configuration cannot belong to any other phase because $\fquattro$ holds and that 
$g''$ is the only moving robot.
	\item[$\h_3$:] No collisions are created as there are no robots between $g''$ and $\mu(g'')$ on $C(R)$.
	\item[$\h_4$:] Once $g''$ reaches the target, predicate $\iuno$ does not hold anymore as $g''=\mu(g'')$ but predicate $\idue$ holds, that is by Lemma~\ref{lem:F4-disjoint} the configuration is in $\Pdue$. The configuration cannot satisfy $\w$ as $g$ is on $C^g(R)$. It is not in any other phase because the same arguments as in $\h_2$.
\end{description}

From $\Puno$ only guard $g$ remains to be positioned in order to form $F$. Now the embedding of $F$ on $R$ is more difficult to detect and $g$ moves according to move $m_{14}$ in phase $\Pdue$.

\begin{description}
	\item[$\h_0$:] Only $g$ is involved in the move. Even though $g$ moves, $\idue$ ensures to always recognize the same robot as $g$ since the distance to the target always decreases.
	\item[$\h_1$:] The configuration is maintained always a leader configuration during the movement because $g$ is the only robot on $\C_\uparrow^1(R)$.  
So it cannot participate to a rotation. Moreover, the configuration cannot admit a reflection as the axis of symmetry should pass through $g$. This means that there exists another embedding of $F$ such that the target of $g$ would be closer than that it has currently calculated, but this contradicts predicate $\idue$ that chooses the target that minimizes the distance. 
	\item[$\h_2$:] We show that $m_{14}$ is safe. As $\q$ only depends on $F$, it is always false.  During the movement $\idue$ remains true, hence, by Lemma~\ref{lem:F4-disjoint}, the configuration always belongs to $\Pdue$ and it cannot belong to any other sub-phase of $\Fquattro$. The configuration cannot belong to $\Fcinque$ because either $\e$ is false and we are done, or we have to show that $\bzero \vee \zuno \vee \zdue$ is false. Assuming $\e$ true, the configuration is asymmetric because of the movement of $g$. $\zuno$ is false because $\stre$ should hold and there should exist an axis of reflection $\ell$ as defined in $\zuno$. Such an axis cannot exist as $g$ should be on it being the only robot on $\C_\uparrow^1$. Predicate $\zdue$ is false because of $\idue$. Moreover, it cannot belong to any other phase because $\fquattro$ holds. It follows that 
$g$ is the only moving robot.
	\item[$\h_3$:] No collisions are created as there are no robots between $g$ and its target since $g$ always moves inside $C_\uparrow^1(F)$ toward its border or toward $c(F)$. In any case no further robot is met as all of them are already positioned according to $F$.
	\item[$\h_4$:] Once $g$ reaches the target, predicate $\w$  holds, that is $F$ has been formed and the configuration does not belong to any phase.
\end{description}

Sub-phases $\Quno$--$\Qquattro$ manage the case where $\q$ is true. The main difficult here is to maintain $C(R)$ unchanged while guards are moving. In fact, if $g''$ rotates toward $\mu(g'')$ as in sub-phase $\Puno$, $C(R)$ could change. 

As first move, in $\Quno$ guard $g$ is moved radially on $C(R)$ by means of move $m_{15}$.

\begin{description}
	\item[$\h_0$:] Only $g$ is involved in the move. As it moves radially toward $C(R)$, angle $\alpha$ is maintained along all the movement, hence $g$ is easily recognizable.
	\item[$\h_1$:] The configuration is maintained always asymmetric as $F$ does not require multiplicities and $g$ is the only robot inside $C(R)$. 
 So it cannot participate to neither a rotation, nor a reflection as the axis of symmetry should pass through $g$, but then the starting configuration $R$ was symmetric, a contradiction.	
	\item[$\h_2$:] We show that $m_{15}$ is safe. As $\q$ only depends on $F$, it is always true. During the movement of $g$, $\itre$ remains true and hence the observed configuration always belongs to $\Quno$. By Lemma~\ref{lem:F4-disjoint}, it cannot belong to any other sub-phase of $\Fquattro$. The configuration cannot belong to $\Fcinque$ because $\e$ is true, and the configuration is asymmetric because of the movement of $g$. $\zuno$ is false because $\stre$ should hold and there should exist an axis of reflection $\ell$ as defined in $\zuno$. Such an axis cannot exist as $g$ should be on it being the only robot on $\C_\uparrow^1$. Predicate $\zdue$ is false because of the radial movement of $g$ from $\C^g(R)$. Moreover, the configuration cannot belong to any other phase because $\fquattro$ holds.

It follows that 
$g$ is the only moving robot.
	\item[$\h_3$:] No collisions are created as $g$ is the only robot inside $C(R)$ and there are no robots forming an angle of $\alpha$ degrees on $C(R)$ as they are all well positioned according to $F$ but for $g''$ that by construction is not on the way of $g$.
	\item[$\h_4$:] Once $g$ reaches the target, predicate $\itre$ becomes false while either $\iquattro$ or $\icinque$ become true, depending whether $g''$ is already on target or not. By Lemma~\ref{lem:F4-disjoint}, this means the configuration may belong to $\Qdue$ or $\Qtre$ and it cannot belong to any other sub-phase of $\Fquattro$. The configuration cannot belong to $\Fcinque$ because $\e$ is true, and the configuration is asymmetric because of $g$. $\zuno$ is false because $\stre$ is false. Predicate $\zdue$ is false because all robots are far from $\C^g(R)$. 
	Moreover, the configuration cannot belong to any other phase because $\fquattro$ holds.
\end{description}

Sub-phase $\Qdue$ is applied if $g''$ is not yet on its target. Since now all robots are in $\partial C(R)$, $g''$ cannot freely move toward $\mu(g'')$ as this could change $C(R)$. A safe place to reach is the antipodal point $p$ to $g$. Move $m_{16}$ rotates $g''$ on $C(R)$ toward the closest point among $p$ and $\mu(g'')$.

\begin{description}
	\item[$\h_0$:] Only $g''$ is involved in the move. The angle $\alpha$ between $g$ and $g'$ maintains $g''$ easily recognizable along all the movement.
	\item[$\h_1$:] The configuration is maintained always asymmetric as the angle $\alpha$ between $g$ and $g'$ guarantees no rotations. Moreover, the only axis of reflection should cut $\alpha$. Since $\q$ holds, $|F|-1$ points occupy a semi-circle. As all robots but $g$ and $g''$ are not yet positioned according to $F$, it follows that $g$ is the only robot in the semi-circle between $g'$ and $g''$ in the clockwise direction. Since by assumption there are at least four robots, this situation cannot hold as there cannot be a robot specular to one which is not a guard.
	\item[$\h_2$:] We show that $m_{16}$ is safe. As $\q$ only depends on $F$, it is always true. During the movement of $g''$, $\iquattro$ remains true. As a consequence, by Lemma~\ref{lem:F4-disjoint}, the observed configuration belongs to $\Qdue$ and it cannot belong to any other sub-phase of $\Fquattro$. The configuration cannot belong to $\Fcinque$ because $\e$ is true, and the configuration is asymmetric because of $g$. $\zuno$ is false because $\stre$ is false. Predicate $\zdue$ is false because all robots are far from $\C^g(R)$. 
	Moreover, the configuration cannot belong to any other phase because $\fquattro$ holds. 
It follows that 
$g''$ is the only moving robot.
	\item[$\h_3$:] No collisions are created as by construction the path between $g''$ and its target along $C(R)$ does not contain further robots.
	\item[$\h_4$:] Once $g''$ reaches the target, predicate $\iquattro$ becomes false while predicate $\icinque$ becomes true. By Lemma~\ref{lem:F4-disjoint}, this means the configuration may belong to $\Qtre$ and it cannot belong to any other sub-phase of $\Fquattro$. Moreover, it cannot belong to any other phase because of the same arguments as in $\h_2$.
\end{description}

In sub-phase $\Qtre$, guard $g$ can freely move toward its final target by means of move $m_{17}$. In fact, because of predicate $\q$, the move does not affect $C(R)$. 

\begin{description}
	\item[$\h_0$:] Only $g$ is involved in the move. Robot $g$ is always recognizable as the only robot of minimum view.
	\item[$\h_1$:] The configuration is maintained always asymmetric as $g$ is the only one with minimum view and its clockwise view is different from its anti-clockwise view.
	\item[$\h_2$:] We show that $m_{17}$ in phase $\Qtre$ is safe. As $\q$ only depends on $F$, it is always true. During the movement of $g$, $\icinque$ remains true. Hence, by Lemma~\ref{lem:F4-disjoint}, the observed configuration belongs to $\Qtre$ and it cannot belong to any other sub-phase of $\Fquattro$. Moreover, it cannot belong to any other phase because of the same arguments as for move $m_{13}$.
It follows that $g$ is the only moving robot.
	\item[$\h_3$:] No collisions are created as by construction the path between $g$ and its target along $C(R)$ does not contain further robots.
	\item[$\h_4$:] Once $g$ reaches the target, predicate $\icinque$ becomes false while either $\w$ or $\isei$ become true, depending whether $g''$ is already on target or not. This means either $F$ is formed or the obtained configuration $R'$ belongs to $\Qquattro$ according to Lemma~\ref{lem:F4-disjoint}. We now show that the configuration cannot belong to $\Fcinque$ since $\bzero\vee \zuno \vee \zdue$ is false, that is $\fcinque$ is false. $R'$ is asymmetric (i.e. $\bzero$ is false) because $\q$ holds and hence $R'$ cannot be rotational and does not admit multiplicities. The only possible axis of reflection in $R'$ should reflect $g$ with $g''$. In turn, this implies that there must be another robot $r$ specular to $g'$ such that $\angolo(r,c(R'),g'')=3\alpha$ with $\mu(g'')$ being on the smaller arc of $C(R')$ between $g''$ and $r$. It follows that $\angolo(r,c(R'),\mu(g''))<3\alpha$, contradicting the definition of $\alpha$. Predicate $\zuno$ is false because $\stre$ is false. Predicate $\zdue$ is false because all robots are far from $\C^g(R')$.
	The configuration cannot belong to any other phase as either $\fquattro$ or $\w$ holds.
\end{description}

In sub-phase $\Qquattro$, guard $g''$ can freely move toward its final target by means of move $m_{13}$. In fact, predicate $\q$ guarantees that the target of $g''$ does not overcome the antipodal point to $g$, hence the movement does not affect $C(R)$.

\begin{description}
	\item[$\h_0$:] Only $g''$ is involved in the move. Robot $g''$ is always recognizable as $\isei$ holds and $g''$ is the only robot not on target.
	\item[$\h_1$:] During the movement, the configuration cannot admit a rotation as the arc from $g'$ to $g''$ in the clockwise direction is greater than half of $C(R)$. There cannot be an axis of symmetry that makes $g$ specular to $g'$. In fact, $g''$ cannot admit a specular robot with respect to such an axis as it is closer to the axis than $g$ which contradicts the property of $g$ being the robot of minimum view. There cannot be an axis making specular $g$ to $g''$. In fact, $g$ cannot admit a specular robot $r$ with respect to such an axis as $\angolo(r,c(R),g'')$ should be equal to $\angolo(g,c(R),g')=3\alpha$, but this is possible only once $g''$ has reached its target being $g$ the robot of minimum view.	
	\item[$\h_2$:] We show that $m_{13}$ in phase $\Qquattro$ is safe. As $\q$ only depends on $F$, it is always true. During the movement of $g''$, $\isei$ remains true. Hence, by Lemma~\ref{lem:F4-disjoint}, the configuration belongs to $\Qtre$ and it cannot belong to any other sub-phase of $\Fquattro$. Moreover, it cannot belong to any other phase because of the same arguments as for move $m_{17}$ holds.
It follows that 
$g''$ is the only moving robot.
	\item[$\h_3$:] No collisions are created as by construction the path between $g''$ and its target along $C(R)$ does not contain further robots.
	\item[$\h_4$:] Once $g''$ reaches the target, predicate $\w$ becomes true. This means the configuration does not belong to any phase.
\end{description}
\qed
\end{proof}

Concerning phase $\Fcinque$, for each move $m$ defined in Table~\ref{tab:F5-phases} we need to show several properties (similar to those used for $\Funo, \Fdue, \Ftre$, and $\Fquattro$) that guarantee to our algorithm to safely evolve until a different phase is reached:
\begin{itemize}
	\item[$\h'_0$:] at the beginning, $m$ involves at most two robots;
	\item[$\h'_1$:] while robots are moving according to $m$, the configuration remains a leader configuration;
	\item[$\h'_2$:] $m$ is safe, and in particular that while robots are moving according to $m$, all other robots are stationary;
	\item[$\h'_3$:] while robots are moving according to $m$, no collisions are created;
	\item[$\h'_4$:] if $m$ is associated to any phase $\mathcal{X}$, then the predicate $\mathcal{X}_e$ holds once the robots have terminated to apply $m$; 
	\item[$\h'_5$:] $m$ preserves stationarity.
\end{itemize}

Note that, in this case we cannot get rid of property $\h_5'$ as we did with $\h_5$ for phases $\Funo$, $\Fdue$, $\Ftre$, $\Fquattro$ as now there might be two robots moving concurrently.

Moreover, during phase $\Fcinque$ it is possible that by chance a robot makes true predicate $\w$ while moving. For the sake of clarity, and in order to not overcharge Procedure $\SMove$, we prefer to ignore such occurrences. This simply implies that the pattern will be formed successively by following the designed transitions, and from there on robots will not move anymore. Differently, if a stationary configuration is reached where $\w$ holds, then no move will be performed anymore. As described in Table~\ref{tab:F5-phases}, all sub-phases of $\Fcinque$ with the exception of $\W$ may lead to stationary configurations where $\w$ holds. In the subsequent analysis we omit to mention all the times such possibilities.

\begin{lemma}\label{lem:corr-F5}
Let $R$ be a stationary configuration in $\Fcinque$. From $R$ the algorithm eventually leads to a stationary configuration belonging to $\Funo$, $\Fdue$, $\Fquattro$ or where $\w$ holds. 
\end{lemma}
\begin{proof}
By Lemma~\ref{lem:F5-disjoint}, exactly one of the predicates for the starting phases in Table~\ref{tab:F5-phases} is true. Then, exactly one of the moves associated to the sub-phases of $\Fcinque$ is applied to $R$. We show that the properties $\h'_0,\ldots,\h'_5$ hold for each possible move applied to $R$.

Let us consider sub-phase $\T$ where move $m_{18}$ is performed.
\begin{description}
	\item[$\h'_0$:] As $\buno$ holds, there is only one robot $r$ on the axis $\ell$ moved by $m_{18}$.
	\item[$\h'_1$:] If $r$ admits a rotational-free path, then move $m_{18}$ makes $r$ move toward $c(R\setminus \{r\})$. Since $r$ moves along the axis $\ell$, then the obtained configuration remains a leader configuration. Moreover, $r$ is always recognizable as the moving robot. 
	 
If $r$ does not admit a rotational-free path, let $p'$ be the point on $\ell$ toward $c(R\setminus \{r\})$ such that, if reached by $r$, the new configuration $R'$ becomes rotational. In this case $r$ is moved to become a faraway robot, that is ‘sufficiently' far from $c(R)$. We now show that during the movement no configuration admitting a rotation can be created. By contradiction, let us assume that that there exists a point $p''$ such that, if reached by $r$ the new configuration $R''$ admits a rotation. If in $p''$ the robot $r$ is not critical, then $R''$ has the same center of $R$ and cannot admit a rotation as in $p'$ there is not a robot. So, in $p''$ robot $r$ must be critical. In this case, the center is not maintained and, if $R''$ is rotational, the robots on $C(R'')$ must form a regular triangle (as in any other $n$-gon the robots are not critical). Let $r_1$ and $r_2$ be the other two robots on $C(R'')$. As they are at the same distance from $\ell$, they are also on a same circle $C^i_\uparrow(R')$. If $r_1$ and $r_2$ are antipodal on $C^i_\uparrow(R')$ then to make $R'$ rotational, $r$ should be in $c(R')$, but it is not possible by hypothesis. If $r_1$ and $r_2$ are not antipodal on $C^i_\uparrow(R')$, then there are $k$ equidistant robots on $C^i_\uparrow(R')$, with $k$ odd, as the robots on $C^i_\uparrow(R')$ consist of $r$ and pairs of symmetric robots with respect to $\ell$. Moreover, $k > 3$ as the case $k=3$ is realized when $r$ is in $p''$. This implies that there should be robots on $C^i_\uparrow(R')$ in the arc between $r_1$ and $r_2$ external to $C(R'')$, a contradiction to the definition of $C(R'')$. Hence, also in this case the obtained configuration remains a leader configuration.
	\item[$\h'_2$:] Move $m_{18}$ is safe as $\bzero$ and $\e$ are both true during the movement and then the  configuration cannot be in $\Funo$, $\Fdue$, $\Ftre$, $\Fquattro$. Moreover, as $\buno \wedge \neg \zuno$ remains true, the configuration cannot be in the other sub-phases of $\Fcinque$ and $r$ remains the only moving robot.
	\item[$\h'_3$:] No collisions are created as $r$ moves on $\ell$ and it is the only robot on $\ell$.
	\item[$\h'_4$:] As $\bzero \wedge \e$ is still true the configuration remains in $\Fcinque$. If $r$ reaches $c(R)$, the configuration is in $\W$ as $\zdue$ is clearly true, and both $\Vuno_s$ and $\Vdue_s$ are false as $\bdue$ and $\btre$ are false. If $r$ becomes faraway, then $\zuno$ becomes true and the configuration is in $\U$. 
	\item[$\h'_5$:] As $m_{18}$ is safe and a single robot is moving, then when $r$ stops the configuration is stationary.
\end{description}

Let us consider sub-phase $\U$ where move $m_{19}$ is performed.
\begin{description}
	\item[$\h'_0$:] There is only one faraway robot $r$ moved by $m_{19}$.
	\item[$\h'_1$:] While $r$ is moving according to $m_{19}$, the configuration remains asymmetric as the three robots on $C(R)$ do not form a regular triangle.
	\item[$\h'_2$:] While $r$ is moving, predicate $\fcinque$  remains true as both $\zuno$ and $\e$ maintain their values. Then the configuration remains in $\Fcinque$. As $\bzero$ is false and $\zuno \Rightarrow \neg \zdue$ (cf the proof of Lemma~\ref{lem:F5-disjoint}), the configuration cannot belong to any other sub-phase of $\Fcinque$. Finally, robot $r$ remains always recognizable during the movement. 
	
	\item[$\h'_3$:] There are only three robots on $C(R)$ and $r$ is moving on $C(R)$ within an arc without other robots.
	\item[$\h'_4$:] When $r$ reaches its final position on $t$, predicate $\zuno$ is false and so are $\bzero$ and $\zdue$. Then the configuration is no more in $\Fcinque$ as $\fcinque$ is false. The configuration can verify $\Funo_s$. In fact $\Fdue_s$ is false as $\sdue$ is false, and $\Ftre_s$ is not verified as $\gdue$ is false being $\gzero$ false. Predicate $\gzero$ is false because on $C^1_{\uparrow}(R)$ there are at least two robots since $R\setminus \{r\}$ is symmetric. Moreover the configuration is not in $\Fquattro$ as $\q$ is false and neither $\iuno$ or $\idue$ are verified, as $\guno$ is false and $\partial C^1_{\uparrow}(R)$ contains at least two robots.
	\item[$\h'_5$:] As $m_{19}$ is safe and a single robot is moving, then when the robot $r$ stops the configuration is stationary.
\end{description}

Let us consider sub-phase $\Vuno$ where move $m_{20}$ is performed.
\begin{description}
	\item[$\h'_0$:] 
    Move $m_{20}$ moves 	$\robotuno$; in some cases, it moves both $\robotuno$ and $\robotdue$ concurrently. Hence, at the beginning, the move involves at most two robots.
	\item[$\h'_1$:] It is worth to note that $\robotuno$ cannot move only if it does not admit a rotational-free path. Such a situation can happen only if the robots on the circle $C^i_\uparrow(R)$ where $\robotdue$ lies form a configuration admitting a rotation. For this reason $\robotdue$ is moved so as to ensure that it is the only robot on the circle where it lies. This implies that no rotational symmetries are created and hence the configurations remains a leader configuration. Moreover, robots $\robotuno$ and $\robotdue$ are always recognizable as such since they move without changing their order on $\ell$.
	\item[$\h'_2$:] Move $m_{20}$ is safe as $\bzero$ and $\e$ are both true during the movement, that is, $\fcinque$ holds and the configuration cannot belong to $\Funo$, $\Fdue$, $\Ftre$, $\Fquattro$. Moreover, as $\bdue$ remains true and at least one among $\robotuno$ and $\robotdue$ is not on its target, then, by Lemma~\ref{lem:F5-disjoint}, the configuration cannot be in the other sub-phases of $\Fcinque$ as $\neg \c \vee \neg \uuno$ holds. Moreover, as already observed in the previous item, robots $\robotuno$ and $\robotdue$ are always recognizable.
	\item[$\h'_3$:] No collisions are created as $\robotuno$ is the robot closest to $c(R)$, its target, and $\robotdue$, after starting to move, always remains the only robot between $C_{\uparrow}^{i-1}(R)$ and $C_{\uparrow}^{i+1}(R)$.
	\item[$\h'_4$:] Once both $\robotuno$ and $\robotdue$ reach their target, still $\bzero \wedge \e$ is true hence the configuration remains in $\Fcinque$, $\zdue$ becomes true but $\neg \c \vee \neg \uuno$ is false, that is $\Vuno_s$ is false. Also $\Vdue_s$ is false as $\bdue$ remains true, and hence $\W_s$ holds. 
	\item[$\h'_5$:] As $m_{20}$ is safe with $\robotuno$ and $\robotdue$ being the moving robots, when they both stop the configuration is stationary.
\end{description}

Let us consider sub-phase $\Vdue$ where move $m_{21}$ is performed.
\begin{description}
	\item[$\h'_0$:] As $\btre$ holds, robots $\robotuno$ and $\robotdue$ are the only robots on the axis $\ell$, exactly the robots moved by means of $m_{21}$.
	\item[$\h'_1$:] The only case in which $\robotuno$ does not admit a rotational-free path occurs when $\robotdue$ is on $t^{60}$ (cf Figure~\ref{fig:fase_F5}.(e)). In fact, this is the only case in which $C(R)$ can admit a rotation, with $\robotdue$ being critical by definition. Of course, the rotational configuration must be avoided as it is not a leader configuration. The only case in which $\robotdue$ does not move from $t^{60}$ is when $t^x=t^{60}$. However, if the current configuration $R$ admits $|\partial C(R)|>4$, then $\robotuno$ can freely move toward $c(R)$ whereas $\robotdue$ is already on target. If $|\partial C(R)|=4$, then $\robotuno \in \partial C(R)$ and moves toward its intermediate target $t$. In so doing, $t^x$ changes and $\robotdue$ can leave $t^{60}$. From now on, $\robotuno$ admits a rotational-free path. Note that, while robots are moving, both their targets change. In particular, $t^x$ can get further from $t^{60}$ or even disappearing. $\delta C(R)$ grows and the target $c(R)$ for $\robotuno$ moves further from $\robotuno$ according to the movement of $\robotdue$. Eventually, $\robotdue$ either reaches $t^x$ that does not change anymore or it reaches $t^{55}$. Hence, also $\robotuno$ can eventually reach its final destination. Hence, during the movements the obtained configuration is always a leader configuration.
	\item[$\h'_2$:] Move $m_{21}$ is safe as $\bzero$ and $\e$ are both true during the movement, that is, $\fcinque$ holds and the configuration cannot belong to $\Funo$, $\Fdue$, $\Ftre$, $\Fquattro$. Moreover, as $\btre$ remains true and at least one among $\robotuno$ and $\robotdue$ is not on its target, then, by Lemma~\ref{lem:F5-disjoint}, the configuration cannot be in the other sub-phases of $\Fcinque$ as $\neg \c \vee \neg \udue$ holds. Moreover, robots $\robotuno$ and $\robotdue$ are always recognizable.
	\item[$\h'_3$:] No collisions are created as there are no further robots on $\ell$, whereas $\robotuno$ and $\robotdue$ cannot meet.
	\item[$\h'_4$:] Once both $\robotuno$ and $\robotdue$ reach their targets, still $\bzero \wedge \e$ is true hence the configuration remains in $\Fcinque$, $\zdue$ becomes true but $\neg \c \vee \neg \udue$ is false, that is $\Vdue_s$ is false. Also $\Vuno_s$ is false as $\btre$ remains true, and hence $\W_s$ holds. 
	\item[$\h'_5$:] As $m_{21}$ is safe with $\robotuno$ and $\robotdue$ being the moving robots, when they both stop the configuration is stationary.
\end{description}

Let us consider sub-phase $\Vtre$ where move $m_{22}$ is performed.
\begin{description}
	\item[$\h'_0$:] As $\bquattro$ holds, only $\robotuno$ moves on the axis $\ell$ by means of $m_{22}$.
	\item[$\h'_1$:] While $r$ moves toward $c(R)$, it remains not critical and it is always detected as the robot to move (cf Figure~\ref{fig:fase_F5}.(f)). Moreover the configuration remains a leader configuration as it cannot become rotational: in fact, the critical robots are the only robots on $C(R)$, then the only possible rotation $\varphi$ is a $180^\circ$ rotation, but in this case $\varphi(\robotuno)$ would not exists.
	\item[$\h'_2$:] Move $m_{22}$ is safe as $\bzero$ and $\e$ are both true during the movement, that is, $\fcinque$ holds and the configuration cannot belong to $\Funo$, $\Fdue$, $\Ftre$, $\Fquattro$. Moreover, as $\neg \c \wedge \bquattro$ remains true, the configuration cannot be in the other sub-phases of $\Fcinque$ and $\robotuno$ remains the only moving robot.
	\item[$\h'_3$:] No collisions are created as $\robotuno$ admits a rotational-free path.
	\item[$\h'_4$:] Once $r$ reaches its target, still $\bzero \wedge \e$ is true hence the configuration remains in $\Fcinque$. Once $r$ is in $c(R)$, $\zdue$ holds. Since $\bquattro$ remains true, neither $\Vuno_s$ nor $\Vdue_s$ can become true, then $\W_s$ holds. 
	\item[$\h'_5$:] As $m_{22}$ is safe and a single robot is moving, then when $\robotuno$ stops the configuration is stationary.
\end{description}

Let us consider sub-phase $\Vquattro$ where move $m_{23}$ is performed.
\begin{description}
	\item[$\h'_0$:] As $\bcinque$ holds, there is only one robot $r$ on the axis $\ell$ moved by $m_{23}$.
	\item[$\h'_1$:] While $r$ moves, it is always detected as the robot on $\ell$ closest to $C(R\setminus \{\robotuno,\robotdue\})$  (cf Figure~\ref{fig:fase_F5}.(g)). Then, no rotational symmetry can be created and hence the configuration remains a leader configuration.
	\item[$\h'_2$:] Move $m_{23}$ is safe as $\bzero$ and $\e$ are both true during the movement, that is, $\fcinque$ holds and the configuration cannot belong to $\Funo$, $\Fdue$, $\Ftre$, $\Fquattro$. Moreover, as $\bcinque \wedge \neg \c$ remains true, the configuration cannot be in the other sub-phases of $\Fcinque$ and $r$ remains the only moving robot.
	\item[$\h'_3$:] No collisions are created as $r$ moves on $\ell$ and the only other robot on $\ell$ is antipodal.
	\item[$\h'_4$:] Once $r$ reaches its target, still $\bzero \wedge \e$ is true hence the configuration remains in $\Fcinque$. Now $r$ is not critical, hence $\btre$ holds and sub-phase $\Vdue$ is reached. 
	\item[$\h'_5$:] As $m_{23}$ is safe and a single robot is moving, then when the robot $r$ stops the configuration is stationary.
\end{description}

Let us consider sub-phase $\W$ where move $m_{24}$ is performed.
\begin{description}
	\item[$\h'_0$:] There is only one robot $r$ moved by $m_{24}$, the one  identified as the closest to $c(R)$.
	\item[$\h'_1$:] While $r$ is moving according to $m_{24}$, the configuration remains asymmetric as $R'=R\setminus \{r\}$ is symmetric, being  $\zdue$ true, and $r$ is not moving along an axis of symmetry. 
	\item[$\h'_2$:] While $r$ is moving, either $\idue$ becomes true or $\zdue \wedge e$ holds and hence $\fcinque$ remains true. In the first case, the configuration may belong to $\Fquattro$, and in particular to sub-phase $\Pdue$. However, from $\Pdue$, by move $m_{14}$, only the same robot $r$ is allowed to move. In the second case, the configuration remains in $\Fcinque$. As $\bzero$ is false and $\zdue \Rightarrow \neg \zuno$, the configuration cannot belong to any other sub-phase of $\Fcinque$.
	
	\item[$\h'_3$:] By definition of $C^g(R)$ there are no robots between $r$ and its target.
	\item[$\h'_4$:] When $r$ reaches its final position on $C^g(R)$, predicate $\zdue$ is false and so are $\bzero$ and $\zuno$. Then the configuration is no more in $\Fcinque$ as $\fcinque$ is false. The obtained configuration can verify $\Funo_s$, $\Fdue_s$ or $\Fquattro_s$. It cannot verify $\Ftre_s$ as $\gdue$ is false being $\guno$ false. Predicate $\guno$ is false since $r$ does not form the reference angle of $\alpha$ degree. Predicate $\w$ cannot hold as $r$ is on $C^g(R)$. 
	\item[$\h'_5$:] As $m_{24}$ is safe and a single robot is moving, then when the robot $r$ stops the configuration is stationary.
\end{description}
\qed
\end{proof}

\begin{theorem}\label{th:correctness}
Let $R$ be an initial leader configuration of \async robots without chirality, and $F$ any pattern (possibly with multiplicities) with $|F|=|R|$. Then, there exists an algorithm able to form $F$ from $R$.
\end{theorem}

\begin{proof}
As remarked in Section~\ref{sec:algorithm}, the cases of $|R|\leq 2$ are either trivial or unsolvable, and hence are not required to be managed by our algorithm. The case of $F$ being a single point is delegated to~\cite{CFPS12}, whereas when $|R|= 3$ Theorem~\ref{teo:3robots} holds.
When $|R|>3$, the claim simply follows by Lemmata~\ref{lem:phases-disjoint},~\ref{lem:corr-F1},~\ref{lem:corr-F2},~\ref{lem:corr-F3},~\ref{lem:corr-F4}, and~\ref{lem:corr-F5}. In fact, Lemma~\ref{lem:phases-disjoint} shows that $R$ belongs to exactly one phase among $\Funo$, $\Fdue$, $\Ftre$, $\Fquattro$, and $\Fcinque$. Lemmata~\ref{lem:corr-F1}--\ref{lem:corr-F4} show that from a given phase among $\Funo, \cdots, \Fquattro$ only subsequent phases can be reached, or $\w$ eventually holds (cf Table~\ref{tab:phases}, first and last column). Lemma~\ref{lem:corr-F5} instead shows that from $\Fcinque$ any other phase but $\Ftre$ can be reached, or $\w$ eventually holds. Cycles among phases can occur only between $\Funo$ and $\Fcinque$. However, the involved moves $m_4$ and $m_{19}$ guarantee that the transitions of such a cycle can be traversed only once. 

Inside each phase among $\Funo$, $\cdots$, $\Fcinque$, the only possible cycles among transitions can occur in phase $\Funo$ among sub-phases $\Auno$ and $\Adue$, or in  $\Ftre$ among sub-phases $\M$, $\N$ and $\O$. However, the corresponding Lemmata~\ref{lem:corr-F1} and~\ref{lem:corr-F3} also show that such cycles can be performed only a finite number of times.
\qed
\end{proof}
 
 From~\cite{FPSW08}, it is possible to state the next theorem. For the sake of completeness we also remind the corresponding easy proof. 

\begin{theorem}~\cite{FPSW08}\label{th:leader}
Let $R$ be an initial configuration of $n$ \async robots without chirality. If \apf is solvable, then also Leader Election can be solved.
\end{theorem}

\begin{proof}
Consider the pattern $F$ with $n-1$ collinear points evenly placed such that the distance between two adjacent points is some $d$, and one further point on the same line at distance $2d$ from its unique neighbor.
If robots can form $F$, then the unique robot with one neighbor at distance $2d$ can be elected as the leader.
\qed
\end{proof}

From Theorem~\ref{th:correctness} and Theorem~\ref{th:leader} we obtain Theorem~\ref{th:main1} that we now rephrase in a more explicit form:

\begin{theorem}\label{th:main}
Let $R$ be an initial configuration of \async robots without chirality. There exists a deterministic transition-safe algorithm that solves \apf from $R$ if and only the Leader Election problem can be solved in $R$, that is, $R$ is a leader configuration.
\end{theorem}

\section{Conclusion}\label{sec:conclusion}
We considered the Arbitrary Pattern Formation problem in the well-known asynchronous Look-Compute-Move model. So far the problem has been mainly investigated with the further assumption on the capability of robots to share a common left-right orientation (chirality). 

Our study removes any assumption concerning the orientation of the robots. We shown that starting from initial leader configurations, robots can deterministically form any pattern including symmetric ones and those containing multiplicities.
This extends the previously known results in terms of required number of robots, orientation capabilities, multiplicities,  and formalisms. In fact, our algorithm does not rely on any assumption on the number of robots, nor on any orientation of the robots, it allows the formation of multiplicities if required by the given pattern, and it is provided in terms of logical predicates that facilitate to check its correctness.

Also, the relevance of our results is shown in light of the consequences obtained with respect to~\cite{FYOKY15},~\cite{BT15,BT16b}, and~\cite{DPVarx09}. 
Recently, errata to~\cite{FYOKY15} and~\cite{BT15,BT16b} have been delivered but the level of formalism and details with respect to our approach is still something remarkable.

The main open question left asks to provide a deterministic algorithm that solves the Pattern Formation rather than the Arbitrary Pattern Formation from any initial configuration, including symmetric ones. Potentially, robots should be able to form any pattern if starting from configurations characterized by symmetries that are included in the final pattern. The main difficulty in designing an algorithm for such cases is that in symmetric configurations many robots may move simultaneously, all those that look equivalent with respect to the symmetry. The adversary can decide to move any subset of such robots, and all of them may traverse different distances as the adversary can stop them in different moments. Hence, during a Look phase, it becomes very difficult to provide a mean to guess how the current configuration has been originated.


\begin{thebibliography}{10}
\providecommand{\url}[1]{{#1}}
\providecommand{\urlprefix}{URL }
\expandafter\ifx\csname urlstyle\endcsname\relax
  \providecommand{\doi}[1]{DOI~\discretionary{}{}{}#1}\else
  \providecommand{\doi}{DOI~\discretionary{}{}{}\begingroup
  \urlstyle{rm}\Url}\fi

\bibitem{BLP16}
B{\'{e}}rard, B., Lafourcade, P., Millet, L., Potop{-}Butucaru, M.,
  Thierry{-}Mieg, Y., Tixeuil, S.: Formal verification of mobile robot
  protocols.
\newblock Distributed Computing \textbf{29}(6), 459--487 (2016)

\bibitem{BCM16}
Bhagat, S., Chaudhuri, S.G., Mukhopadhyaya, K.: Formation of general position
  by asynchronous mobile robots under one-axis agreement.
\newblock In: Proc. 10th Int.'l WS on Algorithms and Computation {(WALCOM)},
  \emph{LNCS}, vol. 9627, pp. 80--91. Springer (2016)

\bibitem{BPT16}
Bonnet, F., Potop{-}Butucaru, M., Tixeuil, S.: Asynchronous gathering in rings
  with 4 robots.
\newblock In: Proc. 5th Int.'l Conf. on Ad-hoc, Mobile, and Wireless Networks
  {(ADHOC-NOW)}, \emph{LNCS}, vol. 9724, pp. 311--324. Springer (2016)

\bibitem{BT15}
Bramas, Q., Tixeuil, S.: {Brief Announcement}: {P}robabilistic asynchronous
  arbitrary pattern formation.
\newblock In: Proc. 35th {ACM SIGACT-SIGOPS} Symp. on Principles of Distributed
  Computing (PODC) (2016)

\bibitem{BT16b}
Bramas, Q., Tixeuil, S.: {Brief Announcement}: {P}robabilistic asynchronous
  arbitrary pattern formation.
\newblock In: Proc. 18th Int.'l Symp. on Stabilization, Safety, and Security of
  Distributed Systems (SSS), \emph{Lecture Notes in Computer Science}, vol.
  10083, pp. 88--93 (2016)

\bibitem{BT16}
Bramas, Q., Tixeuil, S.: {P}robabilistic asynchronous arbitrary pattern
  formation.
\newblock CoRR \textbf{abs/1508.03714} (2016).
\newblock \urlprefix\url{https://arxiv.org/abs/1508.03714}

\bibitem{CGJM14}
Chaudhuri, S.G., Ghike, S., Jain, S., Mukhopadhyaya, K.: Pattern formation for
  asynchronous robots without agreement in chirality.
\newblock CoRR \textbf{abs/1403.2625} (2014).
\newblock \urlprefix\url{http://arxiv.org/abs/1403.2625}

\bibitem{CDN16}
Cicerone, S., {Di Stefano}, G., Navarra, A.: Asynchronous embedded pattern
  formation without orientation.
\newblock In: Proc. 30th Int.'l Symp. on Distributed Computing (DISC),
  \emph{LNCS}, vol. 9888, pp. 85--98. Springer (2016)

\bibitem{CFPS12}
Cieliebak, M., Flocchini, P., Prencipe, G., Santoro, N.: Distributed computing
  by mobile robots: Gathering.
\newblock SIAM J. on Computing \textbf{41}(4), 829--879 (2012)

\bibitem{CieliebakP02}
Cieliebak, M., Prencipe, G.: Gathering autonomous mobile robots.
\newblock In: Proceedings of the 9th International Colloquium on Structural
  Information and Communication Complexity (SIROCCO), vol.~13, pp. 57--72.
  Carleton Scientific (2002)

\bibitem{DDKN12}
{D'Angelo}, G., {Di Stefano}, G., Klasing, R., Navarra, A.: Gathering of robots
  on anonymous grids and trees without multiplicity detection.
\newblock Theor. Comput. Sci. \textbf{610}, 158--168 (2016)

\bibitem{DDN12tr}
{D'Angelo}, G., {Di Stefano}, G., Navarra, A.: Gathering on rings under the
  look-compute-move model.
\newblock Distributed Computing \textbf{27}(4), 255--285 (2014)

\bibitem{DDN11}
D'Angelo, G., Stefano, G.D., Navarra, A.: Gathering six oblivious robots on
  anonymous symmetric rings.
\newblock J. Discrete Algorithms \textbf{26}, 16--27 (2014)

\bibitem{DFSY15}
Das, S., Flocchini, P., Santoro, N., Yamashita, M.: Forming sequences of
  geometric patterns with oblivious mobile robots.
\newblock Distributed Computing \textbf{28}(2), 131--145 (2015)

\bibitem{DMN15}
{Di Stefano}, G., Montanari, P., Navarra, A.: About ungatherability of
  oblivious and asynchronous robots on anonymous rings.
\newblock In: Proc. 26th Int.'l WS on Combinatorial Algorithms {(IWOCA)},
  \emph{LNCS}, vol. 9538, pp. 136--147. Springer (2016)

\bibitem{DPVarx09}
Dieudonn{\'{e}}, Y., Petit, F., Villain, V.: Leader election problem versus
  pattern formation problem.
\newblock CoRR \textbf{abs/0902.2851} (2009).
\newblock \urlprefix\url{http://arxiv.org/abs/0902.2851}

\bibitem{DPV10-2}
Dieudonn{\'{e}}, Y., Petit, F., Villain, V.: Brief announcement: leader
  election vs pattern formation.
\newblock In: Proc. 29th Annual {ACM} Symposium on Principles of Distributed
  Computing {(PODC)}, pp. 404--405. {ACM} (2010)

\bibitem{DPV10}
Dieudonn{\'{e}}, Y., Petit, F., Villain, V.: Leader election problem versus
  pattern formation problem.
\newblock In: Proc. 24th Int.'l Symp. on Distributed Computing (DISC),
  \emph{LNCS}, vol. 6343, pp. 267--281. Springer (2010)

\bibitem{DBO17}
Doan, H.T.T., Bonnet, F., Ogata, K.: Model checking of a mobile robots
  perpetual exploration algorithm.
\newblock In: Proc. 6th Int.'l Work. on Structured Object-Oriented Formal
  Language and Method ({SOFL+MSVL}), \emph{Lecture Notes in Computer Science},
  vol. 10189, pp. 201--219 (2017)

\bibitem{FPSV14}
Flocchini, P., Prencipe, G., Santoro, N., Viglietta, G.: Distributed computing
  by mobile robots: uniform circle formation.
\newblock Distributed Computing \textbf{30}, 413--457 (2017)

\bibitem{FPSW08}
Flocchini, P., Prencipe, G., Santoro, N., Widmayer, P.: Arbitrary pattern
  formation by asynchronous, anonymous, oblivious robots.
\newblock Theor. Comput. Sci. \textbf{407}(1-3), 412--447 (2008)

\bibitem{FYOKY15}
Fujinaga, N., Yamauchi, Y., Ono, H., Kijima, S., Yamashita, M.: Pattern
  formation by oblivious asynchronous mobile robots.
\newblock {SIAM} J. Computing \textbf{44}(3), 740--785 (2015)

\bibitem{FYOKY17}
Fujinaga, N., Yamauchi, Y., Ono, H., Kijima, S., Yamashita, M.: Erratum:
  Pattern formation by oblivious asynchronous mobile robots (2017).
\newblock
  \urlprefix\url{http://tcslab.csce.kyushu-u.ac.jp/$\sim$kijima/papers/ErratumFujinagaSICOMP15v2.pdf}

\bibitem{GM10}
Ghike, S., Mukhopadhyaya, K.: A distributed algorithm for pattern formation by
  autonomous robots, with no agreement on coordinate compass.
\newblock In: Proc. 6th Int.'l Conf. on Distributed Computing and Internet
  Technology, ({ICDCIT}), \emph{LNCS}, vol. 5966, pp. 157--169. Springer (2010)

\bibitem{MV16}
Mamino, M., Viglietta, G.: Square formation by asynchronous oblivious robots.
\newblock In: Proc. of the 28th Canadian Conference on Computational Geometry
  {(CCCG)}, pp. 1--6 (2016)

\bibitem{M83}
Megiddo, N.: Linear-time algorithms for linear programming in
  {R}\({}^{\mbox{3}}\) and related problems.
\newblock {SIAM} J. Comput. \textbf{12}(4), 759--776 (1983)

\bibitem{MPST14}
Millet, L., Potop{-}Butucaru, M., Sznajder, N., Tixeuil, S.: On the synthesis
  of mobile robots algorithms: The case of ring gathering.
\newblock In: Proc. 16th Int.'l Symp. on Stabilization, Safety, and Security of
  Distributed Systems {(SSS)}, \emph{LNCS}, vol. 8756, pp. 237--251. Springer
  (2014)

\bibitem{SY99}
Suzuki, I., Yamashita, M.: Distributed anonymous mobile robots: Formation of
  geometric patterns.
\newblock {SIAM} J. Comput. \textbf{28}(4), 1347--1363 (1999)

\bibitem{Welz91}
Welzl, E.: Smallest enclosing disks (balls and ellipsoids).
\newblock In: Results and New Trends in Computer Science, pp. 359--370.
  Springer-Verlag (1991)

\bibitem{YS10}
Yamashita, M., Suzuki, I.: Characterizing geometric patterns formable by
  oblivious anonymous mobile robots.
\newblock Theor. Comput. Sci. \textbf{411}(26-28), 2433--2453 (2010)

\bibitem{YUKY17}
Yamauchi, Y., Uehara, T., Kijima, S., Yamashita, M.: Plane formation by
  synchronous mobile robots in the three-dimensional euclidean space.
\newblock J. {ACM} \textbf{64}(3), 16:1--16:43 (2017)

\bibitem{YY14}
Yamauchi, Y., Yamashita, M.: Randomized pattern formation algorithm for
  asynchronous oblivious mobile robots.
\newblock In: Proc. 28th Int.'l Symp. on Distributed Computing, {(DISC)},
  \emph{LNCS}, vol. 8784, pp. 137--151. Springer (2014)

\end{thebibliography}

\end{document}